\documentclass[onecolumn]{IEEEtran}
\IEEEoverridecommandlockouts
\usepackage{soul}
\usepackage[mathscr]{eucal}
\usepackage[cmex10]{amsmath}
\usepackage{epsfig,epsf,psfrag}
\usepackage{amssymb,amsmath,amsthm,amsfonts,latexsym}
\usepackage{amsmath,graphicx,bm,xcolor,url,overpic}
\usepackage[caption=false]{subfig}
\usepackage{fixltx2e}
\usepackage{array}
\usepackage{verbatim}
\usepackage{bm}
\usepackage{algpseudocode}
\usepackage{algorithm}
\usepackage{verbatim}
\usepackage{textcomp}
\usepackage{mathrsfs}
\usepackage{epstopdf}

\newcommand{\openone}{\leavevmode\hbox{\small1\normalsize\kern-.33em1}}

\catcode`~=11 \def\UrlSpecials{\do\~{\kern -.15em\lower .7ex\hbox{~}\kern .04em}} \catcode`~=13

\allowdisplaybreaks[4]

\newcommand{\nn}{\nonumber}

\newcommand{\calA}{\mathcal{A}}
\newcommand{\calB}{\mathcal{B}}
\newcommand{\calC}{\mathcal{C}}
\newcommand{\calD}{\mathcal{D}}

\newcommand{\calF}{\mathcal{F}}

\newcommand{\calH}{\mathcal{H}}

\newcommand{\calM}{\mathcal{M}}

\newcommand{\calP}{\mathcal{P}}

\newcommand{\calR}{\mathcal{R}}
\newcommand{\calS}{\mathcal{S}}
\newcommand{\calT}{\mathcal{T}}

\newcommand{\calW}{\mathcal{W}}
\newcommand{\calX}{\mathcal{X}}

\newcommand{\bP}{\mathbf{P}}

\newcommand{\bQ}{\mathbf{Q}}

\newcommand{\bx}{\mathbf{x}}
\newcommand{\bX}{\mathbf{X}}

\newcommand{\rmA}{\mathrm{A}}

\newcommand{\rmG}{\mathrm{G}}

\newcommand{\rmH}{\mathrm{H}}

\newcommand{\rmL}{\mathrm{L}}

\newcommand{\rmN}{\mathrm{N}}

\newcommand{\rmr}{\mathrm{r}}

\newcommand{\rmS}{\mathrm{S}}

\newcommand{\bbN}{\mathbb{N}}

\newcommand{\bbP}{\mathbb{P}}

\newcommand{\bbR}{\mathbb{R}}

\DeclareMathAlphabet{\mathbsf}{OT1}{cmss}{bx}{n}
\DeclareMathAlphabet{\mathssf}{OT1}{cmss}{m}{sl}

\DeclareSymbolFont{bsfletters}{OT1}{cmss}{bx}{n}
\DeclareSymbolFont{ssfletters}{OT1}{cmss}{m}{n}
\DeclareMathSymbol{\bsfGamma}{0}{bsfletters}{'000}
\DeclareMathSymbol{\ssfGamma}{0}{ssfletters}{'000}
\DeclareMathSymbol{\bsfDelta}{0}{bsfletters}{'001}
\DeclareMathSymbol{\ssfDelta}{0}{ssfletters}{'001}
\DeclareMathSymbol{\bsfTheta}{0}{bsfletters}{'002}
\DeclareMathSymbol{\ssfTheta}{0}{ssfletters}{'002}
\DeclareMathSymbol{\bsfLambda}{0}{bsfletters}{'003}
\DeclareMathSymbol{\ssfLambda}{0}{ssfletters}{'003}
\DeclareMathSymbol{\bsfXi}{0}{bsfletters}{'004}
\DeclareMathSymbol{\ssfXi}{0}{ssfletters}{'004}
\DeclareMathSymbol{\bsfPi}{0}{bsfletters}{'005}
\DeclareMathSymbol{\ssfPi}{0}{ssfletters}{'005}
\DeclareMathSymbol{\bsfSigma}{0}{bsfletters}{'006}
\DeclareMathSymbol{\ssfSigma}{0}{ssfletters}{'006}
\DeclareMathSymbol{\bsfUpsilon}{0}{bsfletters}{'007}
\DeclareMathSymbol{\ssfUpsilon}{0}{ssfletters}{'007}
\DeclareMathSymbol{\bsfPhi}{0}{bsfletters}{'010}
\DeclareMathSymbol{\ssfPhi}{0}{ssfletters}{'010}
\DeclareMathSymbol{\bsfPsi}{0}{bsfletters}{'011}
\DeclareMathSymbol{\ssfPsi}{0}{ssfletters}{'011}
\DeclareMathSymbol{\bsfOmega}{0}{bsfletters}{'012}
\DeclareMathSymbol{\ssfOmega}{0}{ssfletters}{'012}

\newcommand{\tilk}{\tilde{k}}

\newcommand{\tilP}{\tilde{P}}

\newcommand{\hatT}{\hat{T}}
\newcommand{\tilt}{\tilde{t}}

\newcommand{\bari}{\bar{i}}

\DeclareMathOperator*{\argmin}{arg\,min}

\newtheorem{theorem}{Theorem}
\newtheorem{lemma}[theorem]{Lemma}

\newtheorem{definition}{Definition}

\graphicspath{{./figures/}}
\usepackage{graphicx,cite}
\usepackage{epstopdf}
\usepackage{enumerate}
\usepackage{color}
\usepackage{bbm}
\usepackage{booktabs}
\usepackage{makecell}
\usepackage[ colorlinks = true,
linkcolor = blue,
urlcolor  = blue,
citecolor = red,
anchorcolor = green,]{hyperref}
\usepackage{multirow}
\usepackage{cellspace}
\usepackage{diagbox}
\newcommand{\subsubsubsection}[1]{\paragraph{#1}\mbox{}}
\IEEEoverridecommandlockouts
\definecolor{Dyellow}{RGB}{254,152,0}
\definecolor{Dgreen}{RGB}{0,176,80}

\makeatletter

\newcommand{\Rmnum}[1]{\expandafter\@slowromancap\romannumeral #1@}
\makeatother
\def\BibTeX{{\rm B\kern-.05em{\sc i\kern-.025em b}\kern-.08em
T\kern-.1667em\lower.7ex\hbox{E}\kern-.125emX}}
\begin{document}

\title{Sequential Outlier Hypothesis Testing under Universality Constraints
\thanks{}
}
\author{Jun Diao and Lin Zhou
\thanks{This paper was partially presented at ITW 2024~\cite{diao2024sequentialOHT}.}
\thanks{The authors are with School of Cyber Science and Technology, Beihang University, Beijing, China (100191). Emails: \{jundiao,~lzhou\}@buaa.edu.cn.
}
}

\maketitle

\begin{abstract}
We revisit sequential outlier hypothesis testing and derive bounds on achievable exponents when both the nominal and anomalous distributions are \emph{unknown}. The task of outlier hypothesis testing is to identify the set of outliers that are generated from an anomalous distribution among all observed sequences where the rest majority are generated from a nominal distribution. In the sequential setting, one obtains a symbol from each sequence per unit time until a reliable decision could be made. For the case with exactly one outlier, our exponent bounds are tight, providing exact large deviations characterization of sequential tests and strengthening a previous result of Li, Nitinawarat and Veeravalli (2017). In particular, the average sample size of our sequential test is bounded universally under any pair of nominal and anomalous distributions and our sequential test achieves larger Bayesian exponent than the fixed-length test, which could not be guaranteed by the sequential test of Li, Nitinawarat and Veeravalli (2017). For the case with at most one outlier, we propose a threshold-based test that has bounded expected stopping time under mild conditions and we bound the exponential decay rate of error probabilities under each non-null hypothesis and the null hypothesis. Our sequential test resolves the tradeoff among the exponential decay rates of misclassification, false reject and false alarm probabilities for the fixed-length test of Zhou, Wei and Hero (TIT 2022). Finally, with a further step towards practical applications, we generalize our results to the cases of multiple outliers and show that there is a penalty in the error exponents when the number of outliers is unknown.
\end{abstract}


\section{Introduction}

Outlier hypothesis testing is a statistical inference problem~\cite{li2014,li2017universal,zhou2022second,bu2019linear}, where one is asked to identify a set of outliers among a given number $M$ of observed sequences. The majority of sequences, named nominal samples, are generated i.i.d. from a nominal distribution and the rest sequences, named outliers, are generated i.i.d. from an anomalous distribution that is different from the nominal distribution.  Both the nominal and anomalous distributions are \emph{unknown}. The number of outliers can be either known or unknown. When the number of outliers is known, the problem is relatively simpler and generalizes statistical classification~\cite{gutman1989asymptotically,zhou2020second}. When the number of outliers is unknown, one could estimate the number of outliers and subsequently identify the set of outliers using a test for the known number case. As a compromise, one could consider the case of at most $T$ outliers, where an upper bound $0<T\le\lceil\frac{M}{2}-1\rceil$ on the number of outliers is known. Under this case, the null hypothesis indicates that there is no outlier while each non-null hypothesis specifies the indices of the outliers. In this paper, we consider both cases of exactly $T$ and at most $T$ outliers. For ease of understanding, we first study the case of $T=1$ and subsequently generalize results to the case of $T>1$. When $T=1$, the cases with known and unknown number of outliers are termed exactly one outlier and at most one outlier, respectively~\cite{li2014,zhou2022second}.

Depending on the test design, a test can be fixed-length or sequential. When the sample size of each observed sequence is fixed, the corresponding test is a fixed-length test. When the sample size is a random variable depending on particular realization of observed sequences, the corresponding test is a sequential test. In a sequential test, one obtains a new sample from each sequence per unit time until one is confident to make a reliable decision. The expected value of the sample size is also known as the expected stopping time. Since the generating distributions of sequences are unknown, for sequential tests, naturally, one can put a universal constraint either on the error probability or the average stopping time~\cite[Def. 2 and 3]{Ihwang2022sequential}. Specifically, for any pair of nominal and anomalous distributions, the error probability universality constraint requires the test to have the error probability bounded by a tolerable value $\beta\in(0,1)$ under each hypothesis while the expected stopping time universality constraint requires that the expected stopping time under each hypothesis is bounded. Correspondingly, for fixed-length tests, only error probability universality constraint is valid since the sample size is fixed a-priori.

For the case with known number of outliers, Li, Nitinawarat and Veeravalli~\cite[Theorems 3 and 10]{li2014} proposed generalized likelihood (GL) fixed-length tests and proved the optimality of the test by having largest exponential decay rates of error probabilities when the number $M$ of observed sequences tends to infinity. Subsequently, Li, Nitinawarat and Veeravalli~\cite[Theorems 3.2 and 4.2]{li2017universal}  generalized the above results to the sequential setting for the case with unknown number of outliers under the error probability universality constraint, which can be specialized to the case with known number of outliers. However, there are several limitations for the results in \cite{li2017universal}. Firstly, only achievability results under the error probability universality constraint were derived. Without a matching converse result, the optimality of error exponents could not be guaranteed. Secondly, the expected stopping time constraint was not considered, which leads to the undesired fact that the sequential tests might stop at very large sample sizes. Thirdly, although the null hypothesis with no outlier is considered, the exponents were only derived for error probabilities under non-null hypotheses. Finally, it was only numerically shown that the sequential test outperforms the fixed-length test under the non-null hypotheses when the average stopping time is relatively large~\cite[Figs. 1 and 2]{li2017universal}. Without a theoretical guarantee, the benefit of the sequential test is not fully uncovered.
Our first contribution addresses all above limitations by considering a slightly easier setting of exactly one outlier and exactly multiple outliers. Furthermore, for all limitations but the first one, our second contribution, stated below, addresses the case with at most $T$ outliers.

For the case with unknown number of outliers, Li, Nitinawarat and Veeravalli~\cite[Theorem 5]{li2014}  proposed another fixed-length test, characterized the exponential decay rate of maximal error probability under non-null hypotheses and proved that the error probability under the null hypothesis vanishes as the sample size increases. Subsequently, Zhou, Wei and Hero~\cite[Theorem 3 and 6]{zhou2022second} proposed an optimal fixed-length test under the generalized Neyman-Pearson criterion and characterized the tradeoff between the exponential decay rates of error probabilities under non-null and null hypotheses. Li, Nitinawarat and Veeravalli~\cite[Theorem 3.2 and 4.2]{li2017universal} generalized the results in \cite{li2014} to the sequential setting and derived the exponential decay rate of maximal error probability under non-null hypotheses. However, for the case with unknown number of outliers, there are several limitations for the sequential tests in \cite{li2017universal}. Firstly, the authors did not characterize the exponential decay rate of error probability under the null hypothesis. Furthermore, there is no theoretical comparison between error exponents of the sequential and the fixed-length tests. One might wonder whether it is possible to design a sequential test that has guaranteed better theoretical performance than the fixed-length test and also resolves tradeoff of error probabilities of the fixed-length test under non-null and null hypotheses. Our second contribution sheds lights on the positive answer for this question for at most one outlier and at most multiple outliers. In the following subsection, we clarify our main contributions.

\subsection{Main Contributions}
Our main contribution is error exponent analyses for sequential tests under universality constraints when both the nominal and anomalous distributions are unknown and when the number of outliers is either known or unknown.

For the case of exactly one outlier, we refine the result in~\cite[Theorem 3.2]{li2017universal} by deriving a matching converse result and re-proving a simpler achievability part under the error probability universality constraint. Furthermore, we derive the exact error exponents under the expected stopping time constraint. In particular, in the achievability part, we propose a sequential test that has bounded average stopping time under any pair of nominal and anomalous distributions and analytically show that the test could have strictly better performance than the fixed-length test in~\cite[Eq. (13)]{li2014}. Through numerical examples, we show that in certain scenarios, the error exponent of an optimal test under the expected stopping time universality is larger than the error exponent of an optimal test under the error probability universality constraint. We also generalize our results to the case of exactly multiple outliers, derive tight bounds on error exponents and analytically demonstrate that our test has strictly better performance than the fixed-length test in~\cite[Eq. (37)]{li2014}.

For the case of at most one outlier, we strengthen \cite[Theorem 3.2]{li2017universal} by proposing a threshold-based test that satisfies the expected stopping time constraint under mild conditions and by characterizing the achievable error exponents under each non-null and the null hypotheses. Our sequential test improves the performance of the fixed-length test in~\cite[Eq. (5)]{zhou2022second} by resolving the tradeoff among the error probabilities under non-null and null hypotheses and having a larger Bayesian error exponent. Furthermore, we generalize our results to the case of at most multiple outliers and derive achievable error exponents. We demonstrate that when the upper bound $T$ on the number of outliers increases, the performance in terms of Bayesian error exponent decreases. For the case of at most $T$ outliers, even when the true number of outliers is one, the performance is inferior compared with the performance for the case of at most one outlier. This is because we need to consider the additional error event where the number of outliers is incorrectly estimated. Finally, comparing with the cases when the number of outliers is known, we theoretically show that there is a penalty in the error exponents when the number of outliers is unknown.

The contributions of our paper go beyond \cite{li2017universal,Ihwang2022sequential}, although our proof techniques are similar to them, using the method of types~\cite{csiszar1998mt} extensively. Specifically, our main contributions beyond \cite{li2017universal} are summarized in the Table \ref{contribution}. Firstly, we prove tight large deviations results for optimal sequential results under the error probability universality constraint while only achievability results were proved in \cite{li2017universal}. Secondly, we propose sequential tests for OHT under the expected stopping time universality constraint and analyze the performance of such tests, which has not been previously addressed. Thirdly, we theoretically verify the superiority of sequential tests for OHT while only numerical results were provided in \cite{li2017universal}. On the other hand, compared with \cite{Ihwang2022sequential} where the authors studied sequential tests for binary classification, we consider the considerably more difficult OHT problem in this paper. Specifically, in binary classification, there are two training sequences and one testing sequence, and the task is to determine the generating distribution of the testing sequence. In contrast, in OHT studied in this paper, one is given a number of sequences and the task is to determine whether there exists outliers and identify the set of outliers if yes. Thus, the contributions of this paper lie in the new theoretical contributions of OHT instead of proposing new proof techniques.
\begin{table}
\centering
\caption{Our contributions beyond prior work under different settings}
\begin{tabular}{ScScSc}
\toprule
\diagbox{Contributions}{Tests} & Existing sequential test~\cite{li2017universal}  & Our sequential test \\
\midrule
\multirow{2}*{Results under error probability universality}& \multirow{2}*{Only achievability results} & Re-proving a simper achievability \\
~& ~& Providing a matching converse\\
\hline
Results under expected stooping time universality  & No result & Deriving exact large deviations\\
\hline
The superiority over fixed-length tests~\cite{li2014,zhou2022second}  & {Only numerical examples}  & {Theoretical verification and numerical illustration} \\
\hline
{Analyses of the null hypothesis}  & \multirow{2}*{No result} & Characterizing error probabilities under the null hypothesis\\
when the number of outliers is unknown &~ & Resolving the tradeoff among three error probabilities\\
\bottomrule
\end{tabular}
\label{contribution}
\end{table}

\subsection{Other Related Works}
We briefly recall other (non-exhausting) related works on outlier hypothesis testing. For both cases with known and unknown number of outliers, Bu, Zou and Veeravalli~\cite{bu2019linear} proposed a low-complexity test for outlier hypothesis testing and showed that the test ensures exponential decay of error probabilities.
Zou~\emph{et al.}~\cite{zou2017nonparametric} used the maximum mean discrepancy metric to design a fixed-length test for outlier hypothesis testing of continuous sequences and showed that the test is exponentially consistent. Recently, Zhu and Zhou~\cite{zhu2024exponentially} generalized the results of \cite{zou2017nonparametric} to derive error exponents of distribution free fixed-length and sequential tests for continuous observed sequences.

Outlier hypothesis testing is also closely related with statistical classification where one is asked to determine whether the testing sequences is generated from one of $M$ unknown distributions using training sequences. Gutman~\cite{gutman1989asymptotically} proposed an optimal test under the generalized Neyman-Pearson criterion and characterized its error exponents. Haghifam, Tan and Khisti~\cite{mahdi2021sequential} generalized Gutman's result to the sequential setting by proposing a sequential test under the error probability universal constraint and deriving achievable error exponents. The results for the binary case of \cite{mahdi2021sequential} were further generalized by Hsu, Li and Wang~\cite{Ihwang2022sequential}, who proposed the sequential tests under both the expected stopping time and error probability universality constraints and derived corresponding tight error exponent bounds. Diao, Zhou and Bai~\cite{diao2023classification,zhou2023achievable} proposed a two-phase test for statistical classification and showed that their test achieves performance close to an optimal sequential test with design complexity propositional to a fixed-length test.

\subsection*{Notation}
We use $\bbR$, $\bbR_+$, $\bbN$ to denote the set of real numbers, non-negative real numbers, and natural numbers respectively. Given any
two integers $(a,b)\in\bbN^2$ such that $1\leq a<b$, we use $[a:b]$ to denote the set of integers $\{a,a+1,\ldots,b\}$ and use $[a]$ to denote $[1:a]$. Random variables and their realizations are denoted by upper case variables (e.g., $X$) and lower case variables (e.g., $x$), respectively. All sets are denoted in calligraphic font (e.g., $\mathcal{X}$). Given any set $\calA$, we use $\calA^c$ to denote its complementary set. Given any integer $n\in\bbN$, let $X^n:=(X_1,\ldots X_n)$ be a random vector of length $n$ and let $x^n=(x_1,\ldots,x_n)$ be a particular realization of $X^n$. The set of all probability distributions on a finite set $\calX$ is denoted as $\calP(\calX)$. We use $\mathbb{E}[\cdot]$ to denote expectation. Given a sequence $x^n\in\calX^n$, the type or empirical distribution is defined as $\hatT_{x^n}(a)=\frac{1}{n}\sum_{i=1}^{n}\mathbbm{1}(x_i=a),\forall a\in\calX$. The set of types formed from length-$n$ sequences with alphabet $\calX$ is denoted by $\calP^n(\calX)$. Given any $P\in\calP^n(\calX)$, the set of all sequences of length $n$ with type $P$, a.k.a. the type class, is denoted by $\calT_P^n$.

\section{Case of Exactly One Outlier}

\subsection{Problem Formulation}
Consider a set of $M$ observed sequences $\bX^{\tau}:=\big\{X^{\tau}_1,\ldots,X^{\tau}_M\big\}$, where $\tau$ is a random stopping time with respect to the filtration $\{\calF_n\}_{n\in\bbN}$ and $\calF_n$ is generated by $\sigma$-algebra $\sigma\{X_1,X_2,\ldots X_n\}$ for each $n\in\bbN$. Most sequences are generated i.i.d. from an unknown nominal distribution $P_\rmN$ while the rest sequences known as outliers are generated i.i.d. from an unknown anomalous distribution $P_\rmA$. In this section, we consider the case of exactly one outlier while in subsequent sections, we generalize our results to the cases of at most one outlier and multiple outliers with the number of outliers known and unknown.

When there is exactly one outlier, the task is to design a test $\Phi=(\tau,\phi)$ that consists of a random stopping time $\tau$ and a decision rule $\phi:\calX^{M\tau}\to\{{\rmH}_1,{\rmH}_2,\ldots,{\rmH}_M\}$ to classify among the following $M$ hypotheses:
\begin{itemize}
\item $\rmH_i,~i\in[M]$: the $i$-th sequence is the outlier.
\end{itemize}
To evaluate the performance of a test, we consider the misclassification probability and the expected stopping time. Specifically, for each $i\in[M]$, the misclassification probability is defined as follows:
\begin{align}
\label{def:misclassify:p}
\psi_i(\Phi|P_\rmN,P_\rmA):=\bbP_i\{\Phi(\bX^\tau)\neq\rmH_i\},i\in[M],
\end{align}
where we define $\bbP_i(\cdot):=\Pr\{\cdot|\rmH_i\}$ to denote the joint distribution of observed sequences $\bX^\tau$, where $X_i^\tau$ is generated i.i.d. from the anomalous distribution $P_\rmA$ and for each $j\in\calM_i:=\{j\in[M]:j\neq i\}$, $X_j^\tau$ is generated i.i.d. from the nominal distribution $P_\rmN$.
Furthermore, the expected stopping time under hypothesis $\rmH_i$ satisfies
\begin{align}\label{def:Ei}
\mathbb{E}_i[\tau]=\sum_{k=1}^{\infty}\bbP_i\{\tau> k\}.
\end{align}

Since the generating distributions $(P_\rmN,P_\rmA)$ are unknown and there are two performance criteria, one could put a universal constraint on either criterion. Motivated by the analyses for sequential binary classification~\cite[Def. 2 and 3]{Ihwang2022sequential}, we define the following two universality constraints on sequential tests.

\begin{definition}
(Error Probability Universality Constraint):  A sequential test $\Phi$ is said to satisfy the error probability universality constraint if there exists a positive real number $\beta\in(0,1)$ such that for any pair of distributions $(P_\rmN,P_\rmA)\in\calP(\calX)^2$,
\begin{align}\label{constraint1:ep}
\max_{i\in[M]}\psi_i(\Phi|P_\rmN,P_\rmA)\le \beta.
\end{align}
\end{definition}
For any sequential test $\Phi$ satisfying the error probability universality constraint, we are interested in the following misclassification exponent for each $i\in[M]$:
\begin{align}
E_i^{\mathrm{EP}}(\Phi|P_\rmN,P_\rmA):=\liminf_{\beta\to 0}\frac{-\log\beta}{\mathbb{E}_i[\tau]}.
\end{align}

\begin{definition}
(Expected Stopping Time Universality Constraint): A sequential test $\Phi$ is said to satisfy the expected stopping time universality constraint if there exists an integer $n\in\bbN$ such that for any pair of distributions $(P_\rmN,P_\rmA)\in\calP(\calX)^2$,
\begin{align}\label{constraint1:est}
\max_{i\in[M]} \mathbb{E}_i[\tau]\le n.
\end{align}
\end{definition}
For a sequential test satisfying the expected stopping time universality constraint, we are interested in the following misclassification exponent for each $i\in[M]$:
\begin{align}\label{def:exponent:est}
E_i^{\mathrm{EST}}(\Phi|P_\rmN,P_\rmA):=\liminf_{n\to\infty}\frac{-\log\psi_i(\Phi|P_\rmN,P_\rmA)}{n}.
\end{align}

\subsection{Existing Results for the Fixed-Length Test}
To compare the performance of sequential tests and fixed-length tests, we first recall the fixed-length test $\Phi_{\rm LNV}$ by Li, Nitinawarat and Veeravalli~\cite[Eq. (15)]{li2014}. Note that the fixed-length test satisfies the stopping time universal constraint. Recall that $\calM_i=\{j\in[M]:j\neq i\}$. Given any two distributions $(P,Q)\in\calP(\calX)^2$, the KL divergence is defined as
\begin{align}
D(P\|Q)&=\sum_{x\in\calX}P(x)\log\frac{P(x)}{Q(x)}.\label{def:KL}
\end{align}
Note that $D(P\|Q)$ is continuous with respect to distributions $P$ and $Q$.
Furthermore, given a tuple of distributions $\bQ=(Q_1,\ldots,Q_M)\in\calP(\calX)^M$, for each $i\in[M]$, define the following linear combination of KL divergence terms between each single distribution and a mixture distribution~\cite[Eq. (4)]{zhou2022second}:
\begin{align}\label{G_i}
\rmG_i(\bQ):=\sum\limits_{j\in\calM_i}D\bigg(Q_j\bigg\|\frac{\sum_{l\in\calM_i}Q_l}{M-1}\bigg).
\end{align}
Note that $\rmG_i(\bQ)$ measures the similarity of distributions $\bQ$ except $Q_i$. The value of $\rmG_i(\bQ)$ equals zero if and only if $Q_j=Q$ for all $j\in\calM_i$ with an arbitrary $Q\in\calP(\calX)$. Due to the continuous property of $D(P\|Q)$, $\rmG_i(\bQ)$ is also continuous with respect to distributions $\bQ$. For the case of exactly one outlier, the test in~\cite[Eq. (15)]{li2014} applies the following minimal scoring function decision rule:
\begin{align}\label{test_Li}
\Phi_{\rm LNV}(\bx^n)=\rmH_j,~\mathrm{if}~j=\argmin\limits_{i\in[M]}\rmS_i(\bx^n),
\end{align}
where the scoring function $\rmS_i(\cdot)$ satisfies
\begin{align}
\rmS_i(\bx^n)=\rmG_i(\hatT_{x_1^n},\ldots,\hatT_{x_M^n})\label{def:scrore:si}.
\end{align}

Li, Nitinawarat and Veeravalli derived the following result~\cite[Theorem 2]{li2014}.
\begin{theorem}\label{fixed}
Given any pair of distributions $(P_\rmN,P_\rmA)\in\calP(\calX)^2$ that are fully supported on the finite alphabet $\calX$, for each $i\in[M]$, the achievable error exponent of the fixed-length test $\Phi_{\rm LNV}$ satisfies
\begin{align}\label{fixed:error}
E_i^{\mathrm{EST}}(\Phi_{\rm LNV}|P_\rmN,P_\rmA)=\min_{\substack{\bQ\in\calP(\calX)^M:\\\rmG_1(\bQ)\ge\rmG_2(\bQ)}}D(Q_1||P_\rmA)+\sum\limits_{j\in[2:M]}D(Q_j||P_\rmN).
\end{align}
\end{theorem}

Our first contribution in this paper is to characterize the optimal error exponent of sequential tests for outlier hypothesis testing under both the expected stopping time universality constraint and the error probability universality constraint. Specifically, for each case, we propose a corresponding sequential test using the empirical distributions of observed sequences and derive exact large deviations for error probabilities. In particular, our results strengthen \cite[Theorem 3.2]{li2017universal} by Li, Nitinawarat and Veeravalli for the case of exactly one outlier in two aspects: i) we derive a matching converse result for optimal sequential tests under the error probability universality constraint, and ii) we analytically show that our sequential test satisfying the expected stopping time universality constraint has larger exponent than the fixed-length test.

\subsection{Error Probability Universality Constraint}

\subsubsection{Test Design and Intuition}\label{test_errorprob}
Recall the definition of the scoring function $\rmS_l(\cdot)$ for any $l\in[M]$ in \eqref{def:scrore:si}. Given any positive real number $\beta\in(0,1)$ and any integer $k\in\bbN$, define a threshold function
\begin{align}\label{gbeta}
g(\beta,k):=\frac{-\log\big(\beta(|\calX|-1)\big)}{k}+\frac{(M+1)|\calX|\log(k+1)}{k}.
\end{align}
and a set of indices
\begin{align}\label{psik}
\Psi(\bx^k):=\Big\{l\in[M]:~\rmS_l(\bx^k)>g(\beta,k)\Big\}.
\end{align}
As we shall show, the threshold $g(\beta,k)$ is critical to universally bound the misclassification probability by $\beta$ (cf. \eqref{gk}).

Under the error probability universality constraint, our sequential test $\Phi^\mathrm{EP}=(\tau^{\rm{EP}},\phi^{\rm{EP}})$ consists of the random stopping time $\tau^{\rm{EP}}$ and the decision rule. Specifically, the random stopping time $\tau^{\rm{EP}}$ satisfies
\begin{align}\label{tau:ep}
\tau^{\rm{EP}}:=\inf\big\{k\in\bbN:~|\Psi(\bx^k)|\ge M-1\big\}.
\end{align}
Note that for each $l\in[M]$, $\rmS_l(\bx^k)$ measures the closeness of types of all sequences except the $l$-th sequence. Thus, the sequential test $\Phi^\mathrm{EP}$ stops if the types of nominal samples and the outlier are far away for all $M-1$ possibilities that mix $M-2$ nominal samples and the outlier. The threshold $g(\beta,k)$ determines how ``far away'' is measured, which decreases with the error probability $\beta$ and the sample size $k$.

At stopping time $\tau^{\rm{EP}}$, our test uses the following decision rule:
\begin{align}
\label{ep:test}
\phi^\mathrm{EP}\big(\bx^{\tau^\mathrm{EP}}\big)=\rmH_i,~\mathrm{if}~i=[M]\backslash\Psi\big(\bx^{\tau^\mathrm{EP}}\big).
\end{align}
The above test generalizes the test for sequential classification in \cite[Sec. IV. B]{Ihwang2022sequential} and \cite[Eq. (24)]{mahdi2021sequential}.

We now explain the asymptotic intuition why the above test works using the weak law of large numbers. Fix any $i\in[M]$. Under hypothesis $\rmH_i$, for each $j\in\calM_i$, the nominal sequence $x_j^k$ is generated from the nominal distribution $P_\rmN$ while the outlier $x_i^k$ is generated from the anomalous distribution $P_\rmA$. It follows from \eqref{def:scrore:si} that the scoring function $\rmS_i(\bx^k)$ satisfies
\begin{align}
\rmS_i(\bx^k)&=\rmG_i(\hatT_{x_1^n},\ldots,\hatT_{x_M^n})\\
&\to \rmG_i(P_\rmN,\ldots,P_\rmA,\ldots,P_\rmN)\label{Sto}\\
&= 0,\label{defG}
\end{align}
where \eqref{Sto} denotes convergence in probability and follows from the continuous property of $\rmG_i(\bQ)$ and the weak law of large numbers, which implies that the empirical distribution $\hatT_{x_j^k}$ of a nominal sequence $x_j^k$ converges to $P_\rmN$ while the empirical distribution $\hatT_{x_i^k}$ of the outlier $x_i^k$ converges to $P_\rmA$ when the sample size $k$ is sufficiently large, and \eqref{defG} follow from the definition of $\rmG_i(\cdot)$ in \eqref{G_i}. Similarly, for each $j\in\calM_i$, the scoring function $\rmS_j(\bx^k)$ satisfies
\begin{align}
\rmS_j(\bx^k)&=\rmG_j(\hatT_{x_1^n},\ldots,\hatT_{x_M^n})\\
&\to \rmG_j(P_\rmN,\ldots,P_\rmA,\ldots,P_\rmN)\\
&=D\Big(P_\rmA\Big\|\frac{P_\rmA+(M-2)P_\rmN}{M-1}\Big)+(M-2)D\Big(P_\rmN\Big\|\frac{P_\rmA+(M-2)P_\rmN}{M-1}\Big)\\
&>0.
\end{align}
Therefore, when $k$ is sufficiently large, it follows that there exists $M-1$ scoring functions with positive values greater than the vanishing value of $g(\beta,k)$ and subsequently, a correct decision could always be made asymptotically.

\subsubsection{Main Results and Discussions}
We need the following definition to present our results.
Given any positive real number $\alpha\in\bbR_+$, the generalized Jensen-Shannon divergence~\cite[Eq. (2.3)]{zhou2020second} is defined as
\begin{align} \label{GJS}
{\rm GJS}(P,Q,\alpha):=\alpha D\left(P\Big\| \frac{\alpha P+Q}{1+\alpha}\right)+D\left(Q\Big \| \frac{\alpha P+Q}{1+\alpha}\right).
\end{align}
The generalized Jensen-Shannon Divergence has the following variational form~\cite[Eq. (6)]{Ihwang2022sequential}:
\begin{align}
\label{GJS:variational}
{\rm GJS}(P,Q,\alpha):=\min_{V\in\calP(\calX)}\alpha D(P||V)+D(Q||V).
\end{align}

\begin{theorem}\label{seq_errorprob}
For any pair of distributions $(P_\rmN,P_\rmA)\in\calP(\calX)^2$ that are fully supported on the finite alphabet $\calX$, our sequential test satisfies the error probability universality constraint, and for each $i\in[M]$, the misclassification exponent of our test satisfies
\begin{align}
E_i^\mathrm{EP}(\Phi^\mathrm{EP}|P_\rmN,P_\rmA)\ge\mathrm{GJS}(P_\rmN,P_\rmA,M-2).
\end{align}
Conversely, for any sequential test $\Phi$ satisfying the error probability universality constraint, under any pair of distributions $(P_\rmN,P_\rmA)\in\calP(\calX)^2$, for each $i\in[M]$, the misclassification exponent satisfies
\begin{align}
E_i^\mathrm{EP}(\Phi|P_\rmN,P_\rmA)\le\mathrm{GJS}(P_\rmN,P_\rmA,M-2).
\end{align}
\end{theorem}
The proof of Theorem \ref{seq_errorprob} is provided in Appendix \ref{proof:ep}, which is inspired by the proof of \cite[Theorem 2]{Ihwang2022sequential} for sequential binary classification under the error probability universality constraint. In the achievability proof, we show that for our sequential test, the random variable $\frac{\tau^{\rm{EP}}}{\log\beta}$ is uniformly integrable and subsequently we obtain the desired exponent by analyzing the convergence properties of $\frac{-\log\beta}{\mathbb{E}[\tau^{\rm{EP}}]}$. In the converse part, we use the binary KL divergence and apply the data processing inequality.

Since $\mathrm{GJS}(P,Q,\alpha)$ increases in $\alpha$, it follows from Theorem \ref{seq_errorprob} that as the number $M$ of observed sequences increases, the misclassification exponent increases. This is consistent with our intuition since with more samples, the estimation of the nominal distribution is more accurate. Thus, it is easier to identify the outlier. In the extreme case of $M\to\infty$, the exponent equals to $D(P_\rmA\|P_\rmN)$, which is exactly the performance of knowing the nominal distribution~\cite[Prop. 3.1]{li2017universal}.

Theorem \ref{seq_errorprob} strengthens \cite[Theorem 3.2]{li2017universal} by deriving a matching converse result. We manage to do so for a slightly easier setting of exactly one outlier by excluding the null hypothesis that claims there is no outlier. Furthermore, we provide a relatively simpler achievability proof by proposing and analyzing another sequential test beyond~\cite[Eq. (3.11)]{li2017universal}.

\subsection{Expected Stopping Time Universality Constraint}
\subsubsection{Test Design and Intuition}
\label{test_stoptime}

Under the expected stopping time universality constraint, our sequential test $\Phi^\mathrm{EST}=(\tau^{\rm{EST}},\phi^{\rm{EST}})$ consists of another random stopping time and decision rule.  The stopping time $\tau^{\rm{EST}}$ satisfies
\begin{align}\label{tau:est}
\tau^{\rm{EST}}&:=\inf\{k\ge n-1: \exists~l\in[M]~\mathrm{s.t.}~\rmS_l(\bx^k)\le f(k)\},
\end{align}
where the scoring function $\rmS_l(\cdot)$ was defined in \eqref{def:scrore:si} and the threshold is given by $f(k)=\tfrac{(M+1)|\calX|\log (k+1)}{k}$. As we shall show, the threshold $f(k)$ is critical to universally bound the expected stopping time (cf. \eqref{typeclass}). Intuitively, the sequential test $\Phi^\mathrm{Est}$ stops if the types of all nominal samples are ``close enough'' to each other, where the threshold $f(k)$ is used to characterize the closeness level.

At the stopping time $\tau^\mathrm{EST}$, using $M$ observed sequences $\bx^{\tau^\mathrm{EST}}_M$, our test applies the following minimal scoring function decision rule:
\begin{align}\label{test_est}
\phi^\mathrm{EST}(\bx^{\tau^\mathrm{EST}})={\rm H}_i,\mathrm{~if~}i=\argmin\limits_{\bari\in[M]} \rmS_{\bari}(\bx^{\tau^\mathrm{EST}}).
\end{align}
Our test generalizes the sequential classification test under expected stopping time universality in \cite[Def. 7]{Ihwang2022sequential}.

We now explain the asymptotic intuition reason why the above test works using the weak law of large numbers. Fix any $i\in[M]$. As discussed below \eqref{ep:test}, under hypothesis ${\rm H}_i$, the scoring function $\rmS_i(\bx^k)$ tends to zero while for each $j\in\calM_i$, the scoring function $\rmS_j(\bx^k)$ tends to a positive real number as the sample size $k$ increases. Therefore, when $k$ is sufficiently large, when the $i$-th sequence is the outlier, our test stops and makes the correct decision $\rmH_i$.

We remark that if the sample size is not large enough, the empirical distributions could be rather different from generating distributions, which might lead to decision error. To avoid such error events, similarly to \cite[Def. 7]{Ihwang2022sequential}, we set the minimal stopping time as $n-1$ for some predefined integer $n\in\bbN$.

\subsubsection{Main Results and Discussions}
We need the following definition to present our result. Given any pair of distributions $(P,Q)\in\calP(\calX)^2$ and any $\alpha\in\bbR_+$, the R\'{e}nyi Divergence of order $\alpha$~\cite[Eq. (1)]{van2014renyi} is defined as
\begin{align}
D_\alpha(P||Q):=\frac{1}{\alpha-1}\log\sum_{x\in\calX}P(x)^{\alpha}Q(x)^{1-\alpha}
\label{def:renyi}.
\end{align}
The R\'{e}nyi Divergence has the following variational form~\cite[Eq. (7)]{Ihwang2022sequential}:
\begin{align}
\label{renyi:variational}
D_{\frac{\alpha}{1+\alpha}}(P||Q):=\min_{V\in\calP(\calX)}\alpha D(V||P)+D(V||Q).
\end{align}
The R\'{e}nyi Divergence is non-decreasing in its order $\alpha$ since i) \eqref{renyi:variational} indicates that $D_\alpha(P||Q):=\min_{V\in\calP(\calX)}{\frac{\alpha}{1-\alpha}} D(V||P)+D(V||Q)$ and ii) $\frac{\alpha}{1-\alpha}=\frac{1}{1-\alpha}-1$ is non-decreasing in $\alpha$.

\begin{theorem}\label{seq_stoptime}
Under any pair of distributions $(P_\rmN,P_\rmA)\in\calP(\calX)^2$ that are fully supported on the finite alphabet $\calX$, our sequential test satisfies the expected stopping time universality constraint and the misclassification exponent of our test satisfies that for each $i\in[M]$,
\begin{align}
E_i^\mathrm{EST}(\Phi^\mathrm{EST}|P_\rmN,P_\rmA)\ge D_{\frac{M-2}{M-1}}(P_\rmN||P_\rmA).\label{Ei}
\end{align}
Conversely, for any sequential test $\Phi$ satisfying the expected stopping time universality constraint, under any pair of distributions $(P_\rmN,P_\rmA)\in\calP(\calX)^2$, the misclassification exponent satisfies that for each $i\in[M]$,
\begin{align}
E_i^\mathrm{EST}(\Phi|P_\rmN,P_\rmA)\le D_{\frac{M-2}{M-1}}(P_\rmN||P_\rmA).\label{con:ei}
\end{align}
\end{theorem}
The proof of Theorem \ref{seq_stoptime} is provided in Appendix \ref{proof:est}.

We make several remarks. Firstly, Theorem \ref{seq_stoptime} strengthens \cite[Theorem 3.2]{li2017universal} by considering sequential tests satisfying the expected stopping time universality constraint and deriving exact exponential decay rates of misclassification probabilities. Compared with the error probability universality constraint in \cite[Sec. 3.1.1]{li2017universal} and Sec. \ref{test_errorprob}, our proposed sequential test in this subsection has the property of having a bounded expected stopping time under any pair of nominal and anomalous distributions. This property is highly desired in practice since one would like the tests to stop early while the tests in \cite[Sec. 3.1.1]{li2017universal} and Sec. \ref{test_errorprob} could stop at very large sample sizes.

Secondly, in contrast to lack of theoretical evidence that the sequential test of \cite[Sec. 3.1.1]{li2017universal} has better performance than the fixed-length test in \eqref{test_Li}, our sequential test has better theoretic performance than the fixed-length test in \eqref{test_Li}. This property is desired since the motivation of using sequential tests is to yield better performance. To clarify, we compare the exponents in Theorems \ref{fixed} and \ref{seq_stoptime}. Given any distribution $Q\in\calP(\calX)$, the tuple of distributions $\bQ=(Q_1,\ldots,Q_{M-1},P_\rmN)$ with $Q_i=Q$ for all $i\in[M-1]$ satisfy the constraints that $\rmG_1(\bQ)\geq \rmG_2(\bQ)$. It follows from \eqref{fixed:error} that the misclassification exponent of the fixed-length test satisfies $E_i^{\mathrm{EST}}(\Phi_{\rm LNV}|P_\rmN,P_\rmA)\leq \min_{Q\in\calP(\calX)}D(Q||P_\rmA)+(M-2)D(Q||P_\rmN)=D_{\frac{M-2}{M-1}}(P_\rmN||P_\rmA)=E_i^\mathrm{EST}(\Phi^\mathrm{EST}|P_\rmN,P_\rmA)$. We numerically verify that $E_i^\mathrm{EST}(\Phi^\mathrm{EST}|P_\rmN,P_\rmA)$ could be strictly greater than $E_i^{\mathrm{EST}}(\Phi_{\rm LVN}|P_\rmN,P_\rmA)$. Specifically, when the nominal and anomalous distributions are $(P_\rmN,P_\rmA)=\mathrm{Bern}(0.3,0.1)$, we have $E_i^\mathrm{EST}(\Phi^\mathrm{EST}|P_\rmN,P_\rmA)=0.0934>E_i^{\mathrm{EST}}(\Phi_{\rm LVN}|P_\rmN,P_\rmA)=0.0471$. Thus, our sequential test under the expected stopping time universality constraint could have strictly better performance than the fixed-length test as desired.

Finally, we compare the optimal error exponents under two universality constraints in Theorems \ref{seq_errorprob} and \ref{seq_stoptime}, i.e.,  $\mathrm{GJS}(P_\rmN,P_\rmA,M-2)$ and $D_{\frac{M-2}{M-1}}(P_\rmN||P_\rmA)$. It follows from \cite[Remark 3]{Ihwang2022sequential} that for any $\alpha\in\bbR_+$ and any $(P_0,P_1)\in\calP(\calX)^2$, $D_{\frac{\alpha}{1+\alpha}}(P_1||P_0)\ge\mathrm{GJS}(P_0,P_1,\alpha)$. Thus, setting $\alpha=M-2$ leads to $D_{\frac{M-2}{M-1}}(P_\rmN||P_\rmA)\ge\mathrm{GJS}(P_\rmN,P_\rmA,M-2)$, which implies that when $M=3$, the sequential test under the expected stopping time constraint achieves larger error exponent because $D_{\frac{1}{2}}(P_\rmN||P_\rmA)\ge\mathrm{GJS}(P_\rmN,P_\rmA,1)=\mathrm{GJS}(P_\rmN,P_\rmA,1)$. However, since $\mathrm{GJS}(P_\rmN,P_\rmA,M-2)\neq \mathrm{GJS}(P_\rmN,P_\rmA,M-2)$ when $M>3$, the performance comparison of sequential tests under two universality constraints depends on the nominal and anomalous distributions. To illustrate, in Table \ref{tab}, we calculate the exponents in Theorems \ref{seq_errorprob} and \ref{seq_stoptime} for various pairs of distributions.

\begin{table}
\centering
\caption{Comparison of Achievable Error Exponents of Two Sequential Tests Under Different Universality Constraints}
\label{tab}
\begin{tabular}{ccc}
\toprule
Parameters& $D_{\frac{M-2}{M-1}}(P_\rmN||P_\rmA)$ & $\mathrm{GJS}(P_\rmN,P_\rmA,M-2)$ \\
\midrule
$M=3$\\$(P_\rmN,P_\rmA)=\mathrm{Bern}(0.2,0.4)$& \textbf{0.0493} &0.0483\\
\hline
$M=4$\\$(P_\rmN,P_\rmA)=\mathrm{Bern}(0.2,0.4)$& 0.0642&\textbf{0.0659}\\
\hline
$M=4$\\$(P_\rmN,P_\rmA)=\mathrm{Bern}(0.3,0.1)$& \textbf{0.0939}&0.0830\\
\bottomrule
\end{tabular}
\end{table}

We present a numerical example to illustrate Theorem \ref{seq_stoptime} in Fig. \ref{simulation_est}. Consider the binary alphabet $\calX=\{0,1\}$. Set the number of sequences $M=4$ and set generating distributions $(P_\rmN,P_\rmA)=\mathrm{Bern}(0.28,0.25)$. We simulate the misclassification probability of our sequential test in \eqref{test_est} versus the exponential estimates in Theorem \ref{seq_stoptime} when there is one outlier of $M$ sequences using $5 \times 10^4$ independent experiments. As observed, when the average sample size is large, our theoretical characterization for the misclassification probability provides a tight upper bound.

\begin{figure}[htbp]
\centering
\includegraphics[width=.5\columnwidth]{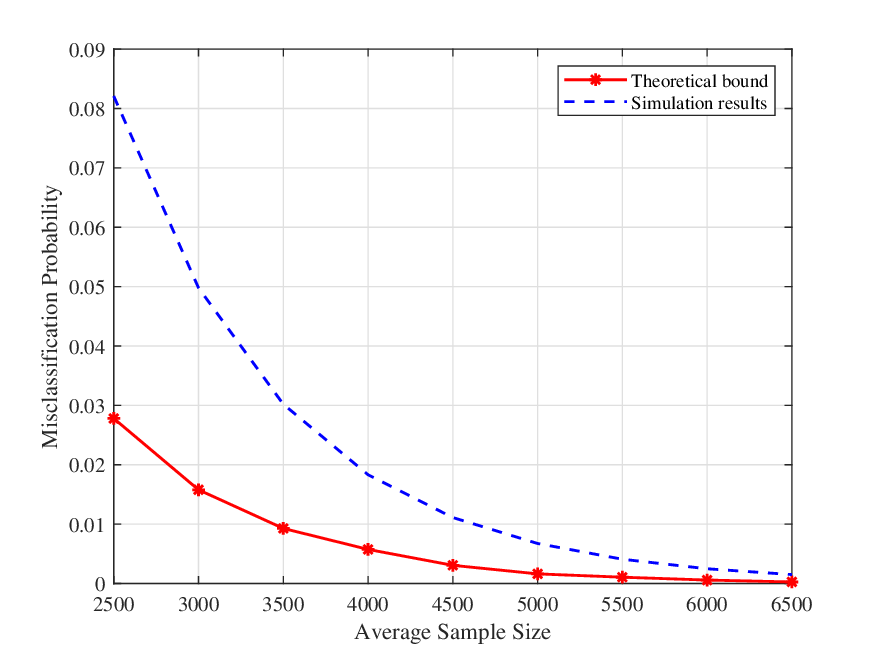}
\caption{Simulated misclassification probability of our test under expected stopping time universality constraint for the case with exactly one outlier under generating distributions $(P_\rmN,P_\rmA)=\mathrm{Bern}(0.28,0.25)$ with $M=4$.}
\label{simulation_est}
\end{figure}

Since the sequential test under the expected stopping time universality constraint has the desired property of having a bounded expected stopping time, in the following, we only consider sequential tests under the expected stopping time universality constraint for the cases of exactly multiple, at most one outlier, and at most multiple outliers.

\section{Case of Exactly $T$ Outlier}

\subsection{Problem Formulation}

In this section, it is assumed that there are $0<T\le\lceil\frac{M}{2}-1\rceil$ identically distributed outliers and the number of outliers $T$ is known. We propose a sequential test under the expected stopping time constraint and derive tight bounds on the achievable error exponents.

Let $\calS(T)$ denote the set of all subsets of $[M]$ of size $T$, i.e., $\calS(T):=\{\calB\subset[M]:|\calB|=T\}$. Our task now is to design a test $\Phi=\{\tau,\phi\}$ with a stopping time $\tau$ and a corresponding decision rule $\phi:~\calX^{M\tau}\to\{\{\rmH_\calB\}_{\calB\in\calS(T)}\}$ to classify among the following $|\calS(T)|$ hypotheses:
\begin{itemize}
\item $\rmH_{\calB}$,~$\calB\in\calS(T)$: for each $j\in\calB$, the $j$-th sequence is an outlier.
\end{itemize}
To evaluate the performance of a test, for each $\calB\in\calS(T)$, we consider the following misclassification probability under hypothesis $\rmH_\calB$:
\begin{align}\label{def:misclassify:pb}
\psi_\calB(\Phi|P_\rmN,P_\rmA)=\bbP_\calB\{\Phi(\bX^\tau)\neq\rmH_\calB\},
\end{align}
where we define $\bbP_\calB(\cdot):=\Pr\{\cdot|\rmH_\calB\}$ to denote the joint distribution of observed sequences $\bX^\tau$, where for each $i\in\calB$, $X_i^\tau$ is generated i.i.d. from the anomalous distribution $P_\rmA$ and for each $j\in\calM_\calB:=[M]\backslash\calB=\{j\in[M]:j\notin \calB\}$, $X_j^\tau$ is generated i.i.d. from the nominal distribution $P_\rmN$.
Furthermore, the expected stopping time under hypothesis $\rmH_\calB$ satisfies
\begin{align}\label{def:EB}
\mathbb{E}_\calB[\tau]&=\sum_{k=1}^{\infty}\bbP_\calB\{\tau> k\}.
\end{align}

\begin{definition}
A sequential test $\Phi$ is said to satisfy the expected stopping time universality constraint if there exists an integer $n\in\bbN$ such that for any pair of distributions $(P_\rmN,P_\rmA)\in\calP(\calX)^2$,
\begin{align}\label{constraint2:est}
\max_{\calB\in\calS(T)}\mathbb{E}_\calB[\tau]\le n.
\end{align}
\end{definition}
For a sequential test satisfying the expected stopping time universality constraint, we are interested in the following error exponent for each $\calB\in\calS(T)$,
\begin{align}\label{def:exponent:exactT}
E_\calB(\Phi|P_\rmN,P_\rmA):=\liminf_{n\to\infty}\frac{-\log\psi_\calB(\Phi|P_\rmN,P_\rmA)}{n}.
\end{align}

\subsection{Existing Results}
To compare the performance of sequential tests and fixed-length tests, we recall the results of the fixed-length test $\Phi_{\rm LNV}$ by Li, Nitinawarat and Veeravalli~\cite[Eq. (37)]{li2014} for the case with exactly $T$ outliers. To present the test, we need the following definitions. Given a tuple of distributions $\bQ=(Q_1,\ldots,Q_M)\in\calP(\calX)^M$, for each $\calB\in\calS(T)$, define the following linear combination of KL divergence terms:
\begin{align}\label{Gli}
\rmG_{\mathrm{Li},\calB}(\bQ):=\sum_{j\in\calM_\calB}D\Big(Q_j\Big\|\frac{\sum_{l\in\calM_\calB}Q_l}{M-|\calB|}\Big).
\end{align}
For the case of exactly $T$ outlier, the test in~\cite[Eq. (37)]{li2014} applies the following minimal scoring function decision rule:
\begin{align}\label{test_Li_exactT}
\Phi_{\rm LNV}(\bx^n)=\rmH_\calC,~\mathrm{if}~\calC=\argmin\limits_{\calB\in\calS(T)}\rmG_{\mathrm{Li},\calB}\big(\hatT_{x_1^k},\ldots,\hatT_{x_M^k}\big).
\end{align}
Recall the definition of $\calM_\calB=\{i\in[M]:i\notin\calB\}$ and $\calS(T)=\{\calB\subset[M]:|\calB|=T\}$. Furthermore, given any $\calB\in\calS(T)$, define $\calS_\calB(T):=\{\calC\in\calS(T):\calC\neq\calB\}$. Li, Nitinawarat and Veeravalli derived the following result~\cite[Theorem 9]{li2014}.
\begin{theorem}\label{fixed:exactT}
Given any pair of distributions $(P_\rmN,P_\rmA)\in\calP(\calX)^2$ that are fully supported on the finite alphabet $\calX$, the achievable error exponent of the fixed-length test satisfies that for each
$\calB\in\calS(T)$,
\begin{align}\label{fixed:exactT_error}
E_\calB(\Phi_{\rm LNV}|P_\rmN,P_\rmA)&=\min_{\calC\in\calS_\calB(T)}\min_{\substack{\bQ\in\calP(\calX)^M:\\\rmG_{\mathrm{Li},\calB}(\bQ)\ge\rmG_{\mathrm{Li},\calC}(\bQ)}}
\sum\limits_{j\in\calM_\calB}D(Q_j||P_\rmN)+\sum\limits_{i\in\calB}D(Q_i||P_\rmA).
\end{align}
\end{theorem}

\subsection{Test Design}\label{test:exactlyT}

To present our tests, we need the following definition. Given a tuple of distributions $\bQ=(Q_1,\ldots,Q_M)\in\calP(\calX)^M$, for each $\calB\in\calS(T)$,  define
\begin{align}\label{G_B}
\rmG_\calB(\bQ)=\sum_{j\in\calM_\calB}D\Big(Q_j\Big\|\frac{\sum_{l\in\calM_\calB}Q_l}{M-|\calB|}\Big)+
\sum_{i\in\calB}D\Big(Q_i\Big\|\frac{\sum_{t\in\calB}Q_t}{|\calB|}\Big).
\end{align}
Note that $\rmG_\calB(\bQ)$ measures the similarity of distributions $\{Q_i\}_{i\in\calB}$ and $\{Q_j\}_{j\in\calM_\calB}$. The measure $\rmG_\calB(\bQ)=0$ if and only if $Q_j=Q_1$ for all $j\in\calM_\calB$ and $Q_i=Q_2$ for all $i\in\calB$ for an arbitrary pair of distributions $(Q_1,Q_2)\in\calP(\calX)^2$.
Furthermore, when $T=1$, $\rmG_\calB(\bQ)$ specializes to $\rmG_i(\bQ)$ in \eqref{G_i} with $\calB=i$.

Under the expected stopping time universality constraint, our sequential test $\Phi_\mathrm{seq}=(\tau,\phi)$ consists of a random stopping time and a decision rule. The stopping time $\tau$ is defined as follows:
\begin{align}\label{tau:exactlyT}
\tau:=\inf\{k\ge n-1: \exists~\calC\in\calS(T)~\mathrm{s.t.}~\rmS_\calC(\bx^k)\le f(k)\}.
\end{align}
where the scoring function $\rmS_\calC(\cdot)$ is defined as follows:
\begin{align}
\rmS_\calC(\bx^k)=\rmG_\calC(\hatT_{x_1^k},\ldots,\hatT_{x_M^k}),\label{def:scrore:sb}
\end{align}
and the threshold $f(k)=\frac{(M+1)|\calX|\log (k+1)}{k}$. The minimal stopping time is set as $n-1$ for a predefined integer $n\in\bbN$ to avoid early stopping. At the stopping time $\tau$, we apply the following decision rule
\begin{align}
\phi(\bx^\tau)=\argmin\limits_{\calC\in\calS(T)}\rmS_\calC(\bx^\tau).
\end{align}
In the following, we show our sequential test satisfies the expected stopping time universality constraint and characterize the exponential decay rate of the error probability.

\subsection{Main Results and Discussions}
Recall the definition of R\'enyi divergence $D_{\alpha}(P\|Q)$ in \eqref{def:renyi}.
Given any set $\calB\in\calS(T)$, for any pair of distributions $(P_\rmN,P_\rmA)\in\calP(\calX)^2$, define the following exponent function
\begin{align}
\mathrm{LD}_\calB(P_\rmN,P_\rmA,M,T)=\min_{t\in\bbN:t\le T-1}(T-t)\Big(D_{\frac{t}{T}}(P_\rmA||P_\rmN)+D_{\frac{M-2T+t}{M-T}}(P_\rmN||P_\rmA)\Big).\label{LDB:renyi}
\end{align}
Since R\'{e}nyi divergence is non-decreasing in its order (see the discussions below \eqref{renyi:variational}), the exponent function $\mathrm{LD}_\calB(P_\rmN,P_\rmA,M,T)$ is non-decreasing in $M$ since $\frac{M-2T+1}{M-T}=1-\frac{T-t}{M-T}$ increases in $M$. However, the dependence of $\mathrm{LD}_\calB(P_\rmN,P_\rmA,M,T)$ on the number of outliers $T$ is more complicated. As shown in Fig. \ref{LDB_T}, $\mathrm{LD}_\calB(P_\rmN,P_\rmA,M,T)$ is not a monotonic function of $T$.

\begin{figure}[htbp]
\centering
\includegraphics[width=.5\columnwidth]{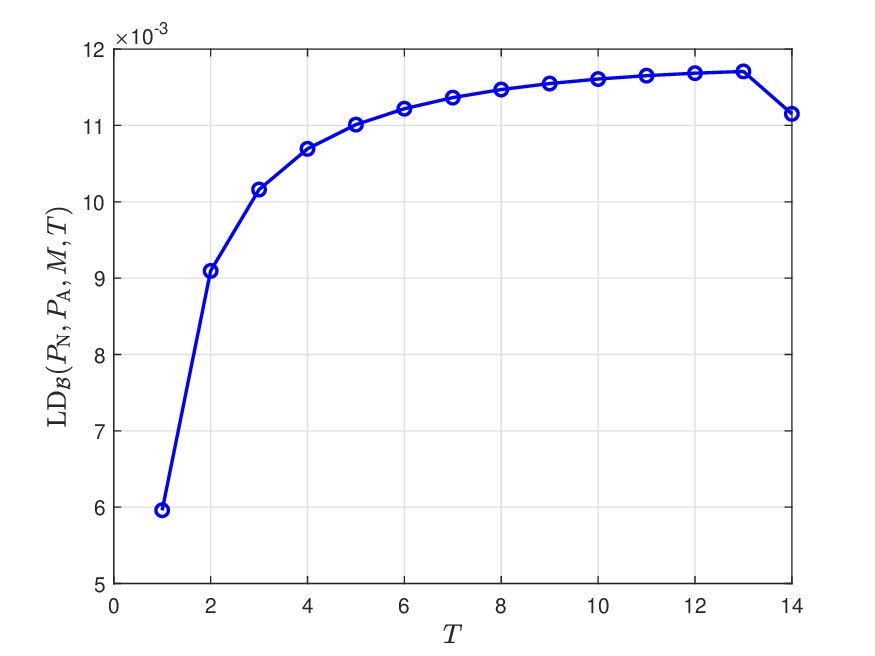}
\caption{Illustration the relationship of $\mathrm{LD}_\calB(P_\rmN,P_\rmA,M,T)$ in $T$ under generating distributions $(P_\rmN,P_\rmA)=\mathrm{Bern}(0.25,0.3)$ with $M=30$.}
\label{LDB_T}
\end{figure}

\begin{theorem}\label{seq_T_time}
Under any pair of distributions $(P_\rmN,P_\rmA)\in\calP(\calX)^2$ that are fully supported on the finite alphabet $\calX$, our sequential test satisfies the expected stopping time universality constraint and the error exponent of our test satisfies that for each $\calB\in\calS(T)$,
\begin{align}
E_\calB(\Phi_\mathrm{seq}|P_\rmN,P_\rmA)\ge \mathrm{LD}_\calB(P_\rmN,P_\rmA,M,T).
\end{align}
Conversely, for any sequential test $\Phi$ satisfying the expected stopping time universality constraint, under any pair of distributions $(P_\rmN,P_\rmA)$, the error exponent satisfies that for each $\calB\in\calS(T)$,
\begin{align}
E_\calB(\Phi|P_\rmN,P_\rmA)\le \mathrm{LD}_\calB(P_\rmN,P_\rmA,M,T).
\end{align}
\end{theorem}
The proof of Theorem \ref{seq_T_time} is similar to that of Theorem \ref{seq_stoptime} and thus, only differences are emphasized in Appendix \ref{proof:T}.

Recall that to construct fixed-length tests, Li, Nitinawarat and Veeravalli used the scoring function $\rmG_{\mathrm{Li},\calB}(\bQ)$~\cite[Eq. (37)]{li2014} when the number of outliers is known (cf. Theorem \ref{fixed:exactT}) and used the scoring function $\rmG_\calB(\bQ)$~\cite[Eq. (52)]{li2014} in \eqref{G_B} when the number of outliers is unknown. However, for sequential test, even when the number of outliers is known, we use the scoring function $\rmG_\calB(\bQ)$ instead of $\rmG_{\mathrm{Li},\calB}(\bQ)$. This is because using the scoring function $\rmG_\calB(\bQ)$ yields a larger and optimal error exponent. Specifically, if we use $\rmG_{\mathrm{Li},\calB}(\bQ)$, as shown in \cite{diao2024sequentialOHT}, the achievable error exponent becomes
\begin{align}
\min_{t\in[T-1]}(T-t)D_{\frac{M-2T+t}{M-T}}(P_\rmN||P_\rmA),
\end{align}
which is strictly less than $\mathrm{LD}_\calB(P_\rmN,P_\rmA,M,T)$ in Theorem \ref{seq_T_time} according to \eqref{LDB:renyi}.

We now discuss the influence of the number $M$ of observed sequences on the asymptotical performance. As $M$ increases, it follows from \eqref{LDB:renyi} that the achievable error exponent increases. This is consistent with our intuition because with more samples, the nominal distribution can be estimated more accurately. Then it is easier to identify the outlier. In the extreme case of $M\to\infty$, the exponent equals to the following value:
\begin{align}
\nn\min_{t\in[T-1]}(T-t)\big(D(P_\rmN||P_\rmA)+D_{\frac{t}{T}}(P_\rmA||P_\rmN)\big).
\end{align}

Finally, to clarify the advantage of the sequential test in Sec. \ref{test:exactlyT}, we compare the exponents in Theorem \ref{fixed:exactT} for the fixed-length test and \ref{seq_T_time} for our sequential test. Given any pair of distributions $(P,Q)\in\calP(\calX)^2$ and any set $\calB\in\calS(T)$, for any set $\calC\in\calS_\calB(T)$, the tuple of distributions $\bQ=(Q_1,\ldots,Q_M)$ with $Q_i=P$ for all $i\in\calC$ and $Q_j=Q$ for all $j\in\calM_\calC$ satisfies the constraint $\rmG_{\mathrm{Li},\calB}(\bQ)\ge\rmG_{\mathrm{Li},\calC}(\bQ)$.Thus, it follows from \eqref{fixed:exactT_error} and \eqref{LDB:renyi} that
\begin{align}\label{compare:exactT}
E_\calB(\Phi_{\rm LNV}|P_\rmN,P_\rmA)\le\mathrm{LD}_\calB(P_\rmN,P_\rmA,M,T)\le E_\calB(\Phi_\mathrm{seq}|P_\rmN,P_\rmA),
\end{align}
where the justification of \eqref{compare:exactT} is provided in Appendix \ref{proof:Tcompare}. In the following, we numerically verify that $E_\calB(\Phi_\mathrm{seq}|P_\rmN,P_\rmA)$ could be strictly greater than $E_\calB(\Phi_{\rm LNV}|P_\rmN,P_\rmA)$. Specifically, when the nominal and anomalous distributions are $(P_\rmN,P_\rmA)=\mathrm{Bern}(0.3,0.1)$, with $M=5$ and $T=2$, we have $E_\calB(\Phi_\mathrm{seq}|P_\rmN,P_\rmA)=0.0855>E_\calB(\Phi_{\rm LNV}|P_\rmN,P_\rmA)=0.0315$. Thus, our sequential test under the expected stopping time universality constraint could have strictly better performance than the fixed-length test as desired.

\section{Case of at most one Outlier}

\subsection{Problem Formulation}
In this section, it is assumed that there exists at most one outlier. The task is to design a test $\Phi=\{\tau,\phi\}$ that consists of a random stopping time $\tau$ and a decision rule $\phi:\calX^{M\tau}\to\{{\rmH}_1,{\rmH}_2,\ldots,{\rmH}_M,\mathrm{H_r}\}$ to classify among the following $M+1$ hypotheses:
\begin{itemize}
\item $\rmH_i,i\in[M]$:  the $i$-th sequence is the outlier.
\item $\mathrm{H_r}$: there is no outlier.
\end{itemize}
To evaluate the performance of a test, we consider the misclassification, false reject, false alarm probabilities and the expected stopping time of a sequential test. Specifically, the three error probabilities are defined as follows:
\begin{align}
\psi_i(\Phi|P_\rmN,P_\rmA)&:=\bbP_i\{\Phi(\bX^\tau)\neq\{\rmH_i,\mathrm{H_r}\}\},i\in[M],\\
\zeta_i(\Phi|P_\rmN,P_\rmA)&:=\bbP_i\{\Phi(\bX^\tau)=\mathrm{H_r}\},i\in[M],\\
\mathrm{P_{fa}}(\Phi|P_\rmN,P_\rmA)&:=\bbP_\rmr\{\Phi(\bX^\tau)\neq\mathrm{H_r}\},\label{def:mostone:fa}
\end{align}
where $\bbP_i$ is defined similarly as in \eqref{def:misclassify:p} and
we define $\bbP_\rmr(\cdot):=\Pr\{\cdot|\mathrm{H_r}\}$ to denote the joint distribution of observed sequences $\bX^\tau$, where for all $j\in[M]$, $X_j^\tau$ is generated i.i.d. from an unknown nominal distribution $P_\rmN$.  Consistent with the literature on outlier hypothesis testing~\cite{zhou2022second}, we define $\psi_i(\Phi|P_\rmN,P_\rmA)$ as the type-$i$ misclassification probability that bounds the probability of the event where the test falsely identifies a nominal sequence as an outlier, define $\zeta_i(\Phi|P_\rmN,P_\rmA)$ as the false reject probability that bounds the probability of the event where the test falsely claims no outlier when there exists one and define $\mathrm{P_{fa}}(\Phi|P_\rmN,P_\rmA)$ as the false alarm probability that bounds the probability of the event where the test falsely claims the existence of an outlier while all observed samples are nominal.

Furthermore, the expected stopping times under hypothesis $\rmH_i$ and $\mathrm{H_r}$ satisfy
\begin{align}
\mathbb{E}_i[\tau]&=\sum_{k=1}^{\infty}\bbP_i\{\tau> k\},\\
\mathbb{E}_\rmr[\tau]&=\sum_{k=1}^{\infty}\bbP_\rmr\{\tau> k\}.
\end{align}

\begin{definition}
A sequential test $\Phi$ is said to satisfy the expected stopping time universality constraint if there exists an integer $n\in\bbN$ such that for any pair of distributions $(P_\rmN,P_\rmA)\in\calP(\calX)^2$,
\begin{align}
\max\Big\{\max_{i\in[M]} \mathbb{E}_i[\tau],\mathbb{E}_\rmr[\tau]\Big\}\le n.
\end{align}
\end{definition}
For a sequential test satisfying the expected stopping time universality constraint, we are interested in the following exponents of the misclassification, false reject, false alarm probabilities:
\begin{align}
E_i(\Phi|P_\rmN,P_\rmA)&:=\liminf_{n\to\infty}\frac{-\log\psi_i(\Phi|P_\rmN,P_\rmA)}{n},~i\in[M]\\
E_{\mathrm{fr},i}(\Phi|P_\rmN,P_\rmA)&:=\liminf_{n\to\infty}\frac{-\log\zeta_i(\Phi|P_\rmN,P_\rmA)}{n},~i\in[M]\\
E_{\mathrm{fa}}(\Phi|P_\rmN,P_\rmA)&:=\liminf_{n\to\infty}\frac{-\log\mathrm{P_{fa}}(\Phi|P_\rmN,P_\rmA)}{n}.
\end{align}

\subsection{Existing Results}
To compare the performance of sequential tests and fixed-length tests, we recall the results of the fixed-length test $\Phi_{\rm ZWH}$ by Zhou, Wei and Hero~\cite[Eq. (5)]{zhou2022second}. Given $M$ observed sequences $\bx^n$ and any positive real number $\lambda\in\bbR_+$, Zhou's test in \cite[Eq. (5)]{zhou2022second} applies the following decision rule:
\begin{align}
\phi(\bx^n)=\left\{\!
\begin{array}{cl}
\rmH_i&\mathrm{if}~\rmS_i(\bx^n)\le \min_{j\in\calM_i}\rmS_j(\bx^n)~\mathrm{and}~\min_{j\in\calM_i}\rmS_j(\bx^n)>\lambda,\\
\rmH_\rmr&\mathrm{if}~\min_{j\in\calM_i}\rmS_j(\bx^n)\le\lambda,
\end{array}
\right.
\end{align}
where the scoring function $\rmS_i(\cdot)$ was defined in \eqref{def:scrore:si}.

To present their results, we need the following definition. Given any $\lambda\in\bbR_+$ and any pair of distributions $(P_\rmN,P_\rmA)\in\calP(\calX)^2$, for each $i\in[M]$, define the following exponent function
\begin{align}
\rmL_i(\lambda,P_\rmN,P_\rmA):=\min_{(j,k)\in\calM_{\rm dis}}\min_{\substack{\bQ\in\calP(\calX)^M:\\ \rmG_j(\bQ)\le\lambda,\rmG_k(\bQ)\le\lambda}}D(Q_i||P_\rmA)+\sum_{t\in\calM_i}D(Q_t||P_\rmN),
\end{align}
where $\calM_{\rm dis}:=\{(i,j)\in\calM^2:i\neq j\}$. Zhou, Wei and Hero derived the following result~\cite[Theorem 3]{zhou2022second}.
\begin{theorem}\label{zhou}
Given any $\lambda\in\bbR_+$ and any pair of distributions $(P_\rmN,P_\rmA)\in\calP(\calX)^2$ that are fully supported on the finite alphabet $\calX$, the achievable error exponents of $\Phi_{\rm ZWH}$ satisfies that for each $i\in[M]$,
\begin{align}
\min\{E_i(\Phi_{\rm ZWH}|P_\rmN,P_\rmA),E_{\mathrm{fa}}(\Phi_{\rm ZWH}|P_\rmN,P_\rmA)\}&\ge\lambda,\\
E_{\mathrm{fr},i}(\Phi_{\rm ZWH}|P_\rmN,P_\rmA)&\ge\rmL_i(\lambda,P_\rmN,P_\rmA).
\end{align}
\end{theorem}
Note that when $0<\lambda<\mathrm{GJS}(P_\rmN,P_\rmA,M-2)$, the false reject exponent $\rmL_i(\lambda,P_\rmN,P_\rmA)$ equals zero and the false reject probability vanishes asymptotically.

\subsection{Test Design and Intuition}
\label{test_most_one}

Our sequential test $\Phi_\mathrm{seq}=(\tau,\phi)$ consists of a random stopping time and a decision rule. Given any two positive real numbers $(\lambda_1,\lambda_2)\in\bbR_+^2$ such that $\lambda_2<\lambda_1$, the stopping time $\tau$ is defined as follows:
\begin{align}\label{mostone_stoptime}
\tau:=\inf\Big\{k\ge n-1: \exists~l\in[M]~\mathrm{s.t.}~\rmS_l(\bx^k)\le\lambda_2~\mathrm{and}~\min_{j\in\calM_l}\rmS_j(\bx^k)>\lambda_1,~\mathrm{or}~\forall~j\in[M],\rmS_j(\bx^k)\le\lambda_2\Big\}.
\end{align}
The minimal stopping time is set as $n-1$ for a predefined integer $n\in\bbN$ to avoid decision errors caused by early stopping. At the stopping time $\tau$, we apply the following decision rule
\begin{align}\label{mostone_decision}
\phi(\bx^\tau)=\left\{\!
\begin{array}{cl}
\rmH_i&\mathrm{if}~\rmS_i(\bx^\tau)\le\lambda_2~\mathrm{and}~\min_{j\in\calM_i}\rmS_j(\bx^\tau)>\lambda_1,\\
\rmH_\rmr&\mathrm{Otherwise}.
\end{array}
\right.
\end{align}

We now explain the asymptotic intuition why the above test works using the weak law of large numbers. Fix any $i\in[M]$ and consider hypothesis $\rmH_i$. As discussed below \eqref{ep:test}, under hypothesis ${\rm H}_i$, the scoring function $\rmS_i(\bx^k)$ tends to zero while for each $j\in\calM_i$, the scoring function $\rmS_j(\bx^k)$ tends to a positive real number as the sample size $k$ increases. Therefore, when $k$ is sufficiently large, when the $i$-th sequence is the outlier, our test stops and makes the correct decision $\rmH_i$ if $0<\lambda_2<\lambda_1$ and $\lambda_1$ is less than a positive value to be specified. Analogously, under the null hypothesis $\rmH_\rmr$, for each $j\in[M]$, the scoring function $\rmS_j(\bx^k)$ tends to zero as the sample size $k$ increases. Thus, the correct decision of the null hypothesis could be made asymptotically if $\lambda_2>0$. In the following, we show that our sequential test satisfies the expected stopping time universality constraint when $n$ is sufficiently large under a mild condition on the thresholds of the test, and we characterize the exponential decay rates of three error probabilities.

\subsection{Main Results and Discussions}
To present our results, we need the following error exponent function.
Given any $\lambda\in\bbR_+$ and any pair of distributions $(P_\rmN,P_\rmA)\in\calP(\calX)^2$, for each $i\in[M]$, define
\begin{align}\label{Omegai}
\Omega_i(\lambda,P_\rmN,P_\rmA):=\min_{j\in\calM_i}\min_{\substack{\bQ\in\calP(\calX)^M:\\ \rmG_j(\bQ)\le\lambda}}D(Q_i||P_\rmA)+\sum_{t\in\calM_i}D(Q_t||P_\rmN).
\end{align}
Note that $\Omega_i(\lambda,P_\rmN,P_\rmA)$ is non-increasing in $\lambda$. In particular, $\Omega_i(\lambda_,P_\rmN,P_\rmA)=0$ if $\lambda$ satisfies
\begin{align}
\lambda&\ge
\mathrm{GJS}(P_\rmN,P_\rmA,M-2).
\end{align}
When $\lambda=0$, $\Omega_i(\lambda,P_\rmN,P_\rmA)$ achieves the following maximum value:
\begin{align}\label{Omegai0}
\Omega_i(0,P_\rmN,P_\rmA)&=D_{\frac{M-2}{M-1}}(P_\rmN||P_\rmA).
\end{align}
As we shall show, $\Omega_i(\lambda,P_\rmN,P_\rmA)$ characterizes the false reject exponent of our test.

\begin{theorem}\label{at_most_one}
Under any pair of distributions $(P_\rmN,P_\rmA)\in\calP(\calX)^2$ that are fully supported on the finite alphabet $\calX$, given any pair of positive real numbers $(\lambda_1,\lambda_2)\in\bbR_+^2$ such that $\lambda_2<\lambda_1<\mathrm{GJS}(P_\rmN,P_\rmA,M-2)$, our sequential test satisfies the expected stopping time universality constraint when $n$ is sufficiently large and ensures that
\begin{itemize}
\item[1)] for each $i\in[M]$, the exponents of the misclassification and false reject probabilities satisfy
\begin{align}
E_i(\Phi_\mathrm{seq}|P_\rmN,P_\rmA)&\ge\lambda_1,\\*
E_{\mathrm{fr},i}(\Phi_\mathrm{seq}|P_\rmN,P_\rmA)&\ge\Omega_i(\lambda_2,P_\rmN,P_\rmA),
\end{align}
\item[2)] the exponent of the false alarm probability satisfies
\begin{align}
E_{\mathrm{fa}}(\Phi_\mathrm{seq}|P_\rmN,P_\rmA)\ge\lambda_1.
\end{align}
\end{itemize}
\end{theorem}
The proof of Theorem \ref{at_most_one} is provided in Appendix \ref{proof:at_most_one}. Firstly, we prove that the expected stopping times under each non-null hypothesis and the null hypothesis are upper bounded by $n$ when $n$ is sufficiently large under a mild condition. Subsequently, we use method of types to lower bound exponential decay rates of all three error probabilities.
Note that the parameters $\lambda_1$ and $\lambda_2$ are fixed a-priori, independent of unknown generating distributions $(P_\rmN,P_\rmA)$. Given chosen thresholds $(\lambda_1,\lambda_2)$, the theoretical guarantees hold if the unknown ground truth distributions $(P_\rmN,P_\rmA)$ satisfy $\lambda_2<\lambda_1<\mathrm{GJS}(P_\rmN,P_\rmA,M-2)$; otherwise, such guarantees can not be claimed.

Comparing Theorems \ref{seq_stoptime} and \ref{at_most_one}, we can obtain the penalty of not knowing whether the outlier exists on the performance of a sequential test. Note that in Theorem \ref{seq_stoptime}, it is known that one outlier exists while in Theorem \ref{at_most_one}, the outlier might either exist or not. For a fair comparison, we need to consider the error probability under each non-hull hypothesis. That is, we should compare the misclassification exponent $D_{\frac{M-2}{M-1}}(P_\rmN||P_\rmA)$ in Theorem \ref{seq_stoptime} with the minimal value of the misclassification exponent and the false reject exponent $\min\{\lambda_1,\Omega_i(\lambda_2,P_\rmN,P_\rmA)\}$ in Theorem \ref{at_most_one}. For any $\lambda_2>0$, it follows from \eqref{Omegai0} that
\begin{align}
\min\{\lambda_1,\Omega_i(\lambda_2,P_\rmN,P_\rmA)\}\le\Omega_i(\lambda_2,P_\rmN,P_\rmA)<D_{\frac{M-2}{M-1}}(P_\rmN||P_\rmA).
\end{align}
Thus, there is a penalty in the error exponent of sequential test when the existence of the outlier is unknown.

Finally, we demonstrate the advantage of the sequential test over the fixed-length test in terms of the achievable Bayesian error exponent. For the case of at most one outlier, the Bayesian error exponent is the smallest exponent among all exponents of misclassification, false reject and false alarm. For the fixed-length test, it follows from Theorem \ref{zhou} that the achievable Bayesian exponent satisfies
\begin{align}
E_\mathrm{Bayesian}(\Phi_{\rm ZWH}|P_\rmN,P_\rmA)=\max_{\lambda\in[0,\mathrm{GJS}(P_\rmN,P_\rmA,M-2)]}\min_{i\in[M]}\min\{\lambda,\rmL_i(\lambda,P_\rmN,P_\rmA)\}\label{bayesian:fix}.
\end{align}
For our sequential test, it follows from Theorem \ref{at_most_one} that the achievable Bayesian exponent satisfies
\begin{align}
E_\mathrm{Bayesian}(\Phi_\mathrm{seq}|P_\rmN,P_\rmA)&=\max_{\substack{(\lambda_1,\lambda_2)\in\bbR_+^2:\\\lambda_2<\lambda_1<\mathrm{GJS}(P_\rmN,P_\rmA,M-2)}}
\min_{i\in[M]}\min\{\lambda_1,\Omega_i(\lambda_2,P_\rmN,P_\rmA)\}\label{bayesian:sequential}.
\end{align}
For the fixed-length test, the threshold $\lambda$ tradeoffs the false reject exponent and the homogeneous misclassification and false alarm exponents. To maximize the term $\min\{\lambda,\rmL_i(\lambda,P_\rmN,P_\rmA)\}$, $\lambda$ should be chosen as a moderate value such that $\lambda=\rmL_i(\lambda,P_\rmN,P_\rmA)$.
In contrast, our sequential test resolves such a tradeoff since the misclassification and false alarm exponents are increasing in $\lambda_1$ and the false reject exponent is non-increasing in $\lambda_2$. This implies our sequential test has a larger Bayesian exponent by selecting an arbitrary $\lambda_2$ close to $0$ such that $\Omega_i(\lambda_2,P_\rmN,P_\rmA)>\lambda$ and by choosing an arbitrary $\lambda_1$ such that $\lambda_1>\lambda$, which leads to $E_\mathrm{Bayesian}(\Phi_\mathrm{seq}|P_\rmN,P_\rmA)>E_\mathrm{Bayesian}(\Phi_{\rm ZWH}|P_\rmN,P_\rmA)$. For example, when $(P_\rmN,P_\rmA)=\mathrm{Bern}(0.25,0.28)$ and $M=4$, the optimizer for \eqref{bayesian:fix} is $\lambda=4\times 10^{-4}$ and the optimizer for \eqref{bayesian:sequential} is $(\lambda_1,\lambda_2)=(9\times 10^{-4},10^{-4})$. Thus, $\Omega_i(\lambda_2,P_\rmN,P_\rmA)=8.54\times 10^{-4}$ and $\rmL_i(\lambda,P_\rmN,P_\rmA)=3.71\times 10^{-4}$. It follows that $E_\mathrm{Bayesian}(\Phi_\mathrm{seq}|P_\rmN,P_\rmA)=8.54\times 10^{-4}>E_\mathrm{Bayesian}(\Phi_{\rm ZWH}|P_\rmN,P_\rmA)=3.71\times 10^{-4}$. Thus, our sequential test can have strictly larger Bayesian exponent than the fixed-length test and demonstrates performance improvement as desired.

We present a numerical example to illustrate Theorem \ref{at_most_one} in Fig. \ref{simulation_miscla}. Consider the binary alphabet $\calX=\{0,1\}$. Set the number of sequences $M=4$ and generating distributions $(P_\rmN,P_\rmA)=\mathrm{Bern}(0.28,0.25)$. We choose the parameters as $(\lambda_1,\lambda_2)=(0.001,0.0005)$. We simulate the misclassification probability of our sequential test in \eqref{mostone_decision} versus the exponential estimates in Theorem \ref{at_most_one} when there is one outlier of $M$ sequences using $8 \times 10^4$ independent experiments. As observed, when the average sample size is large, our theoretical characterization for the misclassification probability provides a tight upper bound.

\begin{figure}[htbp]
\centering
\includegraphics[width=.5\columnwidth]{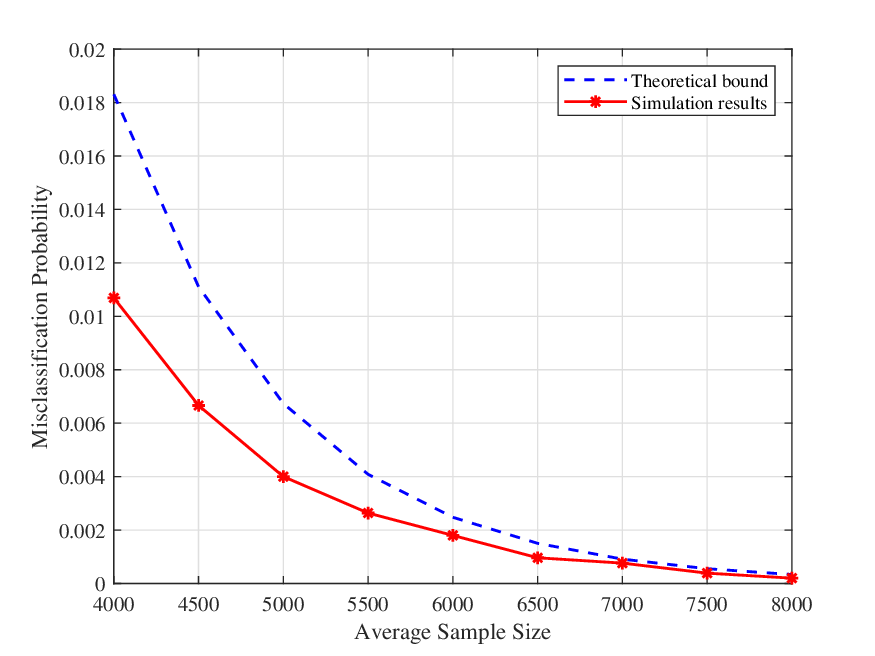}
\caption{Simulated misclassification probability of our test under expected stopping time universality constraint for the case with at most one outlier under generating distributions $(P_\rmN,P_\rmA)=\mathrm{Bern}(0.28,0.25)$ with $M=4$ with $(\lambda_1,\lambda_2)=(0.001,0.0005)$.}
\label{simulation_miscla}
\end{figure}

\section{Case of at most $T$ Outliers}
\subsection{Problem Formulation}
In this section, it is assumed that there are at most $T\le\lceil\frac{M}{2}-1\rceil$ identically distributed outliers, i.e., the number of outliers is unknown but upper bounded by $T$. For each $s\in[T]$, let $\calS(s):=\{\calB\subset[M]:|\calB|=s\}$ and let $\calS:=\bigcup_{s\in[T]}\calS(s)$.
The task is to design a test $\Phi=\{\tau,\phi\}$ that consists of the random stopping time $\tau$ and the decision rule $\phi:~\calX^{M\tau}\to\{\{\rmH_\calB\}_{\calB\in\calS},\mathrm{H_r}\}$ to classify among the following $|\calS|+1$ hypotheses:
\begin{itemize}
\item $\rmH_{\calB}$, $\calB\in\calS$: for each $j\in\calB$, the $j$-th sequence is an outlier.
\item $\mathrm{H_r}$: there is no outlier.
\end{itemize}

To evaluate the performance of a test, we consider the misclassification, false reject, false alarm probabilities and the expected stopping time of a sequential test. Specifically, the three error probabilities is defined as follows:
\begin{align}
\psi_\calB(\Phi|P_\rmN,P_\rmA)&:=\bbP_\calB\{\Phi(\bX^\tau)\neq\{\rmH_\calB,\mathrm{H_r}\}\},\calB\in\calS,\\
\zeta_\calB(\Phi|P_\rmN,P_\rmA)&:=\bbP_\calB\{\Phi(\bX^\tau)=\mathrm{H_r}\},\calB\in\calS,\\
\mathrm{P_{fa}}(\Phi|P_\rmN,P_\rmA)&:=\bbP_\rmr\{\Phi(\bX^\tau)\neq\mathrm{H_r}\},
\end{align}
where $\bbP_\calB$ is defined similarly as in \eqref{def:misclassify:pb} and $\bbP_\rmr$ is defined similarly as in \eqref{def:mostone:fa}.
Furthermore, the expected stopping times under hypothesis $\rmH_\calB$ and $\mathrm{H_r}$ satisfy
\begin{align}
\mathbb{E}_\calB[\tau]&=\sum_{k=1}^{\infty}\bbP_\calB\{\tau> k\},\\
\mathbb{E}_\rmr[\tau]&=\sum_{k=1}^{\infty}\bbP_\rmr\{\tau> k\}.
\end{align}

\begin{definition}
A sequential test $\Phi$ is said to satisfy the expected stopping time universality constraint if there exists an integer $n\in\bbN$ such that for any pair of distributions $(P_\rmN,P_\rmA)\in\calP(\calX)^2$,
\begin{align}
\max\Big\{\max_{\calB\in\calS} \mathbb{E}_\calB[\tau],\mathbb{E}_\rmr[\tau]\Big\}\le n.
\end{align}
\end{definition}
For a sequential test satisfying the expected stopping time universality constraint, we are interested in the following exponent of the misclassification, false reject, false alarm probabilities:
\begin{align}
E_\calB(\Phi|P_\rmN,P_\rmA)&:=\liminf_{n\to\infty}\frac{-\log\psi_\calB(\Phi|P_\rmN,P_\rmA)}{n},~\calB\in\calS,\\
E_{\mathrm{fr},\calB}(\Phi|P_\rmN,P_\rmA)&:=\liminf_{n\to\infty}\frac{-\log\zeta_\calB(\Phi|P_\rmN,P_\rmA)}{n},~\calB\in\calS,\\
E_{\mathrm{fa}}(\Phi|P_\rmN,P_\rmA)&:=\liminf_{n\to\infty}\frac{-\log\mathrm{P_{fa}}(\Phi|P_\rmN,P_\rmA)}{n}.
\end{align}

\subsection{Existing Results}
To compare the performance of sequential tests and fixed-length tests, we recall and revise the results of the fixed-length test $\Phi_{\rm ZWH}$ by Zhou, Wei and Hero~\cite[Eq. (43)]{zhou2022second} for the case with at most multiple outlier.
Given any $\calB\in\calS$, define $\calS_\calB:=\{\calC\in\calS:\calC\neq\calB\}$.
Given $M$ observed sequences $\bx^n$ and any positive real number $\lambda\in\bbR_+$, we revise Zhou's test in \cite[Eq. (43)]{zhou2022second} by applying the following decision rule:
\begin{align}\label{fixed:atmostT}
\phi(\bx^n)=\left\{\!
\begin{array}{cl}
\rmH_\calB&\mathrm{if}~\rmS_\calB(\bx^n)\le \min_{\calC\in\calS_\calB}\rmS_\calC(\bx^n)~\mathrm{and}~\min_{\calC\in\calS_\calB}\rmS_\calC(\bx^n)>\lambda,\\
\rmH_\rmr&\mathrm{if}~\min_{\calC\in\calS_\calB}\rmS_\calC(\bx^n)\le\lambda,
\end{array}
\right.
\end{align}
where the scoring function $\rmS_\calB(\cdot)$ was defined in \eqref{def:scrore:sb}. Note that Zhou, Wei and Hero used the scoring function $\rmG_{\mathrm{Li},\calB}(\bQ)$ in \cite[Eq. (43)]{zhou2022second} and we need to refine the fixed-length test by using $\rmG_\calB(\bQ)$. This is because $\rmG_{\mathrm{Li},\calB}(\bQ)$ is unable to deal with the case with unknown number of outliers.


To present the results, we need the following definition. Given any $\lambda\in\bbR_+$ and any pair of distributions $(P_\rmN,P_\rmA)\in\calP(\calX)^2$, for each $\calB\in\calS$, define the following exponent function
\begin{align}\label{LB}
\rmL_\calB(\lambda,P_\rmN,P_\rmA):=\min_{(\calC,\calD)\in\calS^2:\calC\neq\calD}\min_{\substack{\bQ\in\calP(\calX)^M:\\ \rmG_\calC(\bQ)\le\lambda,\rmG_\calD(\bQ)\le\lambda}}\sum_{i\in\calB}D(Q_i||P_\rmA)+\sum_{t\in\calM_\calB}D(Q_t||P_\rmN).
\end{align}
We slightly refine the results of Zhou, Wei and Hero~\cite[Theorem 6]{zhou2022second} as follows.
\begin{theorem}\label{zhou:T}
Given any $\lambda\in\bbR_+$ and any pair of distributions $(P_\rmN,P_\rmA)\in\calP(\calX)^2$ that are fully supported on the finite alphabet $\calX$, the achievable error exponents of $\Phi_{\rm ZWH}$ satisfies for each $\calB\in\calS$,
\begin{align}
\min\{E_\calB(\Phi_{\rm ZWH}|P_\rmN,P_\rmA),E_{\mathrm{fa}}(\Phi_{\rm ZWH}|P_\rmN,P_\rmA)\}&\ge\lambda,\\
E_{\mathrm{fr},\calB}(\Phi_{\rm ZWH}|P_\rmN,P_\rmA)&\ge\rmL_\calB(\lambda,P_\rmN,P_\rmA).
\end{align}
\end{theorem}
Let $\tilde{\bP}^M=(P_1,\ldots,P_M)$ be a tuple of $M$ distributions such that for each $i\in\calB,P_i=P_\rmA$ and for $i\in\calM_\calB,P_i=P_\rmN$. Furthermore, define the following function
\begin{align}\label{tillambda1}
\nn&\tilde\lambda_1(P_\rmN,P_\rmA)\\*
&:=\min_ {\calC\in\calS_\calB}\rmG_\calC(\tilde{\bP}^M)\\
\nn&=\min_ {\calC\in\calS_\calB}|\calB\cap\calM_\calC|D\Big(P_\rmA\Big\|\frac{|\calB\cap\calM_\calC|P_\rmA+|\calM_\calB\cap\calM_\calC|P_\rmN}{M-|\calC|}\Big)+
|\calM_\calB\cap\calM_\calC|D\Big(P_\rmN\Big\|\frac{|\calB\cap\calM_\calC|P_\rmA+|\calM_\calB\cap\calM_\calC|P_\rmN}{M-|\calC|}\Big)\\*
&\qquad+|\calB\cap\calC|D\Big(P_\rmA\Big\|\frac{|\calB\cap\calC|P_\rmA+|\calM_\calB\cap\calC|P_\rmN}{|\calC|}\Big)+
|\calM_\calB\cap\calC|D\Big(P_\rmN\Big\|\frac{|\calB\cap\calC|P_\rmA+|\calM_\calB\cap\calC|P_\rmN}{|\calC|}\Big).
\end{align}
When $T=1$, The function $\tilde\lambda_1(P_\rmN,P_\rmA)$ specializes to $\mathrm{GJS}(P_\rmN,P_\rmA,M-2)$ (cf. \eqref{GJS}). This is because when $T=1$, $|\calB\cap\calC|=0,~|\calB\cap\calM_\calC|=|\calM_\calB\cap\calC|=1,~|\calM_\calB\cap\calM_\calC|=M-2$ and thus,
\begin{align}
\tilde\lambda_1(P_\rmN,P_\rmA)=D\Big(P_\rmA\Big\|\frac{P_\rmA+(M-2)P_\rmN}{M-1}\Big)+(M-2)D\Big(P_\rmN\Big\|\frac{P_\rmA+(M-2)P_\rmN}{M-1}\Big)
=\mathrm{GJS}(P_\rmN,P_\rmA,M-2).
\end{align}
We remark that when $\lambda\ge\tilde\lambda_1(P_\rmN,P_\rmA)$, it follows from the definition of $\rmL_\calB(\lambda,P_\rmN,P_\rmA)$ in \eqref{LB} that the false reject exponent $\rmL_\calB(\lambda,P_\rmN,P_\rmA)$ equals zero and the false reject probability vanishes asymptotically.

\subsection{Test Design and Preliminaries}
\label{test_most_T}

Our sequential test $\Phi_\mathrm{seq}=(\tau,\phi)$ consists of the random stopping time and the decision rule.
Given $(\lambda_1,\lambda_2)\in\bbR_+^2$ such that $\lambda_2<\lambda_1$, the stopping time $\tau$ is defined as follows:
\begin{align}\label{mostT_stoptime}
\tau:=\inf\Big\{k\ge n-1: \exists~\calC\in\calS~\mathrm{s.t.}~\rmS_\calC(\bx^k)\le\lambda_2,~\mathrm{and}~\min_{\calD\in\calS_\calC}\rmS_\calC(\bx^k)>\lambda_1,~\mathrm{or}~
\forall~\calC\in\calS~\mathrm{s.t.}~\rmS_\calC(\bx^k)\le\lambda_2\Big\},
\end{align}
where the scoring function $\rmS_\calC(\cdot)$ was defined in \eqref{def:scrore:sb}.
The minimal stopping time is set as $n-1$ for a predefined integer $n\in\bbN$ to avoid early stopping. At the stopping time $\tau$, we apply the following decision rule
\begin{align}\label{mostT_decision}
\phi(\bx^\tau)=\left\{
\begin{array}{cl}
\rmH_\calB&\mathrm{if}~\rmS_\calB(\bx^k)\le\lambda_2,~\mathrm{and}~\min_{\calC\in\calS_\calB}\rmS_\calC(\bx^k)>\lambda_1,\\
\rmH_\rmr&\mathrm{Otherwise}.
\end{array}
\right.
\end{align}

In the following, we show our sequential test satisfies the expected stopping time universality constraint under mild conditions when $n$ is sufficiently large, and characterize the exponential decay rates of three error probabilities. To present our results, we need the following two error exponent functions. Fix any $\calB\in\calS$. Define
\begin{align}\label{OmegaB}
\Omega_\calB(\lambda,P_\rmN,P_\rmA):=\min_{\calC\in\calS_\calB}\min_{\substack{\bQ\in\calP(\calX)^M: \\ \rmG_\calC(\bQ)\le\lambda}}&\sum\limits_{j\in\calM_\calB}D(Q_j||P_\rmN)+\sum\limits_{i\in\calB}D(Q_i||P_\rmA).
\end{align}
As we shall show, $\Omega_\calB(\lambda,P_\rmN,P_\rmA)$ characterizes the achievable false reject exponent.

\begin{lemma}\label{properties}
The error exponent functions satisfies the following properties.
\begin{itemize}
\item[(i)] The exponent $\Omega_\calB(\lambda,P_\rmN,P_\rmA)$ are non-increasing in $\lambda$.
\item[(ii)] The function $\Omega_\calB(\lambda,P_\rmN,P_\rmA)=0$ if $\lambda$ satisfies that
\begin{align}
\lambda\ge\min_{\calC\in\calS_\calB}\rmG_\calC(\tilde{\bP}^M)=\tilde\lambda_1(P_\rmN,P_\rmA).
\end{align}
When $\lambda=0$, $\Omega_\calB(\lambda,P_\rmN,P_\rmA)$ achieves the following maximum value:
\begin{align}\label{OmegaB0}
\Omega_\calB(0,P_\rmN,P_\rmA)
=\mathrm{LD}_\calB(P_\rmN,P_\rmA,M,|\calB|).
\end{align}
\end{itemize}
\end{lemma}

\subsection{Main Results and Discussions}

\begin{theorem}\label{at_most_T}
Under any pair of distributions $(P_\rmN,P_\rmA)\in\calP(\calX)^2$ that are fully supported on the finite alphabet $\calX$, for any pair of positive real numbers $(\lambda_1,\lambda_2)\in\bbR_+^2$ such that $\lambda_2<\lambda_1<\tilde\lambda_1(P_\rmN,P_\rmA)$, our sequential test satisfies the expected stopping time universality constraint when $n$ is sufficiently large and ensures that
\begin{itemize}
\item[1)] for each $\calB\in\calS$, the exponents of the misclassification and false reject probabilities satisfy that
\begin{align}
E_\calB(\Phi_\mathrm{seq}|P_\rmN,P_\rmA)&\ge\lambda_1,\label{mostT:exponent}\\*
E_{\mathrm{fr},\calB}(\Phi_\mathrm{seq}|P_\rmN,P_\rmA)&\ge\Omega_\calB(\lambda_2,P_\rmN,P_\rmA),
\end{align}
\item[2)] the exponent of the false alarm probability satisfies
\begin{align}
E_{\mathrm{fa}}(\Phi_\mathrm{seq}|P_\rmN,P_\rmA)\ge\lambda_1.
\end{align}
\end{itemize}
\end{theorem}
The achievability proof is provided in Appendix \ref{proof:at_most_T}, which generalizes the proof of  Theorem \ref{at_most_one} for the case of at most one outlier. Both the misclassification and the false reject exponents are increasing in $\lambda_1$, and the false alarm exponents is non-increasing in $\lambda_2$. Note that when $\lambda_2>\tilde\lambda_1(P_\rmN,P_\rmA)$, both the misclassification and the false reject exponents equal $0$.

We next discuss the influence of $M$ and $T$ on the asymptotical performance. As $M$ increases, it follows from \eqref{OmegaB} that the size of set $\calM_\calB$ increases and thus, the achievable error exponent increases. This is consistent with our intuition because with more samples, the nominal distribution can be estimated more accurately and thus, it is easier to identify the outliers. On the other hand, as $T$ increases, it follows from \eqref{OmegaB} that the size of set $\calS_\calB$ increases and thus, the false reject exponent decreases. This is consistent with our intuition because with more outliers, the estimation of the nominal distribution is less accurate. Thus, it is more challenging to identify the outliers.
When $T=1$, the test in \eqref{mostT_decision} specializes to the test in \eqref{mostone_decision} for the case of at most one outlier and Theorem \ref{at_most_T} specializes to Theorem \ref{at_most_one}.

Comparing Theorems \ref{seq_T_time} and \ref{at_most_T}, we obtain the penalty of not knowing the number  of outliers on the performance of a sequential test. Recall that in Theorem \ref{seq_T_time}, it is known that $T$ outliers exist while in Theorem \ref{at_most_T}, the number of outliers can be any number from $0$ to $T$. For fair comparison, we should consider the error probability under each non-null hypothesis and thus compare the misclassification exponent $\mathrm{LD}_\calB(P_\rmN,P_\rmA,M,|\calB|)$ in Theorem \ref{seq_T_time} with the minimal value of the misclassification and the false reject exponents $\min\{\lambda_1,\Omega_\calB(\lambda_2,P_\rmN,P_\rmA)\}$ in Theorem \ref{at_most_T}. For any $\lambda_2>0$, it follows from the definition of $\Omega_\calB(\lambda,P_\rmN,P_\rmA)$ in \eqref{OmegaB} and the result in \eqref{OmegaB0} that
\begin{align}
\min\{\lambda_1,\Omega_\calB(\lambda_2,P_\rmN,P_\rmA)\}&\le\Omega_\calB(\lambda_2,P_\rmN,P_\rmA)<\mathrm{LD}_\calB(P_\rmN,P_\rmA,M,|\calB|).
\end{align}
Thus, there is a penalty in the exponents of sequential test when the number of outliers is unknown.

Finally, we demonstrate the advantage of the sequential test over the fixed-length test in terms of the achievable Bayesian error exponent. For the fixed-length test, it follows from Theorem \ref{zhou:T} that the achievable Bayesian exponent satisfies
\begin{align}\label{bayesian:fix_mostT}
E_\mathrm{Bayesian}(\Phi_{\rm ZWH}|P_\rmN,P_\rmA)=\max_{\lambda\in[0,\tilde\lambda_1(P_\rmN,P_\rmA)]}\min_{\calB\in\calS}\min\big\{\lambda,\rmL_\calB(\lambda,P_\rmN,P_\rmA)\big\}.
\end{align}
For our sequential test, it follows from Theorem \ref{at_most_T} that the achievable Bayesian exponent satisfies
\begin{align}\label{bayesian:sequential_mostT}
E_\mathrm{Bayesian}(\Phi_\mathrm{seq}|P_\rmN,P_\rmA)=\max_{\substack{(\lambda_1,\lambda_2)\in\bbR_+^2:\\\lambda_2<\lambda_1<\tilde\lambda_1(P_\rmN,P_\rmA)}}
\min_{\calB\in\calS}\min\big\{\lambda_1,\Omega_\calB(\lambda_2,P_\rmN,P_\rmA)\big\}.
\end{align}
For the fixed-length test, the threshold $\lambda$ tradeoffs the false reject exponent and the homogeneous misclassification and false alarm exponents. To maximize the term $\min\{\lambda,\rmL_\calB(\lambda,P_\rmN,P_\rmA)\}$, $\lambda$ should be chosen as a moderate value such that $\lambda=\rmL_\calB(\lambda,P_\rmN,P_\rmA)$.
In contrast, our sequential test resolves such a tradeoff since the misclassification and false reject exponents are non-increasing in $\lambda_2$ and the misclassification and false alarm exponent are increasing in $\lambda_1$. This implies our sequential test has a larger Bayesian exponent by selecting an arbitrary $\lambda_2$ close to $0$ such that $\Omega_\calB(\lambda_2,P_\rmN,P_\rmA)>\lambda$ and an arbitrary $\lambda_1$ such that $\lambda_1>\lambda$, which leads to $E_\mathrm{Bayesian}(\Phi_\mathrm{seq}|P_\rmN,P_\rmA)>E_\mathrm{Bayesian}(\Phi_{\rm ZWH}|P_\rmN,P_\rmA)$. For example, when $(P_\rmN,P_\rmA)=\mathrm{Bern}(0.3,0.1)$ and $M=5$, the optimizer for \eqref{bayesian:fix_mostT} is $\lambda=0.021$ and the optimizer for \eqref{bayesian:sequential_mostT} is $(\lambda_1,\lambda_2)=(0.09,0.001)$. Thus, $\Omega_\calB(\lambda_2,P_\rmN,P_\rmA)=0.0748$, and $\rmL_\calB(\lambda,P_\rmN,P_\rmA)=0.0235$. It follows that $E_\mathrm{Bayesian}(\Phi_\mathrm{seq}|P_\rmN,P_\rmA)=0.0748>E_\mathrm{Bayesian}(\Phi_{\rm ZWH}|P_\rmN,P_\rmA)=0.021$. Thus, our sequential test can have strictly larger Bayesian exponent than the fixed-length test and demonstrates performance improvement as desired.

\section{Conclusion}
We revisited sequential outlier hypothesis testing when both the nominal and anomalous distributions are unknown. In particular, we proposed tests under universality constraints and derived bounds on the achievable error exponents. For the case of exactly one outlier, our results are tight and strengthened a previous result of \cite[Theorem 3.2]{li2017universal} by having a matching converse result and analytically demonstrating the advantage of our sequential test. We also generalized our results to the case of exactly multiple outliers and derived tight results. For the case of at most one outlier, we strengthened \cite[Theorem 3.2]{li2017universal} by proposing a sequential test, analyzing the achievable error exponents and showing that our sequential test has a larger Bayesian error exponent than the fixed-length test in \cite[Eq. (5)]{zhou2022second}. Furthermore, our sequential test resolved the tradeoff among the error probabilities of the fixed-length test in \cite[Eq. (5)]{zhou2022second}. Finally, we generalized our results to the case of at most multiple outliers and theoretically showed that there is a penalty in the misclassification exponents when the number of outliers is unknown.

There are several avenues for future research. Firstly, we only derived achievability results for sequential tests when the number of outliers is unknown. It would be worthwhile to derive matching converse results. To do so, one might combine the generalized Neyman-Pearson criterion \cite{zhou2020second} with the converse proof techniques for sequential tests when the number of outliers is known. Secondly, we derived the asymptotical performance in terms of error exponents. However, any practical application would only provide sequences of a limited length. Thus, it would be valuable to characterize the non-asymptotic performance for optimal sequential tests by generalizing the ideas in \cite{li2020second} for binary hypothesis testing. Thirdly, our results were established under the assumption of a discrete finite alphabet. In practical applications, the observed sequence could take values in a continuous alphabet. It would be rewarding to generalize our results to account for sequences generated from unknown probability density functions, potentially using the kernel methods in~\cite{gretton2012kernel,zhu2024exponentially}.
Fourthly, our sequential tests have exponential complexity with respect to the number of outliers. For practical applications, it is of great interest to propose low complexity tests that achieve theoretical benchmarks derived in this paper. Novel ideas such as clustering~\cite{bu2019linear,xiong2011group} and two-phase tests~\cite{bai2022achievable,diao2023achievable,diao2023classification} might be helpful. Finally, we assumed that the nominal samples are generated from the same nominal distribution and the outliers are generated from the same anomalous distribution. It would be beneficial to study the impact of distribution uncertainty on the performance of sequential tests. To do so, the proof techniques~\cite{hsu2020binary,pan2022asymptotics,boroumand2022mismatched} could be helpful.


\appendix

\subsection{Proof of Theorem \ref{seq_errorprob}}\label{proof:ep}

\subsubsection{Achievability Proof of Theorem \ref{seq_errorprob}}
\label{proof:ep_ach}

\subsubsubsection{Error Probability Universality Constraint}
We first show that our sequential test in Sec. \ref{test_errorprob} satisfies the error probability universality constraint. Fix any $\beta\in(0,1)$ and any pair of distributions $(P_\rmN,P_\rmA)$. For each $i\in[M]$, the misclassification probability under hypothesis $\rmH_i$ satisfies
\begin{align}
\psi_i(\phi^\mathrm{EP}|P_\rmN,P_\rmA)&=\bbP_i\{\phi(\bX^{\tau^\mathrm{EP}})\neq\rmH_i\}\\
&\le\sum_{k=1}^{\infty}\bbP_i\{\phi(\bX^k)\neq\rmH_i\}\label{random}\\
&\le\sum_{k=1}^{\infty}\bbP_i\big\{\rmS_i(\bX^k)>g(\beta,k)\big\}\label{ep_test}\\
&=\sum_{k=1}^{\infty}\sum_{\bx^k\in\calX^{Mk}: \rmS_i(\bx^k)>g(\beta,k)}P_\rmA(x^k_i)\times \Big(\prod_{j\in\calM_i}P_\rmN(x^k_j)\Big)\\
&=\sum_{k=1}^{\infty}\sum_{\bQ\in\calP_k(\calX)^M: \rmG_i(\bQ)>g(\beta,k)}P_\rmA(\calT_{Q_i}^k)\times \Big(\prod_{j\in\calM_i}P_\rmN(\calT_{Q_j}^k)\Big)\label{ep:type}\\
&\le\sum_{k=1}^{\infty}\sum_{\substack{\bQ\in\calP_k(\calX)^M:\\ \rmG_i(\bQ)>g(\beta,k)}}\exp\Big\{-k\Big(D(Q_i||P_\rmA)+\sum_{t\in\calM_i}D(Q_t||P_\rmN)\Big)\Big\}\label{upperbound4}\\
&\le\sum_{k=1}^{\infty}\sum_{\substack{\bQ\in\calP_k(\calX)^M:\\ \rmG_i(\bQ)>g(\beta,k)}}\exp\Big\{-k\Big(D(Q_i||P_\rmA)+(M-1)D\Big(\frac{\sum_{t\in\calM_i}Q_t}{M-1}\Big\|P_\rmN\Big)+\rmG_i(\bQ)\Big)\Big\}\label{equality}\\
&\le\sum_{k=1}^{\infty}\sum_{\substack{\bQ\in\calP_k(\calX)^M:\\ \rmG_i(\bQ)>g(\beta,k)}}\exp\Big\{-k\Big(D(Q_i||P_\rmA)+(M-1)D\Big(\frac{\sum_{t\in\calM_i}Q_t}{M-1}\Big\|P_\rmN\Big)+g(\beta,k)\Big)\Big\}\label{Gge}\\
&\le\sum_{k=1}^{\infty}\sum_{\substack{\bQ\in\calP_k(\calX)^M:\\ \rmG_i(\bQ)>g(\beta,k)}}\exp\{-kg(\beta,k)\}\label{nonneg1}\\
&\le\sum_{k=1}^{\infty}(k+1)^{M|\calX|}\exp\{-kg(\beta,k)\}\label{knumber}\\
&=\sum_{k=1}^{\infty}(k+1)^{M|\calX|}\beta(|\calX|-1)(k+1)^{-(M+1)|\calX|}\label{gk}\\
&\le\beta(|\calX|-1)\sum_{k=1}^{\infty}(k+1)^{-|\calX|}\\
&\le\beta(|\calX|-1)\int_{0}^{\infty}(u+1)^{-|\calX|}\mathrm{d}u\label{continuous}\\
&=\beta(|\calX|-1)\frac{1}{-|\calX|+1}(u+1)^{-|\calX|+1}\Big|_{u=0}^{u=\infty}\\
&=\beta,\label{beta_one}
\end{align}
where \eqref{random} follows from that stopping time $\tau^\mathrm{EP}$ is a random variable ranging from $1$ to infinity, \eqref{ep_test} follows from our sequential test in \eqref{ep:test}, \eqref{ep:type} follows from the definitions of $\rmG_i(\cdot)$ in \eqref{G_i} and $\rmS_i(\cdot)$ in Sec. \ref{test_errorprob}, \eqref{upperbound4} follows from the upper bound on the probability of the type class~\cite[Theorem 11.1.4]{cover2012elements}, \eqref{equality} follows from the following equation~\cite[Eq. (86)]{zhou2022second}:
\begin{align}\label{equation1}
\sum_{t\in\calM_i}D(Q_t||P_\rmN)=(M-1)D\bigg(\frac{\sum_{t\in\calM_i}Q_t}{M-1}\bigg\|P_\rmN\bigg)+\rmG_i(\bQ),
\end{align}
\eqref{Gge} follows from the constraint that $\rmG_i(\bQ)>g(\beta,k)$, \eqref{nonneg1} follows from the fact that KL divergence is non-negative,
\eqref{knumber} follows from \cite[Theorem 11.1.1]{cover2012elements} which implies that $|\calP_k(\calX)|\le(k+1)^{|\calX|}$, \eqref{gk} follows from the definition of $g(\beta,k)$ in \eqref{gbeta} and \eqref{continuous} follows from similarly to~\cite[Appendix C: 1)]{Ihwang2022sequential}.

\subsubsubsection{Achievable Misclassification Exponent}\label{ach_err}
Fix any $i\in[M]$. We now bound the misclassification exponent $E_i^\mathrm{EP}(\Phi|P_\rmN,P_\rmA)=\liminf_{\beta\to 0}\frac{-\log\beta}{\mathbb{E}_i[\tau^\mathrm{EP}]}$. To do so, we first define an upper bound $\tau_i$ on $\tau^\mathrm{EP}$ and show that $-\frac{\tau_i}{\log\beta}$ converges in probability. Subsequently, we demonstrate how the convergence in probability result leads to the desired convergence in mean result $-\frac{\log\beta}{\mathbb{E}[\tau_i]}$. Finally, an upper bound on $E_i^\mathrm{EP}(\Phi|P_\rmN,P_\rmA)$ is obtained by using the fact that $\tau^\mathrm{EP}\leq \tau_i$.

Recall the definitions of $\rmS_j(\cdot)$ in \eqref{def:scrore:si} and define the stopping time
\begin{align}\label{def:taui}
\tau_i:=\inf\big\{k\in\bbN:\forall~j\in\calM_i,~\rmS_j(\bx^k)>g(\beta,k)\big\}.
\end{align}

\begin{lemma}
\label{converge:p}
The random variable $-\frac{\tau_i}{\log\beta}$ converges to $\frac{1}{\mathrm{GJS}(P_\rmN,P_\rmA,M-2)}$ in probability.
\end{lemma}

\begin{proof}
We first show that $\tau_i\to\infty$ as $\beta\to 0$. Given integer $\tilk\in\bbN$, it follows that
\begin{align}
\bbP_i\{\tau_i\le \tilk\}&=\bbP_i\big\{\forall~j\in\calM_i,\rmS_j(\bX^{\tau_i})>g(\beta,\tau_i),~\tau_i\le \tilk\big\}+\bbP_i\big\{\exists~j\in\calM_i,\rmS_j(\bX^{\tau_i})\le g(\beta,\tau_i),~\tau_i\le \tilk\big\}\\
&=\bbP_i\big\{\forall~j\in\calM_i,\rmS_j(\bX^{\tau_i})>g(\beta,\tau_i),~\tau_i\le \tilk\big\}\label{deftau}\\
&\le\bbP_i\big\{\forall~j\in\calM_i,\tilk\rmS_j(\bX^{\tau_i})>-\log\big(\beta(|\calX|-1)\big)\big\}\label{gtaui}\\
&\le\mathbb{I}\big\{\tilk M\log(M-1)>-\log\big(\beta(|\calX|-1)\big)\big\},\label{logM}
\end{align}
where \eqref{deftau} follows from the definition of $\tau_i$ in \eqref{def:taui}, which implies that $\bbP_i\big\{\exists~j\in\calM_i,\rmS_j(\bx^{\tau_i})\le g(\beta,\tau_i),\tau_i\le \tilk\big\}=0$, \eqref{gtaui} follows from the definition of $g(\beta,k)$ in \eqref{gbeta} which implies $g(\beta,\tau_i)\ge g(\beta,\tilk)\ge\frac{-\log(\beta(|\calX|-1))}{\tilk}$ and \eqref{logM} follows from the definition of $\rmS_j(\cdot)$ in Sec. \ref{test_errorprob} and the fact that for all $j\in\calM_i$,
\begin{align}
\sum\limits_{t\in\calM_j}D\bigg(Q_t\bigg\|\frac{\sum_{l\in\calM_j}Q_l}{M-1}\bigg)&\le \sum\limits_{t\in\calM_j}\sum_{x\in\calX}Q_t(x)\log\bigg((M-1)\frac{Q_t(x)}{\sum_{l\in\calM_j}Q_l(x)}\bigg)\\
&\le\sum\limits_{t\in\calM_j}\sum_{x\in\calX}Q_t(x)\log(M-1)\\
&\le M\log(M-1).
\end{align}
Thus, it follows from \eqref{logM} that $\bbP_i\{\tau_i\le \tilk\}=0$ if $\tilk<\frac{-\log\big(\beta(|\calX|-1)\big)}{M\log(M-1)}$. As a result, when $\beta\to 0$, $\frac{-\log(\beta(|\calX|-1))}{M\log(M-1)}\to\infty$ and thus $\tau_i\to\infty$.

We next show that $\rmS_j(\bx^{\tau_i})$ converges to $\mathrm{GJS}(P_\rmN,P_\rmA,M-2)$ as $\tau_i\to\infty$. It follows from the weak law of large numbers that under hypothesis $\rmH_i$, $\hatT_{x_i^{\tau_i}}\to P_\rmA$ and for each $t\in\calM_i$, $\hatT_{x_t^{\tau_i}}\to P_\rmN$. Using the continuity of KL-divergence, when $\tau_i\to\infty$, for each $j\in\calM_i$,
\begin{align}
\rmS_j(\bx^{\tau_i})&\to D\bigg(P_\rmA\bigg\|\frac{P_\rmA+(M-2)P_\rmN}{M-1}\bigg)+(M-2)D\bigg(P_\rmN\bigg\|\frac{P_\rmA+(M-2)P_\rmN}{M-1}\bigg)\\
&=\mathrm{GJS}(P_\rmN,P_\rmA,M-2).\label{SGJS}
\end{align}

Finally, we show that $-\frac{\log \beta}{\tau_i}$ converges in probability to the desired value using the definition of $\tau_i$ and the above two claims. It follows from the definition of $g(\beta,k)$ in \eqref{gbeta} that both $g(\beta,k)$ and $g(\beta,k-1)$ tend to $-\frac{\log\beta}{k}$ when $k\to\infty$ and $\beta\to 0$. Recall the definition of the stopping time $\tau_i$ in \eqref{def:taui}, we have
\begin{align}
\min_{j\in\calM_i}\rmS_j(\bx^{\tau_i})&>g(\beta,\tau_i),\label{ge}\\
\min_{j\in\calM_i}\rmS_j(\bx^{\tau_i-1})&\le g(\beta,\tau_i-1)\label{le}.
\end{align}
Since $\tau_i\to\infty$ as $\beta\to 0$ (our first claim in the proof), it follows that
\begin{align}
g(\beta,\tau_i)&\to-\frac{\log\beta}{\tau_i}\label{taui},\\
g(\beta,\tau_i-1)&\to-\frac{\log\beta}{\tau_i}\label{taui-1}.
\end{align}
Combining \eqref{SGJS}, \eqref{le} and \eqref{taui-1}, we have
\begin{align}
\lim\limits_{\beta\to 0}\bbP_i\Big\{-\frac{\log\beta}{\tau_i}\ge\mathrm{GJS}(P_\rmN,P_\rmA,M-2)\Big\}\to 1.
\end{align}
Analogously, combining \eqref{SGJS}, \eqref{ge} and \eqref{taui}, we have
\begin{align}
\lim\limits_{\beta\to 0}\bbP_i\Big\{-\frac{\log\beta}{\tau_i}\le\mathrm{GJS}(P_\rmN,P_\rmA,M-2)\Big\}\to 1.
\end{align}
Consequently, for any $\varepsilon>0$,
\begin{align}\label{prob_conver}
\lim\limits_{\beta\to 0}\bbP_i\Big\{\Big|-\frac{\tau_i}{\log\beta}-\frac{1}{\mathrm{GJS}(P_\rmN,P_\rmA,M-2)}\Big|\ge\varepsilon\Big\}=0.
\end{align}
\end{proof}

In order to show that $-\frac{\log \beta}{\tau_i}$ converges in mean using Lemma \ref{converge:p}, we need the following lemma, which helps prove that the sequence of random variables $-\frac{\tau_i}{\log\beta}$ is uniformly integrable as $\beta\to 0$.
\begin{lemma}\label{tauj_ge}
When $k$ is sufficiently large, there exists a positive real number $c\in\bbR_+$ such that $\bbP_i\{\tau_i\ge k\}\le\frac{1}{\beta}\exp\{-ck\}$.
\end{lemma}
\begin{proof}

For any positive real number $\alpha\in\bbR_+$, it follows from the definition of $\tau_i$ in \eqref{def:taui} that
\begin{align}
\bbP_i\{\tau_i\ge k\}
&\le\bbP_i\Big\{\exists~j\in\calM_i,~\rmS_j(\bX^{k-1})\le g(\beta,k-1)\Big\}\label{tauidef}\\
&\le\sum\limits_{j\in\calM_i}\bbP_i\Big\{\rmS_j(\bX^{k-1})\le g(\beta,k-1)\Big\}\\
&\le\sum\limits_{j\in\calM_i}\bbP_i\Big\{\rmS_j(\bX^{k-1})\le g(\beta,k-1)+\rmS_i(\bX^{k-1})\Big\}\label{nonneg}\\
&\le\sum\limits_{j\in\calM_i}\bbP_i\Big\{\rmS_j(\bX^{k-1})\le g(\beta,k-1)+\rmS_i(\bX^{k-1}),~\rmS_j(\bX^{k-1})\ge\alpha\Big\}+\sum\limits_{j\in\calM_i}\bbP_i\big\{\rmS_j(\bX^{k-1})<\alpha\big\},\label{twoalpha}
\end{align}
where \eqref{nonneg} follows from the definition of $\rmS_i(\cdot)$ in \eqref{def:scrore:si} that is a linear combination of KL divergence, which indicates $\rmS_i(\bx^k)\ge 0$.

When $0<\alpha<\mathrm{GJS}(P_\rmN,P_\rmA,M-2)$, the first term of \eqref{twoalpha} can be upper bounded as follows:
\begin{align}
\nn&\sum\limits_{j\in\calM_i}\bbP_i\Big\{\rmS_j(\bX^{k-1})\le g(\beta,k-1)+\rmS_i(\bX^{k-1}),~\rmS_j(\bX^{k-1})\ge\alpha\Big\}\\*
&\le(M-1)\bbP_i\Big\{\alpha\le g(\beta,k-1)+ \rmS_i(\bX^{k-1})\Big\}\label{Sj}\\
&\le(M-1)\bbP_i\Big\{\rmS_i(\bX^{k-1})\ge\alpha-g(\beta,k-1)\Big\}\\
&\le(M-1)k^{M|\calX|}\exp\big\{-(k-1)\big(\alpha-g(\beta,k-1)\big)\big\}\label{Si}\\
&\le\frac{M-1}{\beta(|\calX|-1)}k^{(2M+1)|\calX|}\exp\big\{-(k-1)\alpha\big\},\label{alpha}
\end{align}
where \eqref{Sj} follows since $\rmS_j(\bx^{k-1})\ge\alpha$ for sufficiently large $k$ and \eqref{Si} follows analogously to steps leading to \eqref{knumber} and \eqref{alpha} follows from the definition of $g(\beta,k-1)$ in \eqref{gbeta}.

Given any pair of distributions $(P_\rmN,P_\rmA)\in\calP(\calX)^2$ and any $\alpha\in\bbR_+$, define the following function:
\begin{align}
\calH_i(P_\rmN,P_\rmA,\alpha):=\min_{j\in\calM_i}\min_{\substack{\bQ\in\calP(\calX)^M: \\\rmG_j(\bQ)<\alpha}}D(Q_i||P_\rmA)+\sum_{t\in\calM_i}D(Q_t||P_\rmN),
\end{align}
which is strictly positive when $0<\alpha<\mathrm{GJS}(P_\rmN,P_\rmA,M-2)$.
The second term of \eqref{twoalpha} can be upper bounded as follows:
\begin{align}
\sum\limits_{j\in\calM_i}\bbP_i\big\{\rmS_j(\bX^{k-1})<\alpha\big\}&\le(M-1)\max_{j\in\calM_i}\bbP_i\big\{\rmS_j(\bX^{k-1})<\alpha\big\}\\
&\le(M-1)\max_{j\in\calM_i}\sum_{\substack{\bQ\in\calP_k(\calX)^M:\\ \rmG_j(\bQ)<\alpha}}\exp\Big\{-(k-1)\Big(D(Q_i||P_\rmA)+\sum_{t\in\calM_i}D(Q_t||P_\rmN)\Big)\Big\}\label{upperbound6}\\
&\le(M-1)k^{M|\calX|}\exp\big\{-(k-1)\calH_i(P_\rmN,P_\rmA,\alpha)\big\}\label{alpha2}\\
&\le\frac{M-1}{\beta}k^{M|\calX|}\exp\big\{-(k-1)\calH_i(P_\rmN,P_\rmA,\alpha)\big\},\label{fracbeta}
\end{align}
where \eqref{upperbound6} follows from similarly to \eqref{upperbound4} and \eqref{alpha2} follows from the fact that the size of the set of all types of length $n$ satisfies $|\calP_k(\calX)|\le(k+1)^{|\calX|}$~\cite[Theorem 11.1.1]{cover2012elements}.

Finally, combining \eqref{alpha} and \eqref{fracbeta}, we conclude that when $0<\alpha<\mathrm{GJS}(P_\rmN,P_\rmA,M-2)$ and $k$ is sufficiently large, there exists a positive constant $0<c<\min\{\alpha,\calH_i(P_\rmN,P_\rmA,\alpha)\}$ such that
\begin{align}
\bbP_i\{\tau_i\ge k\}\le\frac{1}{\beta}\exp\{-ck\}.
\end{align}
\end{proof}

\begin{lemma}\label{lemma:uni_int}
The sequence of random variables $\{-\frac{\tau_i}{\log\beta}\}_{\beta\in (0,1)}$ is uniformly integrable.
\begin{proof}
Fix $c$ in Lemma \ref{tauj_ge}. Given a sufficiently large positive real number $v$ such that $cv-1>0$, it follows that
\begin{align}
\mathbb{E}_i\bigg[-\frac{\tau_i}{\log\beta}\mathbb{I}\Big\{-\frac{\tau_i}{\log\beta}\ge v\Big\}\bigg]
&=-\frac{1}{\log\beta}\mathbb{E}_i\Big[\tau_i\mathbb{I}\big\{\tau_i\ge -v\log\beta\big\}\Big]\\
&=-\frac{1}{\log\beta}\sum_{k=1}^{\infty}\bbP_i\Big\{\tau_i\mathbb{I}\big\{\tau_i\ge -v\log\beta\big\}\ge k\Big\}\\
&\le-\frac{1}{\log\beta}(-v\log\beta)\bbP_i\big\{\tau_i\ge -v\log\beta\big\}-\frac{1}{\log\beta}\sum_{k=\lceil-v\log\beta\rceil}^{\infty}\bbP_i\big\{\tau_i\ge k\big\}\label{kgev}\\
&\le\frac{v}{\beta}\exp\{cv\log\beta\}-\frac{1}{\beta\log\beta}\sum_{k=\lceil-v\log\beta\rceil}^{\infty}\exp\{-ck\}\label{lemma}\\
&\le v\beta^{cv-1}-\frac{1}{\beta\log\beta}\cdot\frac{\beta^{cv}}{1-\exp\{-c\}}\\
&\le\beta^{cv-1}\Big(v-\frac{1}{\log\beta(1-\exp\{-c\})}\Big)
\end{align}
where \eqref{kgev} follows from considering two cases when $k<\lceil-v\log\beta\rceil$ and $k\ge\lceil-v\log\beta\rceil$ and \eqref{lemma} follows from Lemma \ref{tauj_ge} when $v$ is sufficiently large. Since $\beta\in(0,1)$, when $v$ is sufficiently large such that $cv-1>0$, for any $\varepsilon>0$, we have
\begin{align}
\mathbb{E}_i\bigg[-\frac{\tau_i}{\log\beta}\mathbb{I}\Big\{-\frac{\tau_i}{\log\beta}\ge v\Big\}\ge\varepsilon\bigg]\to0.
\end{align}
Furthermore, it follows from the definition of $\tau_i$ in \eqref{def:taui} that $\tau_i>0$ and thus,
\begin{align}
\mathbb{E}_i\bigg[-\frac{\tau_i}{\log\beta}\mathbb{I}\Big\{\frac{\tau_i}{\log\beta}\ge v\Big\}\ge\varepsilon\bigg]=-\frac{1}{\log\beta}\mathbb{E}_i\Big[\tau_i\mathbb{I}\big\{\tau_i\le v\log\beta\big\}\ge\varepsilon\Big]\to0.
\end{align}
Therefore, for any $\varepsilon>0$, when $v$ is sufficiently large such that $cv-1>0$,
\begin{align}\label{uni_inte}
\mathbb{E}_i\bigg[\Big|-\frac{\tau_i}{\log\beta}\Big|\mathbb{I}\Big\{\Big|-\frac{\tau_i}{\log\beta}\Big|\ge v\Big\}\ge\varepsilon\bigg]\to0,
\end{align}
which indicates the sequence of random variables $-\frac{\tau_i}{\log\beta}$ are uniformly integrable.
\end{proof}
\end{lemma}

Using Lemmas \ref{converge:p} to \eqref{uni_inte}, we obtain the desired convergence in mean result.

\begin{lemma}
\label{converge:mean}
The random variable $-\frac{\tau_i}{\log\beta}$ converges to $\frac{1}{\mathrm{GJS}(P_\rmN,P_\rmA,M-2)}$ in mean, i.e.,
\begin{align}\label{converge_mean}
\lim\limits_{\beta\to 0}\mathbb{E}_i\bigg[\Big|-\frac{\tau_i}{\log\beta}-\frac{1}{\mathrm{GJS}(P_\rmN,P_\rmA,M-2)}\Big|\bigg]=0.
\end{align}
\end{lemma}
\begin{proof}
Fix any $(\varepsilon_1,\varepsilon_2)\in\bbR^2$ such that $\varepsilon_1<\varepsilon_2$. It follows that
\begin{align}
\nn&\mathbb{E}_i\bigg[\Big|-\frac{\tau_i}{\log\beta}-\frac{1}{\mathrm{GJS}(P_\rmN,P_\rmA,M-2)}\Big|\bigg]\\*
\nn&=\mathbb{E}_i\bigg[\Big|-\frac{\tau_i}{\log\beta}-\frac{1}{\mathrm{GJS}(P_\rmN,P_\rmA,M-2)}\Big|
\mathbb{I}\Big\{\Big|-\frac{\tau_i}{\log\beta}-\frac{1}{\mathrm{GJS}(P_\rmN,P_\rmA,M-2)}\Big|>\varepsilon_2\Big\}\bigg]\\*
\nn&\quad+\mathbb{E}_i\bigg[\Big|-\frac{\tau_i}{\log\beta}-\frac{1}{\mathrm{GJS}(P_\rmN,P_\rmA,M-2)}\Big|
\mathbb{I}\Big\{\Big|-\frac{\tau_i}{\log\beta}-\frac{1}{\mathrm{GJS}(P_\rmN,P_\rmA,M-2)}\Big|<\varepsilon_1\Big\}\bigg]\\*
&\quad+\mathbb{E}_i\bigg[\Big|-\frac{\tau_i}{\log\beta}-\frac{1}{\mathrm{GJS}(P_\rmN,P_\rmA,M-2)}\Big|
\mathbb{I}\Big\{\varepsilon_1\le\Big|-\frac{\tau_i}{\log\beta}-\frac{1}{\mathrm{GJS}(P_\rmN,P_\rmA,M-2)}\Big|\le\varepsilon_2\Big\}\bigg]\label{inmean:three}
\end{align}
As $\beta\to 0$, the first term of \eqref{inmean:three} satisfies
\begin{align}
\nn&\mathbb{E}_i\bigg[\Big|-\frac{\tau_i}{\log\beta}-\frac{1}{\mathrm{GJS}(P_\rmN,P_\rmA,M-2)}\Big|
\mathbb{I}\Big\{\Big|-\frac{\tau_i}{\log\beta}-\frac{1}{\mathrm{GJS}(P_\rmN,P_\rmA,M-2)}\Big|>\varepsilon_2\Big\}\bigg]\\*
&\le\mathbb{E}_i\bigg[\Big(\Big|-\frac{\tau_i}{\log\beta}\Big|+\frac{1}{\mathrm{GJS}(P_\rmN,P_\rmA,M-2)}\Big)
\mathbb{I}\Big\{\Big|-\frac{\tau_i}{\log\beta}-\frac{1}{\mathrm{GJS}(P_\rmN,P_\rmA,M-2)}\Big|>\varepsilon_2\Big\}\bigg]\\*
\nn&\le\mathbb{E}_i\bigg[\Big|-\frac{\tau_i}{\log\beta}\Big|\mathbb{I}\Big\{\Big|-\frac{\tau_i}{\log\beta}-\frac{1}{\mathrm{GJS}(P_\rmN,P_\rmA,M-2)}\Big|>\varepsilon_2\Big\}\bigg]\\*
&\quad+\frac{1}{\mathrm{GJS}(P_\rmN,P_\rmA,M-2)}\bbP_i\Big\{\Big|-\frac{\tau_i}{\log\beta}-\frac{1}{\mathrm{GJS}(P_\rmN,P_\rmA,M-2)}\Big|>\varepsilon_2\Big\}\\
&\to 0,\label{inmean:first}
\end{align}
where \eqref{inmean:first} follows since i) $-\frac{\tau_i}{\log\beta}$ converges in probability in Lemma \ref{converge:p}, and ii) $-\frac{\tau_i}{\log\beta}$ is uniformly integrable in Lemma \ref{lemma:uni_int} by letting $\varepsilon_2$ sufficiently large such that $c\varepsilon_2-1>0$.

The second term of \eqref{inmean:three} satisfies
\begin{align}
\nn&\mathbb{E}_i\bigg[\Big|-\frac{\tau_i}{\log\beta}-\frac{1}{\mathrm{GJS}(P_\rmN,P_\rmA,M-2)}\Big|
\mathbb{I}\Big\{\Big|-\frac{\tau_i}{\log\beta}-\frac{1}{\mathrm{GJS}(P_\rmN,P_\rmA,M-2)}\Big|<\varepsilon_1\Big\}\bigg]\\*
&\le\varepsilon_1\bbP_i\Big\{\Big|-\frac{\tau_i}{\log\beta}-\frac{1}{\mathrm{GJS}(P_\rmN,P_\rmA,M-2)}\Big|<\varepsilon_1\Big\}\\
&\le\varepsilon_1.\label{inmean:second}
\end{align}

Similarly, as $\beta\to 0$, the third term of \eqref{inmean:three} satisfies
\begin{align}
\nn&\mathbb{E}_i\bigg[\Big|-\frac{\tau_i}{\log\beta}-\frac{1}{\mathrm{GJS}(P_\rmN,P_\rmA,M-2)}\Big|
\mathbb{I}\Big\{\varepsilon_1\le\Big|-\frac{\tau_i}{\log\beta}-\frac{1}{\mathrm{GJS}(P_\rmN,P_\rmA,M-2)}\Big|\le\varepsilon_2\Big\}\bigg]\\*
&\le\varepsilon_2\bbP_i\Big\{\varepsilon_1\le\Big|-\frac{\tau_i}{\log\beta}-\frac{1}{\mathrm{GJS}(P_\rmN,P_\rmA,M-2)}\Big|\le\varepsilon_2\Big\}\\
&\le\varepsilon_2\bbP_i\Big\{\Big|-\frac{\tau_i}{\log\beta}-\frac{1}{\mathrm{GJS}(P_\rmN,P_\rmA,M-2)}\Big|\ge\varepsilon_1\Big\}.\label{inmean:third}
\end{align}
By the convergence in probability of $-\frac{\tau_i}{\log\beta}$ and letting $\varepsilon_1\to 0$, both \eqref{inmean:second} and \eqref{inmean:third} tend to zero.

The proof of Lemma \ref{converge:mean} is completed by combining \eqref{inmean:three}, \eqref{inmean:first}, \eqref{inmean:second} and \eqref{inmean:third}.
\end{proof}

Using Lemma \ref{converge:mean}, it follows that for each $i\in[M]$, under hypothesis $\rmH_i$, the misclassification exponent of our sequential test satisfies
\begin{align}
E_i^\mathrm{EP}(\Phi|P_\rmN,P_\rmA)&=\liminf_{\beta\to 0}\frac{-\log\beta}{\mathbb{E}_i[\tau^\mathrm{EP}]}\\*
&\ge\liminf_{\beta\to 0}\frac{-\log\beta}{\mathbb{E}_i[\tau_i]}\label{tauep_le}\\
&=\mathrm{GJS}(P_\rmN,P_\rmA,M-2),\label{ei:ep}
\end{align}
where \eqref{tauep_le} follows from the definitions of $\tau^\mathrm{EP}$ in \eqref{tau:ep} and $\tau_i$ in \eqref{def:taui} which imply that $\tau^\mathrm{EP}\le\tau_i$ and \eqref{ei:ep} follows from the result in \eqref{converge_mean}. The achievability proof of Theorem \ref{seq_errorprob} is now completed.

\subsubsection{Converse Proof of Theorem \ref{seq_errorprob}}\label{proof:ep_con}

Given $(p,q)\in(0,1)^2$, define the binary KL-divergence as follows:
\begin{align}\label{binaryKL}
d(p,q):=p\log\frac{p}{q}+(1-p)\log\frac{1-p}{1-q}.
\end{align}
The first-order derivative of $d(p,q)$ on $q$ is:
\begin{align}
\frac{\partial d(p,q)}{\partial q}=\frac{q-p}{q(1-q)}.
\label{fd:bkl}
\end{align}
Thus, $d(p,q)$ is increasing in $q$ when $q>p$ and decreasing in $q$ when $p>q$.

Fix any $i\in[M]$ and any $j\in\calM_i$.
Recall the definition of $\bbP_i(\cdot):=\Pr\{\cdot|\rmH_i\}$ where $X_i^\tau$ is generated i.i.d. from the anomalous distribution $P_\rmA$ and for each $j\in\calM_i$, $X_j^\tau$ is generated i.i.d. from the nominal distribution $P_\rmN$. Similarly we define $\tilde\bbP_j(\cdot):=\Pr\{\cdot|\rmH_j\}$ when the pair of nominal and anomalous distributions is $(\tilP_\rmN,\tilP_\rmA)$. Define the event $\calW:=\{\phi(\bX^\tau)=\rmH_i\}$.
For any two pairs of distributions $(P_\rmN,P_\rmA)\in\calP(\calX)^2$ and $(\tilP_\rmA,\tilP_\rmN)\in\calP(\calX)^2$, for any sequential test $\Phi=(\tau,\phi)$, we have
\begin{align}
d\big(\bbP_i(\calW),\tilde\bbP_j(\calW)\big)&\le D(\bbP_i||\tilde\bbP_j)|_{\calF_\tau}\label{DP}\\*
&=\mathbb{E}_i\bigg[\sum_{k\in[\tau]}\log\frac{P_i(\bX_k)}{\tilP_j(\bX_k)}\bigg]\label{KL}\\
&=\mathbb{E}_i\bigg[\sum_{k\in[\tau]}\sum_{\substack{t\in[M]:\\t\neq i, t\neq j}}\log\frac{P_\rmN(X_{t,k})}{\tilP_\rmN(X_{t,k})}+\sum_{k\in[\tau]}\log\frac{P_\rmA(X_{i,k})}{\tilP_\rmN(X_{i,k})}+
\sum_{k\in[\tau]}\log\frac{P_\rmN(X_{j,k})}{\tilP_\rmA(X_{j,k})}\bigg]\\
&\le(M-2)\mathbb{E}_i[\tau]D(P_\rmN||\tilP_\rmN)+\mathbb{E}_i[\tau]D(P_\rmA||\tilP_\rmN)+\mathbb{E}_i[\tau]D(P_\rmN||\tilP_\rmA)\label{Doob},
\end{align}
where $\bbP_i|_{\calF_\tau}=\sum_{k\in[\tau]}P_i(\bX_k)$ and $\tilde\bbP_j|_{\calF_\tau}=\sum_{k\in[\tau]}P_j(\bX_k)$, \eqref{DP} follows from the data processing inequality of divergence \cite[Theorem 2.8.1]{cover2012elements}, \eqref{KL} follows from the definition of KL divergence in \eqref{def:KL}~\cite[P. 172]{polyanskiy2014lecture}, and \eqref{Doob} follows from Doob’s Optional Stopping Theorem \cite{klenke2014optional} since the sequence $\sum_{k\in[n]}\log\frac{P_i(\bX_k)}{\tilP_j(\bX_k)}$ is a martingale for any integer $n\in\bbN$.

For any sequential test $\Phi$ satisfying the error probability universality constraint in \eqref{constraint1:ep} and any pair of distributions $(P_\rmN,P_\rmA)\in\calP(\calX)^2$, it follows that $\psi_i(\Phi|P_\rmN,P_\rmA)\to 0$ as $\beta\to 0$. Thus,
\begin{align}
\bbP_i(\calW)&=\bbP_i(\phi(\bX^\tau)=\rmH_i)=1-\psi_i(\Phi|P_\rmN,P_\rmA)\to 1,\\*
\tilde\bbP_j(\calW)&=\tilde\bbP_j(\phi(\bX^\tau)=\rmH_i)\le\psi_j(\Phi|\tilP_\rmA,\tilP_\rmN)\to 0.
\end{align}
As a result, $\bbP_i(\calW)>\tilde\bbP_j(\calW)$ and $1-\psi_i(\Phi|P_\rmN,P_\rmA)>\psi_j(\Phi|\tilP_\rmA,\tilP_\rmN)$. Furthermore, as $\beta\to 0$, it follows that
\begin{align}
d\big(\bbP_i(\calW),\tilde\bbP_j(\calW)\big)&\ge d\big(1-\psi_i(\Phi|P_\rmN,P_\rmA),\psi_j(\Phi|\tilP_\rmA,\tilP_\rmN)\big)\label{d_func}\\*
&=-\log\psi_j(\Phi|\tilP_\rmA,\tilP_\rmN)\label{huslemma2}\\
&\ge-\log\beta.\label{beta}
\end{align}
where \eqref{d_func} follows since $d(p,q)$ is decreasing in $q$ when $p>q$, \eqref{huslemma2} follows since i) $\lim\limits_{(p,q)\to(0,0)}\frac{d(1-p,q)}{-\log q}=1$~\cite[Lemma 2]{Ihwang2022sequential} and ii) when $\beta\to 0$, $\psi_i(\Phi|P_\rmN,P_\rmA)\to 0, \psi_j(\Phi|\tilP_\rmA,\tilP_\rmN)\to 0$, and \eqref{beta} follows from the definition of error probability universality constraint in \eqref{constraint1:ep}.

Combining \eqref{Doob} and \eqref{beta}, we have
\begin{align}
-\log\beta\le&(M-2)\mathbb{E}_i[\tau]D(P_\rmN||\tilP_\rmN)+\mathbb{E}_i[\tau]D(P_\rmA||\tilP_\rmN)+\mathbb{E}_i[\tau]D(P_\rmN||\tilP_\rmA).
\end{align}
Therefore, for any sequential test $\Phi=(\tau,\phi)$ satisfying the error probability universality constraint and any two pairs of distributions $(P_\rmN,P_\rmA)\in\calP(\calX)^2$ and $(\tilP_\rmA,\tilP_\rmN)\in\calP(\calX)^2$, the type-$i$ error exponent satisfies
\begin{align}
E_i^\mathrm{EP}(\Phi|P_\rmN,P_\rmA)\le(M-2)D(P_\rmN||\tilP_\rmN)+D(P_\rmA||\tilP_\rmN)+D(P_\rmN||\tilP_\rmA).\label{profcon:ei}
\end{align}
Since \eqref{profcon:ei} is true for all $(\tilP_\rmA,\tilP_\rmN)\in\calP(\calX)^2$, we can minimize the right hand side of \eqref{profcon:ei} with respect to $(\tilP_\rmA,\tilP_\rmN)$ by letting $\tilP_\rmA=P_\rmN$ and optimizing $\tilP_\rmN$. Thus,
\begin{align}
E_i^\mathrm{EP}(\Phi|P_\rmN,P_\rmA)&\le\min_{Q\in\calP(\calX)}(M-2)D(P_\rmN||Q)+D(P_\rmA||Q)\\
&=\mathrm{GJS}(P_\rmN,P_\rmA,M-2),
\end{align}
where the last step follows from the variational form of generalized Jensen-Shannon Divergence in \eqref{GJS:variational}.

\subsection{Proof of Theorem \ref{seq_stoptime}}\label{proof:est}

\subsubsection{Achievability Proof of Theorem \ref{seq_stoptime}}\label{proof:est_ach}

\subsubsubsection{Expected Stopping Time Universality Constraint}\label{proof:constraint}
We first prove our sequential test in Sec. \ref{test_stoptime} satisfies the expected stopping time universality constraint. For each $i\in[M]$, the average stopping time can be expressed as the following form:
\begin{align}\label{expected}
\mathbb{E}_i[\tau]=\sum\limits_{k=1}^{\infty}\bbP_i\{\tau> k\}=n-1+\sum\limits_{k=n-1}^{\infty}\bbP_i\{\tau> k\}.
\end{align}
Given $k\in\bbN$ such that $k\ge n-1$, we upper bound the term $\sum\limits_{k=n-1}^{\infty}\bbP_i\{\tau^\mathrm{EST}> k\}$ as follows,
\begin{align}
\bbP_i\{\tau^\mathrm{EST}>k\}
&\le\bbP_i\big\{\rmS_i(\bX^k)\ge f(k)\big\}\label{Sige}\\
&\le\sum_{\bQ\in\calP_k(\calX)^M: \rmG_i(\bQ)\ge f(k)}\exp\Big\{-kD(Q_i||P_\rmA)-k\sum_{j\in\calM_i}D(Q_j||P_\rmN)\Big\}\label{upperbound1}\\
&\le\sum_{\bQ\in\calP_k(\calX)^M: \rmG_i(\bQ)\ge f(k)}\exp\Big\{-k\Big(D(Q_i||P_\rmA)+(M-1)D\Big(\frac{\sum_{t\in\calM_i}Q_t}{M-1}\Big\|P_\rmN\Big)+f(k)\Big)\Big\}\label{sum_Qj}\\
&\le\sum_{\bQ\in\calP_k(\calX)^M: \rmG_i(\bQ)\ge f(k)}\exp\{-kf(k)\}\label{positive}\\
&\le (k+1)^{M|\calX|}(k+1)^{-(M+1)|\calX|}\label{typeclass}\\
&=(k+1)^{-|\calX|},
\end{align}
where \eqref{Sige} follows from the definition of $\tau^\mathrm{EST}$ in \eqref{tau:est}, 
\eqref{upperbound1} follows from the similar manner as \eqref{upperbound4}, \eqref{sum_Qj} follows from the equation in \eqref{equation1} and the constraint $\rmG_i(\bQ)\ge f(k)$, \eqref{positive} follows since KL divergence is nonnegative, and \eqref{typeclass} follows from \cite[Theorem 11.1.1]{cover2012elements} which implies that $|\calP_k(\calX)|\le(k+1)^{|\calX|}$.

Thus, for $n\ge 2$, we have
\begin{align}
\sum\limits_{k=n-1}^{\infty}\bbP_i\{\tau^\mathrm{EST}> k\}\le\sum\limits_{k=n-1}^{\infty}(k+1)^{-|\calX|}
\le&\int_{n-2}^{\infty}(u+1)^{-|\calX|}\mathrm{d}u\\*
=&\frac{1}{-|\calX|+1}(u+1)^{-|\calX|+1}\Big|_{u=n-2}^{u=\infty}\\*
=&\frac{(n-1)^{-(|\calX|-1)}}{|\calX|-1}\le 1.\label{sum:k+1}
\end{align}

Combining \eqref{sum:k+1} and \eqref{expected}, we conclude that $\mathbb{E}_i[\tau^\mathrm{EST}]\le n$ for each $i\in[M]$.

\subsubsubsection{Achievable Misclassification Exponent}\label{proof:est_exp}
Recall the definition of $f(k)=\tfrac{(M+1)|\calX|\log (k+1)}{k}$.
Fix any $i\in[M]$. Given any pair of distributions $(P_\rmN,P_\rmA)\in\calP(\calX)^2$ and any $n\in\bbN$, define the error exponent function
\begin{align}\label{deltai}
\Delta_i(n,P_\rmN,P_\rmA):=\min_{j\in\calM_i}\min_{\substack{\bQ\in\calP_n(\calX)^M: \\\rmG_j(\bQ)\le f(n)}}D(Q_i||P_\rmA)+\sum_{t\in\calM_i}D(Q_t||P_\rmN).
\end{align}

We upper bound the misclassification probability as follows:
\begin{align}
\psi_i(\phi^\mathrm{EST}|P_\rmN,P_\rmA)
&=\bbP_i\{\phi(\bX^{\tau^\mathrm{EST}})\neq \rmH_i\}\label{est:tau}\\
&=\bbP_i\Big\{\exists~j\in\calM_i~\mathrm{s.t.}~\rmS_j(\bX^{\tau^\mathrm{EST}})\le f(\tau)\Big\}\\
&\le\bbP_i\Big\{\exists~j\in\calM_i~\mathrm{and}~k\ge n-1~\mathrm{s.t.}~\rmS_j(\bX^k)\le f(k)\Big\}\\
&\le \sum\limits_{j\in\calM_i}\sum_{k=n-1}^{\infty}\bbP_i\Big\{\rmS_j(\bX^k)\le f(k)\Big\}\\
&\le (M-1)\max_{j\in\calM_i}\sum_{k=n-1}^{\infty}\bbP_i\{\rmS_j(\bX^k)\le f(k)\}\\
&\le (M-1)\max_{j\in\calM_i}\sum_{k=n-1}^{\infty}\sum_{\bQ\in\calP_k(\calX)^M: \rmG_j(\bQ)\le f(k)}\exp\Big\{-k\Big(D(Q_i||P_\rmA)+\sum_{t\in\calM_i}D(Q_t||P_\rmN)\Big)\Big\}\label{type_beta}\\
&\le (M-1)\sum_{k=n-1}^{\infty}\exp\Big\{-k\Big(\Delta_i(k,P_\rmN,P_\rmA)-\frac{M|\calX|\log(k+1)}{k}\Big)\Big\}\label{number1}\\
&\le (M-1)\sum_{k=n-1}^{\infty}\exp\Big\{-k\Big(\Delta_i(n-1,P_\rmN,P_\rmA)-\frac{M|\calX|\log n}{n-1}\Big)\Big\}\label{sum_k},
\end{align}
where \eqref{type_beta} follows from the same argument as \eqref{upperbound4}, \eqref{number1} follows from the same argument as \eqref{knumber}, and \eqref{sum_k} follows from the function $\Delta_i(n,P_\rmN,P_\rmA)$ is increasing in $n$ and $\tfrac{M|\calX|\log n}{n-1}$ is decreasing in $n$.

As $n\to\infty,f(n-1)\rightarrow 0$ and $\tfrac{M|\calX|\log n}{n-1}$ tends to $0$. Thus,
\begin{align}
\lim\limits_{n\to\infty}\Delta_i(n-1,P_\rmN,P_\rmA)
&=\min_{(Q_1,Q_2)\in\calP(\calX)^2}D(Q_1||P_\rmA)+(M-2)D(Q_1||P_\rmN)+D(Q_2||P_\rmN)\label{nto0}\\
&=\min_{Q\in\calP(\calX)}D(Q||P_\rmA)+(M-2)D(Q||P_\rmN)\label{Q2}\\
&=D_{\frac{M-2}{M-1}}(P_\rmN||P_\rmA)\label{asym:exponent},
\end{align}
where \eqref{nto0} follows from i) the constraint in \eqref{deltai} when $f(n-1)\rightarrow 0$, i.e, $\rmG_j(\bQ)\le f(n)=0$, which leads to the fact that $\rmG_j(\bQ)=0$ since KL divergence is non-negative, and ii) the definition in \eqref{G_i} that $\rmG_j(\bQ)=0$ which implies $Q_j=Q_1$ for all $j\in\calM_i$ with an arbitrary $Q_1\in\calP(\calX)$, \eqref{Q2} follows by letting $Q_2=P_\rmN$ and \eqref{asym:exponent} follows from the variational form of the R\'{e}nyi Divergence in \eqref{renyi:variational}.

Using \eqref{asym:exponent}, as $n\to\infty$, \eqref{sum_k} can be upper bounded as follows:
\begin{align}
(M-1)\sum_{k=n-1}^{\infty}\exp\Big\{-kD_{\frac{M-2}{M-1}}(P_\rmN||P_\rmA)\Big\}=
(M-1)\frac{\exp\Big\{-(n-1)D_{\frac{M-2}{M-1}}(P_\rmN||P_\rmA)\Big\}}{1-\exp\Big\{-D_{\frac{M-2}{M-1}}(P_\rmN||P_\rmA)\Big\}}.
\end{align}
Thus, the misclassification exponent satisfies
\begin{align}
E_i^\mathrm{EST}(\Phi^\mathrm{EST}|P_\rmN,P_\rmA)&\ge\liminf_{n\to\infty}\Big\{\frac{n-1}{n}D_{\frac{M-2}{M-1}}(P_\rmN||P_\rmA)-\frac{\log(M-1)}{n}+\frac{1}{n}\log\Big(1-\exp\big\{D_{\frac{M-2}{M-1}}(P_\rmN||P_\rmA)\big\}\Big)\Big\}\\
&=D_{\frac{M-2}{M-1}}(P_\rmN||P_\rmA).
\end{align}

\subsubsection{Converse Proof of Theorem \ref{seq_stoptime}}\label{proof:est_con}
Fix any $j\in[M]$ and any $i\in\calM_j$. Recall the definition of $\bbP_i(\cdot)$ and $\tilde\bbP_j(\cdot):=\Pr\{\cdot|\rmH_j\}$ below \eqref{fd:bkl}. We next recall the definition of the binary KL-divergence in \eqref{binaryKL}.
Define the event $\calW:=\{\phi(\bX^\tau)=\rmH_i\}$. For any two pairs of distributions $(P_\rmN,P_\rmA)\in\calP(\calX)^2$ and $(\tilP_\rmA,\tilP_\rmN)\in\calP(\calX)^2$, for any sequential test $\Phi=(\tau,\phi)$ satisfying the expected stopping time universality constraint, we have
\begin{align}
d\big(\bbP_i(\calW),\tilde\bbP_j(\calW)\big)&\le(M-2)\mathbb{E}_i[\tau]D(P_\rmN||\tilP_\rmN)+\mathbb{E}_i[\tau]D(P_\rmA||\tilP_\rmN)+\mathbb{E}_i[\tau]D(P_\rmN||\tilP_\rmA)\label{Doob2}\\
&\le(M-2)nD(P_\rmN||\tilP_\rmN)+nD(P_\rmA||\tilP_\rmN)+nD(P_\rmN||\tilP_\rmA)\label{Etau2},
\end{align}
where \eqref{Doob2} follows from the similar argument to \eqref{Doob} and \eqref{Etau2} follows from the definition of expected stopping time universality constraint in \eqref{constraint1:est}.

Fix any sequential test $\Phi$ with positive error exponents under any pair of distributions $(P_\rmN,P_\rmA)\in\calP(\calX)^2$. It follows from \eqref{def:exponent:est} that $\psi_i(\Phi|P_\rmN,P_\rmA)\to 0$ as $n\to\infty$. Furthermore, using the steps analogously to those leading to the result in \eqref{huslemma2} in Appendix \ref{proof:ep_con}, as $n\to\infty$, we have
\begin{align}
d\big(\bbP_i(\calW),\tilde\bbP_j(\calW)\big)&\ge-\log\psi_j(\Phi|\tilP_\rmA,\tilP_\rmN).\label{dlebeta}
\end{align}
Combining \eqref{Etau2} and \eqref{dlebeta}, we obtain that
\begin{align}
-\log\psi_j(\Phi|\tilP_\rmA,\tilP_\rmN)\le(M-2)nD(P_\rmN||\tilP_\rmN)+nD(P_\rmA||\tilP_\rmN)+nD(P_\rmN||\tilP_\rmA).
\end{align}
Therefore,  under any two pairs of distributions $(P_\rmN,P_\rmA)\in\calP(\calX)^2$ and $(\tilP_\rmA,\tilP_\rmN)\in\calP(\calX)^2$, the type-$j$ error exponent satisfies
\begin{align}
E_j^\mathrm{EST}(\Phi|\tilP_\rmA,\tilP_\rmN)\le(M-2)D(P_\rmN||\tilP_\rmN)+D(P_\rmA||\tilP_\rmN)+D(P_\rmN||\tilP_\rmA).\label{con:ej}
\end{align}
Since \eqref{con:ej} is true for all $(P_\rmN,P_\rmA)\in\calP(\calX)^2$, we can minimize the right hand side of \eqref{con:ej} with respect to $(P_\rmN,P_\rmA)$ by letting $P_\rmA=\tilP_\rmN$ and optimizing $P_\rmN$. Thus,
\begin{align}\label{con:est_one}
E_j^\mathrm{EST}(\Phi|\tilP_\rmA,\tilP_\rmN)&\le\min_{Q\in\calP(\calX)}(M-2)D(Q||\tilP_\rmN)+D(Q||\tilP_\rmA)\\
&=D_{\frac{M-2}{M-1}}(\tilP_\rmN||\tilP_\rmA),
\end{align}
where the last step follows from the variational form of R\'{e}nyi Divergence in \eqref{renyi:variational}.

\subsection{ Proof of Theorem \ref{seq_T_time}}\label{proof:T}

\subsubsection{Achievability Proof of Theorem \ref{seq_T_time}}\label{proof:T_ach}

\subsubsubsection{Expected Stopping Time Universality Constraint}
Fix any $\calB\in\calS(T)$. The average stopping time can be expressed as the following form:
\begin{align}
\mathbb{E}_\calB[\tau]=\sum\limits_{k=1}^{\infty}\bbP_\calB\{\tau> k\}=n-1+\sum\limits_{k=n-1}^{\infty}\bbP_\calB\{\tau> k\}.
\end{align}

Analogously to the steps in \eqref{Sige}-\eqref{typeclass}, we have
\begin{align}
\nn&\bbP_\calB\{\tau>k\}\\*
&\le\bbP_\calB\big\{\rmS_\calB(\bX^k)\ge f(k)\big\}\label{SB}\\
&=\sum_{\bQ\in\calP_k(\calX)^M: \rmG_\calB(\bQ)\ge f(k)}\Big(\prod_{i\in\calB}P_\rmA(\calT_{Q_i}^k)\Big)\times \Big(\prod_{j\in\calM_\calB}P_\rmN(\calT_{Q_j}^k)\Big)\label{type3}\\
&\le\sum_{\bQ\in\calP_k(\calX)^M: \rmG_\calB(\bQ)\ge f(k)}\exp\Big\{-k\sum_{i\in\calB}D(Q_i||P_\rmA)-k\sum_{j\in\calM_\calB}D(Q_j||P_\rmN)\Big\}\\
&\le\sum_{\bQ\in\calP_k(\calX)^M: \rmG_\calB(\bQ)\ge f(k)}\exp\bigg\{-k\bigg(T\cdot D\Big(\frac{\sum_{t\in\calB}Q_t}{T}\Big\|P_\rmA\Big)+(M-T)D\Big(\frac{\sum_{t\in\calM_\calB}Q_t}{M-T}\Big\|P_\rmN\Big)+f(k)\bigg)\bigg\}\label{sum_QB}\\
&\le\sum_{\bQ\in\calP_k(\calX)^M: \rmG_\calB(\bQ)\ge f(k)}\exp\{-kf(k)\}\label{kfk}\\
&=(k+1)^{-|\calX|},
\end{align}
where \eqref{SB} follows from the definition of $\tau$ in \eqref{tau:exactlyT}, \eqref{type3} follows from the definitions of $\rmG_\calB(\cdot)$ in \eqref{G_B} and $\rmS_\calB(\cdot)$ in Sec. \ref{test:exactlyT}, and \eqref{sum_QB} follows from the constraint $\rmG_\calB(\bQ)\ge f(k)$ and the following equations:
\begin{align}\label{equation2}
\sum_{j\in\calM_\calB}D(Q_j||P_\rmN)&=(M-T)D\Big(\frac{\sum_{t\in\calM_\calB}Q_t}{M-T}\Big\|P_\rmN\Big)
+\sum_{j\in\calM_\calB}D\Big(Q_j\Big\|\frac{\sum_{t\in\calM_\calB}Q_t}{M-T}\Big),\\
\sum_{i\in\calB}D(Q_i||P_\rmA)&=T\cdot D\Big(\frac{\sum_{l\in\calB}Q_l}{T}\Big\|P_\rmA\Big)
+\sum_{i\in\calB}D\Big(Q_i\Big\|\frac{\sum_{l\in\calB}Q_l}{T}\Big),
\end{align}
which follow from \eqref{equation1}. Similarly to the results in \eqref{sum:k+1}, it follows that
\begin{align}
\sum\limits_{k=n-1}^{\infty}\bbP_\calB\{\tau>k\}\le\sum\limits_{k=n-1}^{\infty}(k+1)^{-|\calX|}\le 1,
\end{align}
and thus,
\begin{align}
\mathbb{E}_\calB[\tau]= n-1+\sum\limits_{k=n-1}^{\infty}\bbP_\calB\{\tau> k\}\le n.
\end{align}

\subsubsubsection{Achievable Misclassification Exponent}
Fix any $\calB\in\calS(T)$. Recall the definitions of set $\calM_\calB=[M]\backslash\calB$ and $\calS_\calB(T):=\{\calC\in\calS(T):\calC\neq\calB\}$ and the threshold $f(k)=\tfrac{(M+1)|\calX|\log (k+1)}{k}$. We need the following definitions to present our proof. Given any pair of distributions $(P_\rmN,P_\rmA)$, for each $\calB\in\calS(T)$, define
\begin{align}
\Delta_\calB(n,P_\rmN,P_\rmA):=\min_{\calC\in\calS_\calB(T)}\min_{\substack{\bQ\in\calP_n(\calX)^M: \\\rmG_\calC(\bQ)\le f(n)}}\sum_{i\in\calB}D(Q_i||P_\rmA)+\sum_{t\in\calM_\calB}D(Q_t||P_\rmN).
\end{align}

Analogously to the steps in \eqref{est:tau}-\eqref{sum_k}, we upper bound the misclassification probability under hypothesis $\rmH_\calB$ as follows:
\begin{align}
\nn&\psi_\calB(\Phi|P_\rmN,P_\rmA)\\*
&=\bbP_\calB\Big\{\exists~\calC\in\calS_\calB(T)~\mathrm{s.t.}~\rmS_\calC(\bX^\tau)\le f(\tau)\Big\}\\
&\le\bbP_\calB\Big\{\exists~\calC\in\calS_\calB(T),~k\ge n-1~\mathrm{s.t.}~\rmS_\calC(\bX^k)\le f(k)\Big\}\\
&\le (|\calS(T)|-1)\sum_{k=n-1}^{\infty}\exp\Big\{-k\Big(\Delta_\calB(n-1,P_\rmN,P_\rmA)-\frac{M|\calX|\log n}{n-1}\Big)\Big\}.\label{sum_k3}
\end{align}
As $n\to\infty$, $f(n-1)\rightarrow 0$, $\frac{M|\calX|\log n}{n-1}\to 0$ and thus,
\begin{align}
\nn&\lim\limits_{n\to\infty}\Delta_\calB(n-1,P_\rmN,P_\rmA)\\*
&=\min_{\calC\in\calS_\calB(T)}\min_{(Q_1,Q_2)\in\calP(\calX)^2}|\calB\cap\calC|D(Q_1||P_\rmA)+|\calB\cap\calM_\calC| D(Q_2||P_\rmA)+(M-|\calB\cup\calC|)D(Q_2||P_\rmN)+|\calC\cap\calM_\calB|D(Q_1||P_\rmN)\label{LDB}\\
&=\min_{t\in\bbN:t\le T-1}\min_{(Q_1,Q_2)\in\calP(\calX)^2}\big((T-t)D(Q_1||P_\rmA)+(M-2T+t)D(Q_1||P_\rmN)\big)+\big((T-t)D(Q_2||P_\rmN)+tD(Q_2||P_\rmA)\big)\label{LD_t}\\
&=\mathrm{LD}_\calB(P_\rmN,P_\rmA,M,T),\label{defLD}
\end{align}
where \eqref{LD_t} follows because we use $t$ to denote $|\calB\cap\calC|$ for each $\calC\in\calS_\calB(T)$ where $0\le t\le T-1$ and thus $M-|\calB\cup\calC|=M-2T+t$ and $|\calB\cap\calM_\calC|=|\calC\cap\calM_\calB|=T-t$, and \eqref{defLD} follows from the definition of $\mathrm{LD}_\calB(P_\rmN,P_\rmA,M,T)$ in \eqref{LDB:renyi}.
Thus, it follows from \eqref{sum_k3} and \eqref{defLD} that
\begin{align}
E_\calB(\Phi|P_\rmN,P_\rmA)\ge\mathrm{LD}_\calB(P_\rmN,P_\rmA,M,T).
\end{align}

\subsubsection{Converse Proof of Theorem \ref{seq_T_time}}\label{proof:T_con}
For each $\calB\in\calS(T)$, recall the definition of $\bbP_\calB(\cdot):=\Pr\{\cdot|\rmH_\calB\}$ where for each $i\in\calB$, $X_i^\tau$ is generated i.i.d. from the anomalous distribution $P_\rmA$ and for each $j\in\calM_\calB$, $X_j^\tau$ is generated i.i.d. from the nominal distribution $P_\rmN$. Similarly we define $\tilde\bbP_\calC(\cdot):=\Pr\{\cdot|\rmH_\calC\}$ when the pair of nominal and anomalous distributions is $(\tilP_\rmN,\tilP_\rmA)$.
We next recall the definition of the binary KL-divergence in \eqref{binaryKL} and its increasing and decreasing in Appendix \ref{proof:ep_con} that $d(p,q)$ is increasing in $q$ when $q>p$ and decreasing in $q$ when $p>q$.

Fix any $\calC\in\calS(T)$ and any $\calB\in\calS_\calC(T)$. Define the event $\calW:=\{\phi(\bx^\tau)=\rmH_\calB\}$. For any two pairs of distributions $(P_\rmN,P_\rmA)\in\calP(\calX)^2$ and $(\tilP_\rmA,\tilP_\rmN)\in\calP(\calX)^2$, and any sequential test $\Phi=(\tau,\phi)$ satisfying the expected stopping time universality constraint, we have
\begin{align}
\nn&d\big(\bbP_\calB(\calW),\tilde\bbP_\calC(\calW)\big)\\*
&\le D(\bbP_\calB||\tilde\bbP_\calC)|_{\calF_\tau}\label{DP4}\\*
&=\mathbb{E}_\calB\bigg[\sum_{k\in[\tau]}\log\frac{P_\calB(\bX_k)}{\tilP_\calC(\bX_k)}\bigg]\label{KL4}\\
&=\mathbb{E}_\calB\bigg[\sum_{k\in[\tau]}\sum_{t\in\calM_{\calB\cup\calC}}\log\frac{P_\rmN(X_{t,k})}{\tilP_\rmN(X_{t,k})}+\sum_{k\in[\tau]}\sum_{i\in\calB\cap\calM_\calC}\log\frac{P_\rmA(X_{i,k})}{\tilP_\rmN(X_{i,k})}
+\sum_{k\in[\tau]}\sum_{j\in\calC\cap\calM_\calB}\log\frac{P_\rmN(X_{j,k})}{\tilP_\rmA(X_{j,k})}+\sum_{k\in[\tau]}\sum_{l\in\calB\cap\calC}\log\frac{P_\rmA(X_{l,k})}{\tilP_\rmA(X_{l,k})}\bigg]\\
&\le\mathbb{E}_\calB[\tau]\Big((M-|\calB\cup\calC|)D(P_\rmN||\tilP_\rmN)+|\calB\cap\calM_\calC|D(P_\rmA||\tilP_\rmN)
+|\calC\cap\calM_\calB|D(P_\rmN||\tilP_\rmA)+|\calB\cap\calC|D(P_\rmA||\tilP_\rmA)\Big)\label{Doob4}\\
&\le n\Big((M-|\calB\cup\calC|)D(P_\rmN||\tilP_\rmN)+|\calB\cap\calM_\calC|D(P_\rmA||\tilP_\rmN)
+|\calC\cap\calM_\calB|D(P_\rmN||\tilP_\rmA)+|\calB\cap\calC|D(P_\rmA||\tilP_\rmA)\Big)\label{Etau3},
\end{align}
where $\bbP_\calB|_{\calF_\tau}=\sum_{k\in[\tau]}P_\calB(\bX_k)$ and $\tilde\bbP_\calC|_{\calF_\tau}=\sum_{k\in[\tau]}P_\calC(\bX_k)$, \eqref{DP4} follows from the data processing inequality of divergence \cite[Theorem 2.8.1]{cover2012elements}, \eqref{KL4} follows from the definition of KL divergence in \eqref{def:KL}~\cite[P. 172]{polyanskiy2014lecture}, \eqref{Doob4} follows from Doob’s Optional Stopping Theorem \cite{klenke2014optional} since the sequence $\sum_{k\in[n]}\log\frac{P_\calB(\bX_k)}{\tilP_\calC(\bX_k)}$ is a martingale for any integer $n\in\bbN$ and \eqref{Etau3} follows from the definition of expected stopping time universality constraint in \eqref{constraint2:est}.

Fix any sequential test $\Phi$ with positive error exponents under any pair of distributions $(P_\rmN,P_\rmA)\in\calP(\calX)^2$. It follows from \eqref{def:exponent:exactT} that $\psi_\calB(\Phi|P_\rmN,P_\rmA)\to 0$ as $n\to\infty$. Thus, we have
\begin{align}
\bbP_\calB(\calW)&=\bbP_\calB(\phi(\bx^\tau)=\calB)=1-\psi_\calB(\Phi|P_\rmN,P_\rmA)\to 1,\\*
\tilde\bbP_\calC(\calW)&=\tilde\bbP_\calC(\phi(\bx^\tau)=\calB)\le\psi_\calC(\Phi|\tilP_\rmA,\tilP_\rmN)\to 0.
\end{align}
Thus, we obtain $\bbP_\calB(\calW)>\tilde\bbP_\calC(\calW)$ and $1-\psi_\calB(\Phi|P_\rmN,P_\rmA)>\psi_\calC(\Phi|\tilP_\rmA,\tilP_\rmN)$. Furthermore, as $n\to\infty$, we have
\begin{align}
d\big(\bbP_\calB(\calW),\tilde\bbP_\calC(\calW)\big)&\ge d\big(1-\psi_\calB(\Phi|P_\rmN,P_\rmA),\psi_\calC(\Phi|\tilP_\rmA,\tilP_\rmN)\big)\label{d_func3}\\*
&=-\log\psi_\calC(\Phi|\tilP_\rmA,\tilP_\rmN),\label{huslemma23}
\end{align}
where \eqref{d_func3} follows since $d(p,q)$ is decreasing in $q$ when $p>q$ and \eqref{huslemma23} follows from the similar manner to \eqref{huslemma2}.

Combining \eqref{Etau3} and \eqref{huslemma23}, we obtain that
\begin{align}
\nn&-\log\psi_\calC(\Phi|\tilP_\rmA,\tilP_\rmN) \\* &\le n\min_{\calB\in\calS_\calC(T)}(M-|\calB\cup\calC|)D(P_\rmN||\tilP_\rmN)+|\calB\cap\calM_\calC|D(P_\rmA||\tilP_\rmN)
+|\calC\cap\calM_\calB|D(P_\rmN||\tilP_\rmA)+|\calB\cap\calC|D(P_\rmA||\tilP_\rmA).
\end{align}
Therefore, under any two pairs of distributions $(P_\rmN,P_\rmA)\in\calP(\calX)^2$ and $(\tilP_\rmA,\tilP_\rmN)\in\calP(\calX)^2$, the error exponent satisfies
\begin{align}
\nn&E_\calC(\Phi|\tilP_\rmA,\tilP_\rmN)\\*
&\le\min_{\calB\in\calS_\calC(T)}(M-|\calB\cup\calC|)D(P_\rmN||\tilP_\rmN)+|\calC\cap\calM_\calB|D(P_\rmN||\tilP_\rmA)+|\calB\cap\calM_\calC|D(P_\rmA||\tilP_\rmN)
+|\calB\cap\calC|D(P_\rmA||\tilP_\rmA).\label{con:ec}
\end{align}
Since \eqref{con:ec} is true for all $(P_\rmN,P_\rmA)\in\calP(\calX)^2$, we can minimize the right hand side of \eqref{con:ec} with respect to $(P_\rmN,P_\rmA)$ and obtain
\begin{align}
\nn E_\calC(\Phi|\tilP_\rmA,\tilP_\rmN)&\le\min_{(Q_1,Q_2)\in\calP(\calX)^2}\min_{\calB\in\calS_\calC(T)}(M-|\calB\cup\calC|)D(Q_1||\tilP_\rmN)
+|\calC\cap\calM_\calB|D(Q_1||\tilP_\rmA)\\*
&\qquad\qquad\qquad\qquad\qquad+|\calB\cap\calM_\calC|D(Q_2||\tilP_\rmN)+|\calB\cap\calC|D(Q_2||\tilP_\rmA)\\
&=\mathrm{LD}_\calC(\tilP_\rmN,\tilP_\rmA,M,T).
\end{align}

\subsection{Proof of Theorem \ref{at_most_one}}\label{proof:at_most_one}

\subsubsection{Expected Stopping Time Universality Constraint}

We first prove our sequential test in Sec. \ref{test_most_one} satisfies the expected stopping time universality constraint with $n$. For each $i\in[M]$, the average stopping time under hypothesis $\rmH_i$ can be expressed as the following form:
\begin{align}\label{Etau:i}
\mathbb{E}_i[\tau]=\sum\limits_{k=1}^{\infty}\bbP_i\{\tau> k\}=n-1+\sum\limits_{k=n-1}^{\infty}\bbP_i\{\tau> k\}.
\end{align}

We next upper bound the term $\sum\limits_{k=n-1}^{\infty}\bbP_i\{\tau> k\}$.
For simplicity, define $\calA=\{t\ge n-1:\exists~l\in[M]~\mathrm{s.t.}~\rmS_l(\bx^t)\le\lambda_2~\mathrm{and}~\min_{j\in\calM_l}\rmS_j(\bx^t)>\lambda_1\}$ and $\calD=\{t\ge n-1:\forall~j\in[M],~\rmS_j(\bx^t)\le\lambda_2\}$. From the definition of stopping time $\tau$ in \eqref{mostone_stoptime},  we have $\tau=\inf\{\tilt:\tilt\in\calA\cup\calD\}$. Given any $k\ge n-1$, $\tau>k$ implies that $k\in\calA^c\cap\calD^c$.
Thus, we have
\begin{align}
\sum\limits_{k=n-1}^{\infty}\bbP_i\{\tau>k\}
&\le\sum\limits_{k=n-1}^{\infty}\bbP_i\{k\in\calA^c\cap\calD^c\}\\
&\le\sum\limits_{k=n-1}^{\infty}\bbP_i\{k\in\calA^c\}\\
&\le\sum\limits_{k=n-1}^{\infty}\bbP_i\Big\{\forall~l\in[M],~\rmS_l(\bX^k)>\lambda_2~\mathrm{or}~\min_{j\in\calM_l}\rmS_j(\bX^k)\le\lambda_1\Big\}\\
&\le\sum\limits_{k=n-1}^{\infty}\bbP_i\big\{\forall~l\in[M],~\rmS_l(\bX^k)>\lambda_2\big\}+\sum\limits_{k=n-1}^{\infty}\bbP_i\Big\{\forall~l\in[M],~\min_{j\in\calM_l}\rmS_j(\bX^k)\le\lambda_1\Big\}.\label{mostone:taugek}
\end{align}

The first term of \eqref{mostone:taugek} can be upper bounded as follows:
\begin{align}
\nn&\sum\limits_{k=n-1}^{\infty}\bbP_i\big\{\forall~l\in[M],~\rmS_l(\bX^k)>\lambda_2\big\}\\
&\le\sum\limits_{k=n-1}^{\infty}\bbP_i\big\{\rmS_i(\bX^k)>\lambda_2\big\}\\
&=\sum\limits_{k=n-1}^{\infty}\sum_{\bQ\in\calP_k(\calX)^M: \rmG_i(\bQ)>\lambda_2}P_\rmA(\calT_{Q_i}^k)\times \Big(\prod_{t\in\calM_i}P_\rmN(\calT_{Q_t}^k)\Big)\label{type2}\\
&\le\sum\limits_{k=n-1}^{\infty}\sum_{\bQ\in\calP_k(\calX)^M: \rmG_i(\bQ)>\lambda_2}\exp\Big\{-k\Big(D(Q_i||P_\rmA)+\sum_{t\in\calM_i}D(Q_t||P_\rmN)\Big)\Big\}\label{upperbound8}\\
&\le\sum\limits_{k=n-1}^{\infty}\sum_{\bQ\in\calP_k(\calX)^M: \rmG_i(\bQ)>\lambda_2}\exp\Big\{-k\Big(D(Q_i||P_\rmA)+(M-1)D\Big(\frac{\sum_{t\in\calM_i}Q_t}{M-1}\Big\|P_\rmN\Big)+\lambda_2\Big\}\label{equality2}\\
&\le\sum\limits_{k=n-1}^{\infty}\exp\Big\{-k\Big(\lambda_2-\frac{M|\calX|\log(k+1)}{k}\Big)\Big\}\label{tauk:numer}\\
&\le\sum\limits_{k=n-1}^{\infty}\exp\Big\{-k\Big(\lambda_2-\frac{M|\calX|\log n}{n-1}\Big)\Big\},\label{decrease4}\\
&=\frac{\exp\Big\{-(n-1)\Big(\lambda_2-\frac{M|\calX|\log n}{n-1}\Big)\Big\}}{1-\exp\Big\{-\Big(\lambda_2-\frac{M|\calX|\log n}{n-1}\Big)\Big\}},\label{one:taugek1}
\end{align}
where \eqref{type2} follows from the definitions of $\rmG_i(\cdot)$ in \eqref{G_i} and $\rmS_i(\cdot)$ in Sec. \ref{test_most_one}, \eqref{upperbound8} follows from the upper bound on the probability of the type class \cite[Theorem 11.1.4]{cover2012elements}, \eqref{equality2} follows from the constraint $\rmG_i(\bQ)>\lambda_2$ and the equation in \eqref{equation1}, \eqref{tauk:numer} follows from the steps analogously to those leading to the result in \eqref{knumber}, and \eqref{decrease4} follows from $\frac{M|\calX|\log(k+1)}{k}$ is decreasing with $k$.
Therefore, for each $i\in[M]$, $\sum\limits_{k=n-1}^{\infty}\bbP_i\{\forall~l\in[M],~\rmS_l(\bX^k)>\lambda_2\}\le\frac{1}{2}$ when $n$ is sufficiently large.

The second term of \eqref{mostone:taugek} can be upper bounded as follows:
\begin{align}
&\sum\limits_{k=n-1}^{\infty}\bbP_i\Big\{\forall~l\in[M],~\min_{j\in\calM_l}\rmS_j(\bX^k)\le\lambda_1\Big\}\\
&\le\sum\limits_{k=n-1}^{\infty}\bbP_i\Big\{\min_{j\in\calM_i}\rmS_j(\bX^k)\le\lambda_1\Big\}\\
&\le\sum\limits_{k=n-1}^{\infty}\bbP_i\{\exists~j\in \calM_i~\mathrm{s.t.}~\rmS_j(\bX^k)\le\lambda_1\}\\
&\le(M-1)\sum\limits_{k=n-1}^{\infty}\exp\Big\{-k\Big(\Omega_i(\lambda_1,P_\rmN,P_\rmA)-\frac{M|\calX|\log n}{n-1}\Big)\Big\}\label{taugek:omega}\\
&=(M-1)\frac{\exp\Big\{-(n-1)\Big(\Omega_i(\lambda_1,P_\rmN,P_\rmA)-\frac{M|\calX|\log n}{n-1}\Big)\Big\}}{1-\exp\Big\{-\Big(\Omega_i(\lambda_1,P_\rmN,P_\rmA)-\frac{M|\calX|\log n}{n-1}\Big)\Big\}},\label{one:taugek2}
\end{align}
where \eqref{taugek:omega} follows from the steps analogously to those leading to the result in \eqref{sum_k} and the definition of $\Omega_i(\lambda,P_\rmN,P_\rmA)$ in \eqref{Omegai}. Therefore, we have for each $i\in[M]$, $\sum\limits_{k=n-1}^{\infty}\bbP_i\{\forall~l\in[M]~\mathrm{s.t.}~\min_{j\in\calM_l}\rmS_j(\bX^k)\le\lambda_1\}\le\frac{1}{2}$ when $n$ is sufficiently large and $0<\lambda_1<\mathrm{GJS}(P_\rmN,P_\rmA,M-2)$.

Combining \eqref{mostone:taugek}, \eqref{one:taugek1} and \eqref{one:taugek2}, we have $\sum\limits_{k=n-1}^{\infty}\bbP_i\{\tau>k\}\le 1$. Therefore, it follows from \eqref{Etau:i} that $\mathbb{E}_i[\tau]\le n$ for each $i\in[M]$ when $n$ is sufficiently large and $0<\lambda_1<\mathrm{GJS}(P_\rmN,P_\rmA,M-2)$.

We next bound the expected stopping time under the null hypothesis. As discussed below \eqref{Etau:i}, similarly to the results in \eqref{one:taugek1}, the term $\sum\limits_{k=n-1}^{\infty}\bbP_\rmr\{\tau> k\}$ can be upper bounded as follows:
\begin{align}
\sum_{k=n-1}^{\infty}\bbP_\rmr\{\tau>k\}
&\le\sum\limits_{k=n-1}^{\infty}\bbP_\rmr\{k\in\calA^c\cap\calD^c\}\\
&\le\sum\limits_{k=n-1}^{\infty}\bbP_\rmr\{k\in\calD^c\}\\
&\le\sum\limits_{k=n-1}^{\infty}\bbP_\rmr\big\{\exists~j\in[M]:~\rmS_j(\bX^k)>\lambda_2\big\}\\
&\le\sum_{j\in[M]}\sum\limits_{k=n-1}^{\infty}\bbP_\rmr\big\{\rmS_j(\bX^k)>\lambda_2\big\}\\
&\le\sum_{j\in[M]}\sum_{k=n-1}^{\infty}\sum_{\bQ\in\calP_k(\calX)^M: \rmG_j(\bQ)>\lambda_2}\exp\Big\{-k\sum_{l\in[M]}D(Q_l||P_\rmN)\Big\}\\
&= \sum_{j\in[M]}\sum_{k=n-1}^{\infty}\sum_{\bQ\in\calP_k(\calX)^M: \rmG_j(\bQ)>\lambda_2}\exp\Big\{-k\Big(D(Q_j||P_\rmN)+\sum_{t\in\calM_j}D(Q_t||P_\rmN)\Big)\Big\}\\
&\le \sum_{j\in[M]}\sum_{k=n-1}^{\infty}\sum_{\substack{\bQ\in\calP_k(\calX)^M:\rmG_j(\bQ)>\lambda_2}}\exp\Big\{-k\Big(D(Q_j||P_\rmN)+(M-1)
D\Big(\frac{\sum_{t\in\calM_j}Q_t}{M-1}\Big\|P_\rmN\Big)+\lambda_2\Big\}\\
&\le M\sum_{k=n-1}^{\infty}\exp\Big\{-k\Big(\lambda_2-\frac{M|\calX|\log n}{n-1}\Big)\Big\}\label{decrease1}\\
&= M\frac{\exp\Big\{-(n-1)\Big(\lambda_2-\frac{M|\calX|\log n}{n-1}\Big)\Big\}}{1-\exp\Big\{-\Big(\lambda_2-\frac{M|\calX|\log n}{n-1}\Big)\Big\}}.\label{mostone:Er}
\end{align}

Thus, it follows that $\sum\limits_{k=n-1}^{\infty}\bbP_\rmr\{\tau> k\}\le 1$ when $n$ is sufficiently large.
The average stopping time under hypothesis $\rmH_\rmr$ can also be expressed as the following form:
\begin{align}\label{Etau:r}
\mathbb{E}_\rmr[\tau]=\sum\limits_{k=1}^{\infty}\bbP_\rmr\{\tau> k\}=n-1+\sum\limits_{k=n-1}^{\infty}\bbP_\rmr\{\tau> k\}.
\end{align}
Therefore, we have $\mathbb{E}_\rmr[\tau]\le n$ when $n$ is sufficiently large for any positive $\lambda_2\in\bbR_+$.

\subsubsection{Achievable Error Exponents}

Fix $i\in[M]$. The misclassification probability can be upper bounded as follows:
\begin{align}
\psi_i(\Phi|P_\rmN,P_\rmA)
&=\bbP_i\{\phi(\bX^\tau)\neq\{\rmH_i,\mathrm{H_r}\}\}\\
&\le\sum_{k=n-1}^{\infty}\bbP_i\Big\{\exists~j\in\calM_i:~\rmS_j(\bX^k)\le\lambda_2~\mathrm{and}~\min_{t\in\calM_j}\rmS_t(\bX^k)>\lambda_1\Big\}\\
&\le\sum_{k=n-1}^{\infty}\bbP_i\big\{\exists~j\in\calM_i,~\forall~t\in\calM_j:~\rmS_t(\bX^k)>\lambda_1\big\}\\
&\le(M-1)\sum_{k=n-1}^{\infty}\bbP_i\big\{\rmS_i(\bX^k)>\lambda_1\big\}\\
&\le(M-1)\frac{\exp\Big\{-(n-1)\Big(\lambda_1-\frac{M|\calX|\log n}{n-1}\Big)\Big\}}{1-\exp\Big\{-\Big(\lambda_1-\frac{M|\calX|\log n}{n-1}\Big)\Big\}}\label{mostone:mis},
\end{align}
where \eqref{mostone:mis} follows from the steps analogously to those leading to the result in \eqref{one:taugek1}.
As $n\to\infty$, $\frac{M|\calX|\log n}{n-1}\to 0$ and thus, \eqref{mostone:mis} tends to $(M-1)\frac{\exp\{-(n-1)\lambda_1\}}{1-\exp\{-\lambda_1\}}$. Therefore, we have the misclassification exponent satisfies
\begin{align}
E_i(\Phi|P_\rmN,P_\rmA)\ge\lambda_1.
\end{align}

Furthermore, the false reject probability can be upper bounded as follows:
\begin{align}
\zeta_i(\Phi|P_\rmN,P_\rmA)&=\bbP_i\{\phi(\bX^\tau)=\mathrm{H_r}\}\\
&\le\bbP_i\Big\{\forall~j\in[M],~\rmS_j(\bX^\tau)\le\lambda_2\Big\}\\
&\le\bbP_i\Big\{\forall~j\in\calM_i,~\rmS_j(\bX^\tau)\le\lambda_2\Big\}\\
&\le\frac{\exp\Big\{-(n-1)\Big(\Omega_i(\lambda_2,P_\rmN,P_\rmA)-\frac{M|\calX|\log n}{n-1}\Big)\Big\}}{1-\exp\Big\{-\Big(\Omega_i(\lambda_2,P_\rmN,P_\rmA)-\frac{M|\calX|\log n}{n-1}\Big)\Big\}}\label{mostone:fr},
\end{align}
where \eqref{mostone:fr} follows from the steps analogously to those leading to the result in \eqref{sum_k} and the definition of $\Omega_i(\lambda,P_\rmN,P_\rmA)$ in \eqref{Omegai}.
As $n\to\infty$, $\frac{M|\calX|\log n}{n-1}\to 0$ and thus, \eqref{mostone:fr} tends to $\frac{\exp\{-(n-1)\Omega_i(\lambda_2,P_\rmN,P_\rmA)\}}{1-\exp\{-\Omega_i(\lambda_2,P_\rmN,P_\rmA)\}}$. Therefore, we have false reject exponent satisfies
\begin{align}
E_{\mathrm{fr},i}(\Phi|P_\rmN,P_\rmA)\ge\Omega_i(\lambda_2,P_\rmN,P_\rmA).
\end{align}

Finally, the false alarm probability can be upper bounded as follows:
\begin{align}
\mathrm{P_{fa}}(\Phi|P_\rmN,P_\rmA)&=\bbP_\rmr\{\phi(\bX^\tau)\neq\mathrm{H_r}\}\\
&\le\sum_{k=n-1}^{\infty}\bbP_\rmr\{\exists~j\in[M],~\rmS_j(\bX^k)>\lambda_1\}\\
&\le\sum_{k=n-1}^{\infty}\sum_{j\in[M]}\bbP_\rmr\{\rmS_j(\bX^k)>\lambda_1\}\\
&=M\frac{\exp\Big\{-(n-1)\Big(\lambda_1-\frac{M|\calX|\log n}{n-1}\Big)\Big\}}{1-\exp\Big\{-\Big(\lambda_1-\frac{M|\calX|\log n}{n-1}\Big)\Big\}}\label{mostone:fa},
\end{align}
where \eqref{mostone:fa} follows from the steps analogously to those leading to the result in \eqref{mostone:Er}.
As $n\to\infty$, $\frac{M|\calX|\log n}{n-1}\to 0$ and thus, \eqref{mostone:fa} tends to $M\frac{\exp\{-(n-1)\lambda_1\}}{1-\exp\{-\lambda_1\}}$. Therefore, we have the false alarm exponent satisfies
\begin{align}
E_{\mathrm{fa}}(\Phi|P_\rmN,P_\rmA)\ge\lambda_1.
\end{align}

\subsection{Justification of \eqref{compare:exactT}}\label{proof:Tcompare}
Fix any set $\calB\in\calS(T)$. Given any pair of distributions $(P_\rmN,P_\rmA)\in\calP(\calX)^2$, recall the error exponent of the fixed-length test as follows
\begin{align}
E_\calB(\Phi_{\rm LNV}|P_\rmN,P_\rmA)&=\min_{\calC\in\calS_\calB(T)}\min_{\substack{\bQ\in\calP(\calX)^M:\\\rmG_{\mathrm{Li},\calB}(\bQ)\ge\rmG_{\mathrm{Li},\calC}(\bQ)}}
\sum\limits_{j\in\calM_\calB}D(Q_j||P_\rmN)+\sum\limits_{i\in\calB}D(Q_i||P_\rmA).
\end{align}
Given any pair of distributions $(P,Q)\in\calP(\calX)^2$, for any set $\calC\in\calS_\calB(T)$, recall the tuple of distributions $\bQ=(Q_1,\ldots,Q_M)$ where $Q_i=P$ for all $i\in\calC$ and $Q_j=Q$ for all $j\in\calM_\calC$ satisfies the constraint $\rmG_{\mathrm{Li},\calB}(\bQ)\ge\rmG_{\mathrm{Li},\calC}(\bQ)$. Thus, we have
\begin{align}
\nn&E_\calB(\Phi_{\rm LNV}|P_\rmN,P_\rmA)\\*
&\le\min_{\calC\in\calS_\calB(T)}\min_{(P,Q)\in\calP(\calX)^2}|\calB\cap\calC|D(P||P_\rmA)+|\calB\cap\calM_\calC| D(Q||P_\rmA)+(M-|\calB\cup\calC|)D(Q||P_\rmN)+|\calC\cap\calM_\calB|D(P||P_\rmN)\\
&=\mathrm{LD}_\calB(P_\rmN,P_\rmA,M,T)\label{compare}\\
&\le E_\calB(\Phi_\mathrm{seq}|P_\rmN,P_\rmA),\label{leE}
\end{align}
where \eqref{compare} follows from the steps in \eqref{LDB}-\eqref{defLD} and \eqref{leE} follows from the result in Theorem \ref{seq_T_time}.

%
%
%

\subsection{Proof of Theorem \ref{at_most_T}}\label{proof:at_most_T}
\subsubsection{Expected Stopping Time Universality Constraint}
We first prove our sequential test in Sec. \ref{test_most_T} satisfies the expected stopping time universality constraint with $n$. Fix any $\mu\in[T]$ and $\calB\in\calS(\mu)$. The average stopping time under hypothesis $\rmH_\calB$ can be expressed as the following form:
\begin{align}\label{Etau:T}
\mathbb{E}_\calB[\tau]=\sum\limits_{k=1}^{\infty}\bbP_\calB\{\tau> k\}=n-1+\sum\limits_{k=n-1}^{\infty}\bbP_\calB\{\tau> k\}.
\end{align}

We next upper bound the term $\sum\limits_{k=n-1}^{\infty}\bbP_\calB\{\tau> k\}$. For simplicity, define $\calF=\{t\ge n-1:\exists~\calC\in\calS~\mathrm{s.t.}~\rmS_\calC(\bx^t)\le\lambda_2,~\mathrm{and}~\min_{\calD\in\calS_\calC}\rmS_\calC(\bx^t)>\lambda_1\}$ and $\calR=\{t\ge n-1:\forall~\calC\in\calS~\mathrm{s.t.}~\rmS_\calC(\bx^t)\le\lambda_2\}$. As discussed below \eqref{Etau:i}, given $k\ge n-1$, it follows from the definition of stopping time in \eqref{mostT_stoptime} that the event $\tau>k$ implies that $k\in\calF^c\cap\calR^c$. Thus, we have
\begin{align}
\sum\limits_{k=n-1}^{\infty}\bbP_\calB\{\tau>k\}
&\le\sum\limits_{k=n-1}^{\infty}\bbP_\calB\{k\in\calF^c\cap\calR^c\}\\
&\le\sum\limits_{k=n-1}^{\infty}\bbP_\calB\{k\in\calF^c\}\\
&\le\sum\limits_{k=n-1}^{\infty}\bbP_\calB\Big\{\forall~\calC\in\calS,~\rmS_\calC(\bX^k)>\lambda_2,~\mathrm{or}~\min_{\calD\in\calS_\calC}\rmS_\calC(\bX^k)\le\lambda_1\Big\}\\
&\le\sum\limits_{k=n-1}^{\infty}\bbP_\calB\big\{\forall~\calC\in\calS,~\rmS_\calC(\bX^k)>\lambda_2\big\}+
\sum\limits_{k=n-1}^{\infty}\bbP_\calB\Big\{\forall~\calC\in\calS,~\min_{\calD\in\calS_\calC}\rmS_\calD(\bX^k)\le\lambda_1\Big\}.\label{mostT:taugek}
\end{align}
The first term of \eqref{mostT:taugek} can be upper bounded as follows:
\begin{align}
\nn&\sum\limits_{k=n-1}^{\infty}\bbP_\calB\{\forall~\calC\in\calS:\rmS_\calC(\bX^k)>\lambda_2\}\\
&\le\sum_{k=n-1}^{\infty}\bbP_\calB\{\rmS_\calB(\bX^k)>\lambda_2\}\label{SBge}\\
&\le\sum_{k=n-1}^{\infty}\sum_{\substack{\bQ\in\calP_k(\calX)^M:\rmG_\calB(\bQ)>\lambda_2}}
\exp\Big\{-k\Big(\sum\limits_{j\in\calM_\calB}D(Q_j||P_\rmN)+\sum\limits_{i\in\calB}D(Q_i||P_\rmA)\Big)\Big\}\label{mostT:upperbound1}\\
&\le\sum_{k=n-1}^{\infty}\sum_{\substack{\bQ\in\calP_k(\calX)^M:\\\rmG_\calB(\bQ)>\lambda_2}}
\exp\bigg\{-k\bigg(s\cdot D\Big(\frac{\sum_{t\in\calB}Q_t}{s}\Big\|P_\rmA\Big)+(M-s)D\Big(\frac{\sum_{t\in\calM_\calB}Q_t}{M-s}\Big\|P_\rmN\Big)+\lambda_2\bigg)\bigg\}\label{mostT:sumGB}\\
&\le\sum_{k=n-1}^{\infty}\exp\Big\{-k\Big(\lambda_2-\frac{M|\calX|\log n}{n-1}\Big)\Big\}\label{decrease7}\\
&=\frac{\exp\Big\{-(n-1)\Big(\lambda_2-\frac{M|\calX|\log n}{n-1}\Big)\Big\}}{1-\exp\Big\{-\Big(\lambda_2-\frac{M|\calX|\log n}{n-1}\Big)\Big\}},\label{mostT:tauk1}
\end{align}
where \eqref{mostT:upperbound1} follows from the upper bound on the probability of the type class \cite[Theorem 11.1.4]{cover2012elements} and the definitions of $\rmG_\calB(\cdot)$ in \eqref{G_B} and $\rmS_\calB(\cdot)$ in Sec. \ref{test_most_T}, \eqref{mostT:tauk1} follows from the constraint $\rmG_\calB(\bQ)>\lambda_2$ and the following equations:
\begin{align}
\sum_{j\in\calM_\calB}D(Q_j||P_\rmN)&=(M-\mu)D\Big(\frac{\sum_{t\in\calM_\calB}Q_t}{M-\mu}\Big\|P_\rmN\Big)
+\sum_{j\in\calM_\calB}D\Big(Q_j\Big\|\frac{\sum_{t\in\calM_\calB}Q_t}{M-\mu}\Big),\label{T:sumPN}\\
\sum_{i\in\calB}D(Q_i||P_\rmA)&=\mu\cdot D\Big(\frac{\sum_{l\in\calB}Q_l}{\mu}\Big\|P_\rmA\Big)
+\sum_{i\in\calB}D\Big(Q_i\Big\|\frac{\sum_{l\in\calB}Q_l}{\mu}\Big),
\end{align}
which follow from \eqref{equation1}, and \eqref{decrease7} follows from the similar manner to \eqref{decrease4}.
Using \eqref{mostT:tauk1}, we have
\begin{align}\label{mostT:tauk12}
\sum\limits_{k=n-1}^{\infty}\bbP_\calB\{\forall~\calC\in\calS,~\rmS_\calC(\bX^k)>\lambda_2\}\le\frac{1}{2},
\end{align}
when $n$ is sufficiently large.

Given any $\lambda\in\bbR_+$ and any pair of distributions $(P_\rmN,P_\rmA)\in\calP(\calX)^2$, define
\begin{align}
\omega_\calB(\lambda,P_\rmN,P_\rmA):=\min_{\calC\in\calS_\calB}\min_{\substack{\bQ\in\calP(\calX)^M: \\ \rmG_\calC(\bQ)\le\lambda}}\sum\limits_{j\in\calM_\calB}D(Q_j||P_\rmN)+\sum\limits_{i\in\calB}D(Q_i||P_\rmA),
\end{align}
which is non-increasing in $\lambda$. In particular, $\omega_\calB(\lambda,P_\rmN,P_\rmA)=0$ if $\lambda>\tilde\lambda_1(P_\rmN,P_\rmA)$ where $\tilde\lambda_1(P_\rmN,P_\rmA)$ is defined in \eqref{tillambda1}.
The second term of \eqref{mostT:taugek} can be upper bounded as follows:
\begin{align}
\sum\limits_{k=n-1}^{\infty}\bbP_\calB\Big\{\forall~\calC\in\calS,~\min_{\calD\in\calS_\calC}\rmS_\calD(\bX^k)\le\lambda_1\Big\}
&\le\sum\limits_{k=n-1}^{\infty}\bbP_\calB\Big\{\min_{\calD\in\calS_\calB}\rmS_\calD(\bX^k)\le\lambda_1\Big\}\\
&=(|\calS|-1)\frac{\exp\Big\{-(n-1)\Big(\omega_\calB(\lambda_1,P_\rmN,P_\rmA)-\frac{M|\calX|\log n}{n-1}\Big)\Big\}}{1-\exp\Big\{-\Big(\omega_\calB(\lambda_1,P_\rmN,P_\rmA)-\frac{M|\calX|\log n}{n-1}\Big)\Big\}},\label{mostT:tauk2}
\end{align}
where \eqref{mostT:tauk2} follows from the steps analogously to those leading to the result in \eqref{one:taugek2}. Using \eqref{mostT:tauk2}, we have
\begin{align}\label{mostT:tauk22}
\sum\limits_{k=n-1}^{\infty}\bbP_\calB\big\{\forall~\calC\in\calS~\mathrm{s.t.}~\min_{\calD\in\calS_\calC}\rmS_\calC(\bX^k)\le\lambda_1\big\}\le\frac{1}{2},
\end{align}
when $n$ is sufficiently large and $0<\lambda_1<\tilde\lambda_1(P_\rmN,P_\rmA)$.

Combining \eqref{mostT:taugek}, \eqref{mostT:tauk12} and \eqref{mostT:tauk22}, we have $\sum\limits_{k=n-1}^{\infty}\bbP_i\{\tau>k\}\le 1$. Therefore, it follows from \eqref{Etau:T} that $\mathbb{E}_\calB[\tau]\le n$ when $n$ is sufficiently large and $0<\lambda_1<\tilde\lambda_1(P_\rmN,P_\rmA)$.

Similarly, we have the average stopping time under hypothesis $\rmH_\rmr$ as follows:
\begin{align}\label{Ertau:T}
\mathbb{E}_\rmr[\tau]=n-1+\sum\limits_{k=n-1}^{\infty}\bbP_\rmr\{\tau> k\}.
\end{align}
As discussed below \eqref{Etau:T}, similarly to the results in \eqref{mostT:tauk1}, the term $\sum\limits_{k=n-1}^{\infty}\bbP_\rmr\{\tau> k\}$ can be upper bounded as follows:
\begin{align}
\nn&\sum\limits_{k=n-1}^{\infty}\bbP_\rmr\{\tau> k\}\\*
&\le\sum\limits_{k=n-1}^{\infty}\bbP_\rmr\{k\in\calF^c\cap\calR^c\}\\
&\le\sum\limits_{k=n-1}^{\infty}\bbP_\rmr\{k\in\calR^c\}\\
&\le\sum\limits_{k=n-1}^{\infty}\bbP_\rmr\{\exists~\calC\in\calS~\mathrm{s.t.}~\rmS_\calC(\bX^k)>\lambda_2\}\\
&\le\sum\limits_{\calC\in\calS}\sum\limits_{k=n-1}^{\infty}\bbP_\rmr\{\rmS_\calC(\bX^k)>\lambda_2\}\\
&\le\sum\limits_{\calC\in\calS}\sum\limits_{k=n-1}^{\infty}\sum_{\substack{\bQ\in\calP_k(\calX)^M:\\\rmG_\calC(\bQ)>\lambda_2}}
\exp\Big\{-k\Big(\sum\limits_{j\in[M]}D(Q_j||P_\rmN)\Big)\Big\}\\
&\le\sum\limits_{\calC\in\calS}\sum\limits_{k=n-1}^{\infty}\sum_{\substack{\bQ\in\calP_k(\calX)^M:\\\rmG_\calC(\bQ)>\lambda_2}}
\exp\bigg\{-k\bigg(|\calC|\cdot D\Big(\frac{\sum_{t\in\calC}Q_t}{|\calC|}\Big\|P_\rmN\Big)+(M-|\calC|)D\Big(\frac{\sum_{t\in\calM_\calC}Q_t}{M-|\calC|}\Big\|P_\rmN\Big)+\lambda_2\bigg)\bigg\}\label{mostT:sumG}\\
&\le|\calS|\sum\limits_{k=n-1}^{\infty}\exp\Big\{-k\Big(\lambda_2-\frac{M|\calX|\log n}{n-1}\Big)\Big\}\label{decrease5}\\
&=|\calS|\frac{\exp\Big\{-(n-1)\Big(\lambda_2-\frac{M|\calX|\log n}{n-1}\Big)\Big\}}{1-\exp\Big\{-\Big(\lambda_2-\frac{M|\calX|\log n}{n-1}\Big)\Big\}}\label{mostT:taur},
\end{align}
where \eqref{mostT:sumG} follows the following equations:
\begin{align}
\sum_{j\in\calM_\calC}D(Q_j||P_\rmN)&=(M-|\calC|)D\Big(\frac{\sum_{t\in\calM_\calC}Q_t}{M-|\calC|}\Big\|P_\rmN\Big)
+\sum_{j\in\calM_\calC}D\Big(Q_j\Big\|\frac{\sum_{t\in\calM_\calC}Q_t}{M-|\calC|}\Big),\\
\sum_{i\in\calC}D(Q_i||P_\rmN)&=|\calC|\cdot D\Big(\frac{\sum_{l\in\calC}Q_l}{|\calC|}\Big\|P_\rmN\Big)
+\sum_{i\in\calC}D\Big(Q_i\Big\|\frac{\sum_{l\in\calC}Q_l}{|\calC|}\Big),
\end{align}
and \eqref{decrease5} follows from the steps analogously to those leading to the result in \eqref{decrease4}.
Thus, we have $\sum\limits_{k=n-1}^{\infty}\bbP_\rmr\{\tau> k\}\le 1$ when $n$ is sufficiently large. Therefore, it follows from \eqref{Ertau:T} that $\mathbb{E}_\rmr[\tau]\le n$ when $n$ is sufficiently large for any positive $\lambda_2\in\bbR_+$.

\subsubsection{Achievable Error Exponents}

Fix any $\calB\in\calS$. Recall the definition of $\calS_\calB=\{\calC\in\calS:\calC\neq\calB\}$.
The misclassification probability can be upper bounded as follows:
\begin{align}
\nn&\psi_\calB(\Phi|P_\rmN,P_\rmA)\\*
&=\bbP_\calB\{\Phi(\bX^\tau)\neq\{\rmH_\calB,\mathrm{H_r}\}\}\\
&\le\bbP_\calB\{\exists~\calC\in\calS_\calB:\rmS_\calC(\bX^\tau)\le\lambda_2,\;\mathrm{and}\;\min_{\calD\in\calS_\calC}\rmS_\calD(\bX^\tau)>\lambda_1\}\label{twoterms}\\
&\le\sum_{k=n-1}^{\infty}\bbP_\calB\{\exists~\calC\in\calS_\calB,~\forall~\calD\in\calS_\calC:\rmS_\calD(\bX^k)>\lambda_1\}\\
&\le|\calS_\calB|\sum_{k=n-1}^{\infty}\bbP_\calB\{\rmS_\calB(\bX^k)>\lambda_1\}\\
&=|\calS_\calB|\frac{\exp\Big\{-(n-1)\Big(\lambda_1-\frac{M|\calX|\log n}{n-1}\Big)\Big\}}{1-\exp\Big\{-\Big(\lambda_1-\frac{M|\calX|\log n}{n-1}\Big)\Big\}},\label{mostT:first}
\end{align}
where \eqref{mostT:first} follows from the steps analogously to those leading to the result in \eqref{mostT:tauk1}.
\begin{align}
E_\calB(\Phi|P_\rmN,P_\rmA)\ge\lambda_1
\end{align}

The false reject probability can be upper bounded as follows:
\begin{align}
\zeta_\calB(\Phi|P_\rmN,P_\rmA)
&=\bbP_\calB\{\Phi(\bX^\tau)=\mathrm{H_r}\}\\
&=\sum_{k=n-1}^{\infty}\bbP_\calB\big\{\forall~\calC\in\calS,~\rmS_\calC(\bX^k)\le\lambda_2\big\}\\
&\le\sum_{k=n-1}^{\infty}\bbP_\calB\{\forall~\calC\in\calS_\calB(\mu),~\rmS_\calC(\bX^k)\le\lambda_2\}\\
&\le\sum_{k=n-1}^{\infty}\max_{\calC\in\calS_\calB(\mu)}\bbP_\calB\{\rmS_\calC(\bX^k)\le\lambda_2\}\\
&\le\frac{\exp\Big\{-(n-1)\Big(\Omega_\calB(\lambda_2,\mu,P_\rmN,P_\rmA)+\frac{M|\calX|\log n}{n-1}\Big)\Big\}}{1-\exp\Big\{-\Big(\Omega_\calB(\lambda_2,\mu,P_\rmN,P_\rmA)+\frac{M|\calX|\log n}{n-1}\Big)\Big\}},\label{mostT:fr}
\end{align}
where \eqref{mostT:fr} follows from the steps analogously to those leading to the result in \eqref{mostone:fr} and the definition of $\Omega_\calB(\lambda_2,\mu,P_\rmN,P_\rmA)$ in \eqref{OmegaB}.
Since $\frac{M|\calX|\log n}{n-1}\to 0$ as $n\to\infty$, we obtain the false reject exponent satisfies
\begin{align}
E_{\mathrm{fr},\calB}(\Phi|P_\rmN,P_\rmA)\ge\Omega_\calB(\lambda_2,\mu,P_\rmN,P_\rmA).
\end{align}

The false alarm probability can be upper bounded as follows:
\begin{align}
\mathrm{P_{fa}}(\Phi|P_\rmN,P_\rmA)&=\bbP_\rmr\{\Phi(\bX^\tau)\neq\mathrm{H_r}\}\\
&\le\bbP_\rmr\{\exists~\calC\in\calS,~\rmS_\calC(\bX^\tau)>\lambda_1\}\\
&\le\sum_{k=n-1}^{\infty}\sum\limits_{\calC\in\calS}\bbP_\rmr\{\rmS_\calC(\bX^k)>\lambda_1\}\\
&\le|\calS|\frac{\exp\Big\{-(n-1)\Big(\lambda_1+\frac{M|\calX|\log n}{n-1}\Big)\Big\}}{1-\exp\Big\{-\Big(\lambda_1+\frac{M|\calX|\log n}{n-1}\Big)\Big\}},\label{mostT:fa}
\end{align}
where \eqref{mostT:fa} follows from the steps analogously to those leading to the result in \eqref{mostT:taur}.
Since $\frac{M|\calX|\log n}{n-1}\to 0$ as $n\to\infty$, we obtain the false alarm exponent satisfies
\begin{align}
E_\mathrm{fa}(\Phi|P_\rmN,P_\rmA)\ge\lambda_1.
\end{align}

\bibliographystyle{IEEEtran}
\bibliography{IEEEfull_paper}

\begin{thebibliography}{10}
\providecommand{\url}[1]{#1}
\csname url@samestyle\endcsname
\providecommand{\newblock}{\relax}
\providecommand{\bibinfo}[2]{#2}
\providecommand{\BIBentrySTDinterwordspacing}{\spaceskip=0pt\relax}
\providecommand{\BIBentryALTinterwordstretchfactor}{4}
\providecommand{\BIBentryALTinterwordspacing}{\spaceskip=\fontdimen2\font plus
\BIBentryALTinterwordstretchfactor\fontdimen3\font minus
  \fontdimen4\font\relax}
\providecommand{\BIBforeignlanguage}[2]{{%
\expandafter\ifx\csname l@#1\endcsname\relax
\typeout{** WARNING: IEEEtran.bst: No hyphenation pattern has been}%
\typeout{** loaded for the language `#1'. Using the pattern for}%
\typeout{** the default language instead.}%
\else
\language=\csname l@#1\endcsname
\fi
#2}}
\providecommand{\BIBdecl}{\relax}
\BIBdecl

\bibitem{diao2024sequentialOHT}
J.~Diao and L.~Zhou, ``Sequential outlier hypothesis testing under universality
  constraints,'' in \emph{IEEE ITW}, 2024, pp. 378--383.

\bibitem{li2014}
Y.~Li, S.~Nitinawarat, and V.~V. Veeravalli, ``Universal outlier hypothesis
  testing,'' \emph{IEEE Trans. Inf. Theory}, vol.~60, no.~7, pp. 4066--4082,
  2014.

\bibitem{li2017universal}
------, ``Universal sequential outlier hypothesis testing,'' \emph{Seq. Anal.},
  vol.~36, no.~3, pp. 309--344, 2017.

\bibitem{zhou2022second}
L.~Zhou, Y.~Wei, and A.~O. Hero, ``Second-order asymptotically optimal outlier
  hypothesis testing,'' \emph{IEEE Trans. Inf. Theory}, vol.~68, no.~6, pp.
  3585--3607, 2022.

\bibitem{bu2019linear}
Y.~Bu, S.~Zou, and V.~V. Veeravalli, ``Linear-complexity
  exponentially-consistent tests for universal outlying sequence detection,''
  \emph{IEEE Trans. Signal Process.}, vol.~67, no.~8, pp. 2115--2128, 2019.

\bibitem{gutman1989asymptotically}
M.~Gutman, ``Asymptotically optimal classification for multiple tests with
  empirically observed statistics,'' \emph{IEEE Trans. Inf. Theory}, vol.~35,
  no.~2, pp. 401--408, 1989.

\bibitem{zhou2020second}
L.~Zhou, V.~Y.~F. Tan, and M.~Motani, ``Second-order asymptotically optimal
  statistical classification,'' \emph{Information and Inference: A Journal of
  the IMA}, vol.~9, no.~1, pp. 81--111, 2020.

\bibitem{Ihwang2022sequential}
C.~Y. Hsu, C.~F. Li, and I.~H. Wang, ``On universal sequential classification
  from sequentially observed empirical statistics,'' in \emph{IEEE ITW}, 2022,
  pp. 642--647.

\bibitem{csiszar1998mt}
I.~Csiszar, ``The method of types [information theory],'' \emph{IEEE Trans.
  Inf. Theory}, vol.~44, no.~6, pp. 2505--2523, 1998.

\bibitem{zou2017nonparametric}
S.~Zou, Y.~Liang, H.~V. Poor, and X.~Shi, ``Nonparametric detection of
  anomalous data streams,'' \emph{IEEE Trans. Signal Process.}, vol.~65,
  no.~21, pp. 5785--5797, 2017.

\bibitem{zhu2024exponentially}
L.~Zhu and L.~Zhou, ``Exponentially consistent outlier hypothesis testing for
  continuous sequences,'' \emph{IEEE Trans. Inf. Theory}, 2025.

\bibitem{mahdi2021sequential}
M.~Haghifam, V.~Y.~F. Tan, and A.~Khisti, ``Sequential classification with
  empirically observed statistics,'' \emph{IEEE Trans. Inf. Theory}, vol.~67,
  no.~5, pp. 3095--3113, 2021.

\bibitem{diao2023classification}
J.~Diao, L.~Zhou, and L.~Bai, ``Achievable error exponents for almost
  fixed-length {$M$}-ary classification,'' in \emph{IEEE ISIT}, 2023, pp.
  1568--1573.

\bibitem{zhou2023achievable}
L.~Zhou, J.~Diao, and L.~Bai, ``Achievable error exponents for two-phase
  multiple classification,'' \emph{arXiv:2210.12736}, 2023.

\bibitem{van2014renyi}
T.~Van~Erven and P.~Harremos, ``R{\'e}nyi divergence and {K}ullback-{L}eibler
  divergence,'' \emph{IEEE Trans. Inf. Theory}, vol.~60, no.~7, pp. 3797--3820,
  2014.

\bibitem{li2020second}
Y.~Li and V.~Y.~F. Tan, ``Second-order asymptotics of sequential hypothesis
  testing,'' \emph{IEEE Trans. Inf. Theory}, vol.~66, no.~11, pp. 7222--7230,
  2020.

\bibitem{gretton2012kernel}
A.~Gretton, K.~M. Borgwardt, M.~J. Rasch, B.~Sch{\"o}lkopf, and A.~Smola, ``A
  kernel two-sample test,'' \emph{J.Mach. Learn. Res.}, vol.~13, no.~1, pp.
  723--773, 2012.

\bibitem{xiong2011group}
L.~Xiong, B.~P{\'o}czos, and J.~Schneider, ``Group anomaly detection using
  flexible genre models,'' in \emph{Adv. Neural Inf. Process. Syst.}, vol.~24,
  2011.

\bibitem{bai2022achievable}
L.~Bai, J.~Diao, and L.~Zhou, ``Achievable error exponents for almost
  fixed-length binary classification,'' in \emph{IEEE ISIT}, 2022, pp.
  1336--1341.

\bibitem{diao2023achievable}
J.~Diao, L.~Zhou, and L.~Bai, ``Achievable error exponents for almost
  fixed-length {$M$}-ary hypothesis testing,'' in \emph{IEEE ICASSP}, 2023, pp.
  1--5.

\bibitem{hsu2020binary}
H.-W. Hsu and I.-H. Wang, ``On binary statistical classification from
  mismatched empirically observed statistics,'' in \emph{IEEE ISIT}, 2020, pp.
  2533--2538.

\bibitem{pan2022asymptotics}
J.~Pan, Y.~Li, and V.~Y. Tan, ``Asymptotics of sequential composite hypothesis
  testing under probabilistic constraints,'' \emph{IEEE Trans. Inf. Theory},
  vol.~68, no.~8, pp. 4998--5012, 2022.

\bibitem{boroumand2022mismatched}
P.~Boroumand and A.~G. i~F{\`a}bregas, ``Mismatched binary hypothesis testing:
  Error exponent sensitivity,'' \emph{IEEE Trans. Inf. Theory}, vol.~68,
  no.~10, pp. 6738--6761, 2022.

\bibitem{cover2012elements}
T.~M. Cover and J.~A. Thomas, \emph{Elements of information theory}.\hskip 1em
  plus 0.5em minus 0.4em\relax John Wiley \& Sons, 2012.

\bibitem{polyanskiy2014lecture}
Y.~Polyanskiy and Y.~Wu, ``Lecture notes on information theory,'' \emph{Lecture
  Notes for ECE563 (UIUC)}, vol.~6, no. 2012-2016, p.~7, 2014.

\bibitem{klenke2014optional}
A.~Klenke and A.~Klenke, ``Optional sampling theorems,'' \emph{Probability
  Theory: A Comprehensive Course}, pp. 205--215, 2014.

\end{thebibliography}

\end{document}